\documentclass[12pt]{article}
\usepackage{fullpage}
\usepackage{amsfonts}
\usepackage{amsmath}
\usepackage{amssymb}
\usepackage{mathtools}
\usepackage{graphicx}
\usepackage{caption}
\usepackage{subcaption}
\usepackage{float}
\usepackage{physics}
\usepackage{bbold}
\usepackage{tabularx}
\usepackage{multirow}
\usepackage{pbox}
\usepackage[utf8]{inputenc}
\usepackage[absolute]{textpos}
\usepackage[dvipsnames]{xcolor}
\usepackage[colorlinks=true,linkcolor=blue,citecolor=red,plainpages=false,pdfpagelabels]
{hyperref}
\usepackage{newtxtext,newtxmath}%

\allowdisplaybreaks
\providecommand{\U}[1]{\protect\rule{.1in}{.1in}}
\newcommand{\cmmnt}[1]{}
\newtheorem{theorem}{Theorem}

\newtheorem{corollary}{Corollary}

\newtheorem{definition}{Definition}

\newtheorem{lemma}{Lemma}

\newtheorem{proposition}{Proposition}
\newtheorem{remark}{Remark}

\newenvironment{proof}[1][Proof]{\noindent\textbf{#1.} }{\ \rule{0.5em}{0.5em}}
\allowdisplaybreaks

\begin{document}

\title{\vspace{-.3in}\textbf{Extendibility limits quantum-secured communication and key distillation}}

 \author{Vishal Singh\thanks{School of Applied and Engineering Physics, Cornell University, Ithaca, New York 14850, USA} \and 
 Mark M.~Wilde\thanks{School of Electrical and Computer Engineering, Cornell University, Ithaca, New York 14850, USA} 
 }

\date{ }
\maketitle

\begin{abstract}
    Secret-key distillation from quantum states and channels is a central task of interest in quantum information theory, as it facilitates private communication over a quantum network. Here, we study the task of secret-key distillation from bipartite states and point-to-point quantum channels using local operations and one-way classical communication (one-way LOCC). We employ the resource theory of unextendible entanglement to study the transformation of a bipartite state under one-way LOCC, and we obtain several efficiently computable upper bounds on the number of secret bits that can be distilled from a bipartite state using one-way LOCC channels; these findings apply not only in the one-shot setting but also in some restricted asymptotic settings. We extend our formalism to  private communication over a quantum channel assisted by forward classical communication. We obtain efficiently computable upper bounds on the one-shot forward-assisted private capacity of a channel, thus addressing a question in the theory of quantum-secured communication that has been open for some time now. Our formalism also provides upper bounds on the rate of private communication when using a  large number of channels in such a way that the error in the transmitted private data decreases exponentially with the number of channel uses. Moreover, our bounds can be computed using semidefinite programs, thus providing a computationally feasible method to understand the limits of private communication over a quantum network.
\end{abstract}

\tableofcontents

\section{Introduction}

\subsection{Motivation}

The existence of a quantum network facilitates the distribution of secret keys between distant parties~\cite{BB84,Ekert91}, which ensures secure  communication by means of the one-time pad protocol. However, realizing an ideal quantum network can be very expensive. This motivates an in-depth study of the number of secret bits that can be established with the available resources, which, in the context of the quantum internet, are partially entangled states and quantum channels.

Our ability to distill secret keys from a bipartite state or a quantum channel depends on the operations that we can perform. The three most common settings studied in any non-local resource distillation task are as follows: local operations, local operations with one-way classical communication, and local operations with two-way classical communication. Here we consider the task of secret-key distillation from bipartite states and point-to-point channels in the presence of local operations and one-way classical communication. In what follows, we first discuss our contributions on understanding secret-key distillation from states, and thereafter we discuss our related contributions for channels.  

\subsection{Secret-key distillation from states}

The task of distilling secret keys from a bipartite state using local operations and one-way classical communication, abbreviated as one-way LOCC, has been studied extensively in the past~\cite{DW05, RR12, KKGW21}. From an information-theoretic perspective, the main quantities of interest are the one-shot, one-way distillable key of a state and the asymptotic one-way distillable key of the state. The one-shot, one-way distillable key of a state is roughly defined as the maximum number of ``approximate'' secret bits that can be distilled from a state using a one-way LOCC channel with respect to a fixed error parameter, and the asymptotic one-way distillable key of the state is the maximum rate at which secret bits can be distilled from an arbitrarily large number of independent and identically distributed copies of the state when using one-way LOCC channels.

Computing the one-shot, one-way distillable key and the asymptotic one-way distillable key is a challenging task. Lower bounds on the one-way distillable key in the one-shot regime, as well as the asymptotic regime, have been found in previous works~\cite{DW05, RR12, KKGW21}. Upper bounds on the one-shot distillable key of a state when using two-way LOCC have been found in terms of the smooth-min relative entropy of entanglement~\cite{WTB17} and the squashed entanglement~\cite{Christandl06, CEHHOR07, CSW12, Wilde16}. Naturally, these quantities also bound the one-shot, one-way distillable key from above. However, computing the smooth-min relative entropy of entanglement of a state is related to the  NP-hard problem of optimizing over the set of separable states~\cite{Gur03,G10}, and the squashed entanglement is not even known to be computable in the Turing sense, due to it involving an optimization over a state having a system of unbounded size. Moreover, the aforementioned quantities bound the one-shot distillable key of a state, which is expected to be larger than the one-shot, one-way distillable key of the state in general, leaving room for significant improvement in the estimation of the latter quantity.

In this work, we invoke the framework of unextendibility to obtain upper bounds on the one-shot, one-way distillable key of a state, which can be computed by means of a semidefinite program. As such, to the best of our knowledge, ours is the first general upper bound on this quantity that is efficiently computable, in contrast to the smooth min-relative entropy of entanglement and the squashed entanglement. We also give an upper bound on the maximum rate at which secret bits can be distilled from an arbitrarily large number of i.i.d.~copies of a state when using one-way LOCC channels, provided that the error in distillation decreases exponentially with the number of copies of the resource state. 

\subsection{Private communication over channels}

The one-shot setting of private communication has been the subject of several studies~\cite{RR11, WTB17, Wilde17, RSW17, KKGW21}. In the context of private communication, we are interested in the maximum number of private bits that can be sent through a quantum channel when using some freely available operations, which can be local operations and classical communication, local operations with only forward-classical communication, or local operations only. The corresponding quantities are called the one-shot two-way-assisted private capacity, the one-shot forward-assisted private capacity, and the one-shot unassisted private capacity, respectively. In the presence of forward-classical assistance, the task of secret-key distillation is equivalent to the task of private communication, which allows us to immediately extend our understanding of secret-key distillation from channels to private communication.

Finding efficiently computable upper bounds on the one-shot private capacity of a channel has remained an unsolved problem since early works on private capacity~\cite{CWY04, DW05}. Several upper bounds on the one-shot, two-way-assisted private capacity have been obtained~\cite{TGW14, WTB17, QSW18}. However, none of them are known to be efficiently computable. Even in the asymptotic regime, computable upper bounds on the unassisted private capacity and two-way-assisted private capacity are known only for qubit channels~\cite{FF21}. 

Here we contribute to this growing body of knowledge by giving upper bounds on the one-shot forward-assisted private capacity of a channel, which can be computed efficiently using a semidefinite program. We also give a semidefinite computable upper bound on the maximum rate at which private bits can be transmitted through a quantum channel when the error in transmission is required to decay exponentially with the number of channel uses. 

\subsection{Methods used in this work}

The resource theory of unextendible entanglement developed in~\cite{WWW24} serves as the primary mathematical framework in our investigation. The set of free states in the resource theory of unextendible entanglement is a state-dependent set comprising of all symmetric extensions of the state in question. The set of two-extendible channels serves as the set of free operations, which was defined in~\cite{KDWW19, KDWW21}. All one-way LOCC channels are two-extendible channels, which makes the resource theory of unextendible entanglement useful for the analysis of private communication with one-way LOCC.

The unextendible entanglement of quantum channels was defined in~\cite{SW24_channels}, which is the primary mathematical framework that we use to investigate private communication through a quantum channel assisted by one-way LOCC. In this resource theory, the set of free channels is a channel-dependent set, which consists of channels that are symmetrically-\textit{compatible} with the channel in question, where compatible channels were defined in~\cite{HMZ16}. The set of free operations are two-extendible superchannels, which form a semidefinite relaxation of the set of one-way LOCC superchannels originally considered in~\cite{LM15, RBL18}.

In the past, the unextendible entanglement of states has been used to study the \textit{exact} and \textit{probabilistic} distillation of secret keys from states using one-way LOCC channels in~\cite{WWW24,SW24}, and the unextendible entanglement of channels has been used to study zero-error private communication through channels in~\cite{SW24_channels}. By using the resource theory of unextendible entanglement to study ``approximate'' secret-key distillation from states and channels, we demonstrate that this resource theory can be used to study more practical settings in which an arbitrarily small error is allowed in resource distillation.

\subsection{Summary of results and organization of the paper}

The main contributions of our paper are as follows: We give upper bounds on the one-shot, one-way distillable key of a bipartite state and on the one-shot, forward-assisted private capacity of point-to-point quantum channels. As mentioned previously, to the best of our knowledge, these are the first efficiently computable upper bounds on these quantities. Extending our results to the asymptotic setting, we give upper bounds on the maximum rate of distilling secret keys from i.i.d.~copies of a bipartite state or channel when using one-way LOCC, albeit in a particular setting in which the error in distillation is required to decay exponentially with the number of copies of the resource. Several of our bounds can be computed using semidefinite programs, adding to their practical relevance. Finally, with this work, we demonstrate the power of the resource theory of unextendible entanglement in studying resource distillation. Prior to our work here, it was unclear how to apply this concept to the  setting of approximate key distillation and left as an open question since~\cite{WWW24}.

\begin{table}[t]
    \begin{center}
        \begin{tabular}{|p{5.9cm}|p{5.8cm}|p{3.2cm}|}
        \hline
        Setting & Divergence used for upper bound & Reference\\
        \hline\hline
        &&\\[-0.7em]
         One-shot setting & Smooth-min relative entropy & Theorem~\ref{theo:distill_key_st_ub_hypo_test} \\
        &&\\[-0.5em]
        Simplified bounds for 1-shot setting  & Smooth-min relative entropy & Corollaries~\ref{cor:dist_key_st_alg_ub} and~\ref{cor:dist_key_st_td_ub}\\
        &&\\[-0.5em]
        $n$-Shot setting & Sandwiched R\'enyi relative entropy & Corollary~\ref{cor:dist_key_sandwich_ub_n_copies}\\
        &&\\[-0.5em]
        Asymptotic setting & Umegaki relative entropy & Theorem~\ref{theo:dist_key_asymptotic_bnd}\\[0.5em]
        \hline
    \end{tabular}
    \end{center}
    \caption{A list of our results for secret-key distillation from a bipartite state when using one-way LOCC channels, in the one-shot and asymptotic settings.}
    \label{tab:st_results}
\end{table}

\begin{table}[t]
    \begin{center}
        \begin{tabular}{|p{5.9cm}|p{6.3cm}|p{2.0cm}|}
        \hline
        Setting &\centering Divergence used for upper bound & Reference\\
        \hline\hline
        &&\\[-0.7em]
         One-shot setting & Smooth-min relative entropy & Theorem~\ref{theo:dist_key_ch_hypo_test_bnd} \\
        &&\\[-0.5em]
        Simplified bounds for 1-shot setting  & Smooth-min relative entropy & Corollary~\ref{cor:smooth_min_simplified_bnd_channels}\\
        &&\\[-0.5em]
        $n$-shot setting & Geometric R\'enyi relative entropy & Corollary~\ref{cor:dist_key_ch_geo_bnd}\\
        &&\\[-0.5em]
        Asymptotic setting & Belavkin--Staszewski relative entropy & Theorem~\ref{theo:dist_key_asymptotic_bnd_channels}\\[0.5em]
        \hline
    \end{tabular}
    \end{center}
    \caption{A list of our results for forward-assisted private communication from point-to-point quantum channels in the one-shot and asymptotic settings.}
    \label{tab:ch_results}
\end{table}

In Table~\ref{tab:st_results}, we present a brief summary of our results on one-way secret-key distillation from bipartite states, and in Table~\ref{tab:ch_results}, we give a brief summary of our results on forward-assisted private communication over channels. We note here that the Python codes used for calculating the semidefinite programs in this paper are available with the arXiv posting.

An outline of our  paper is as follows:
\begin{itemize}
    \item Section~\ref{sec:notations}: Definitions and notations used in the paper, along with basic facts about quantum states, channels, and superchannels.
    
    \item Section~\ref{sec:key_dist_setup}: Discussion on secret-key distillation from bipartite states using one-way LOCC channels, and definition of the one-shot, one-way distillable key of a state, which is the primary quantity of interest.
    
    \item Section~\ref{sec:two_extendibility}: Review of the concepts of two-extendibility and the unextendible entanglement of states. Discussion on the unextendible entanglement of states induced by smooth-min relative entropy and $\alpha$-sandwiched R\'enyi relative entropy, which are the primary ingredients in the main result obtained for one-way secret-key distillation from bipartite states.
    
    \item Section~\ref{sec:distillable_key_results}: Main results on one-way secret-key distillation from an arbitrary bipartite state. Numerical demonstration of the upper bounds on the one-shot, one-way distillable key of isotropic states using semidefinite programs.
    \item Section~\ref{sec:priv_comm_results}: Discussion of private communication over quantum channels using one-way LOCC superchannels. Review of the unextendible entanglement of channels induced by smooth-min relative entropy and $\alpha$-geometric R\'enyi relative entropy. Main results on one-way private communication over an arbitrary channel. Demonstrating the upper bound on the one-shot, forward-assisted private capacity of the erasure channel using analytical expressions. 
\end{itemize}

\section{Notation and Preliminaries}\label{sec:notations}

In this section, we review background material on the three major elements that we use in the rest of the work: quantum states, channels, and superchannels. 

\subsection{Quantum states and channels}

A quantum state $\rho_A$ is a positive semidefinite, unit-trace operator acting on a Hilbert space $\mathcal{H}_A$. We denote the set of all linear operators acting on the Hilbert space $\mathcal{H}_A$ by $\mathcal{L}(A)$ and the set of all quantum states acting on this Hilbert space by $\mathcal{S}(A)$. A bipartite quantum state $\rho_{AB}$ acting on the Hilbert space $\mathcal{H}_{A}\otimes \mathcal{H}_B$ is called separable if it can be written as
\begin{equation}
	\rho_{AB} = \sum_{x\in \mathcal{X}} p(x)\sigma^x_{A}\otimes\tau^x_{B},
 \label{eq:def-sep-state}
\end{equation}
where $\{p(x)\}_{x\in \mathcal{X}}$ is a probability distribution and $\{\sigma^x_{A}\}_{x\in \mathcal{X}}$ and $\{\tau^x_{B}\}_{x\in \mathcal{X}}$ are sets of states. Any quantum state that is not separable is said to be entangled. The maximally entangled state vector in the Hilbert space $\mathcal{H}_{A}\otimes\mathcal{H}_B$ is denoted as follows:
\begin{equation}
	|\Phi^d\rangle_{AB} \coloneqq \frac{1}{\sqrt{d}}\sum_{i=0}^{d-1} |i\rangle_A|i\rangle_B,
\end{equation}
where $\{|i\rangle\}_{i=0}^{d-1}$ is an orthonormal basis and $d$ is the Schmidt rank of the state. We denote the correponding density operator as $\Phi^d_{AB} \equiv |\Phi^d\rangle\!\langle \Phi^d|_{AB}$. 

A quantum channel $\mathcal{N}_{A\to B}$ is a completely positive (CP) and trace-preserving (TP) linear map that takes an operator acting on the Hilbert space $\mathcal{H}_A$ as input and outputs an operator acting on the Hilbert space $\mathcal{H}_B$. We denote the Choi operator of a channel $\mathcal{N}_{A\to B}$ by  $\Gamma^{\mathcal{N}}_{RB}$, which is defined as follows:
\begin{equation}\label{eq:Choi_op_defn}
	\Gamma^{\mathcal{N}}_{RB} \coloneqq \mathcal{N}_{A\to B}\!\left(d\Phi^d_{RA}\right),
\end{equation}
where $\Phi^d_{RA}$ is the maximally entangled state of Schmidt rank $d$ and system $R$ is isomorphic to system $A$. The normalized Choi operator is called the Choi state of the channel, and it is defined as follows:
\begin{equation}
    \Phi^{\mathcal{N}}_{RB} \coloneqq \mathcal{N}_{A\to B}\!\left(\Phi^d_{RB}\right).
\end{equation}

An important class of channels that is central to our work consists of one-way LOCC channels. We use the symbol $\mathcal{L}^{\to}$ for a one-way LOCC channel. 

\begin{figure}
    \centering
    \begin{subfigure}{0.45\textwidth}
        \includegraphics[width=0.9\linewidth]{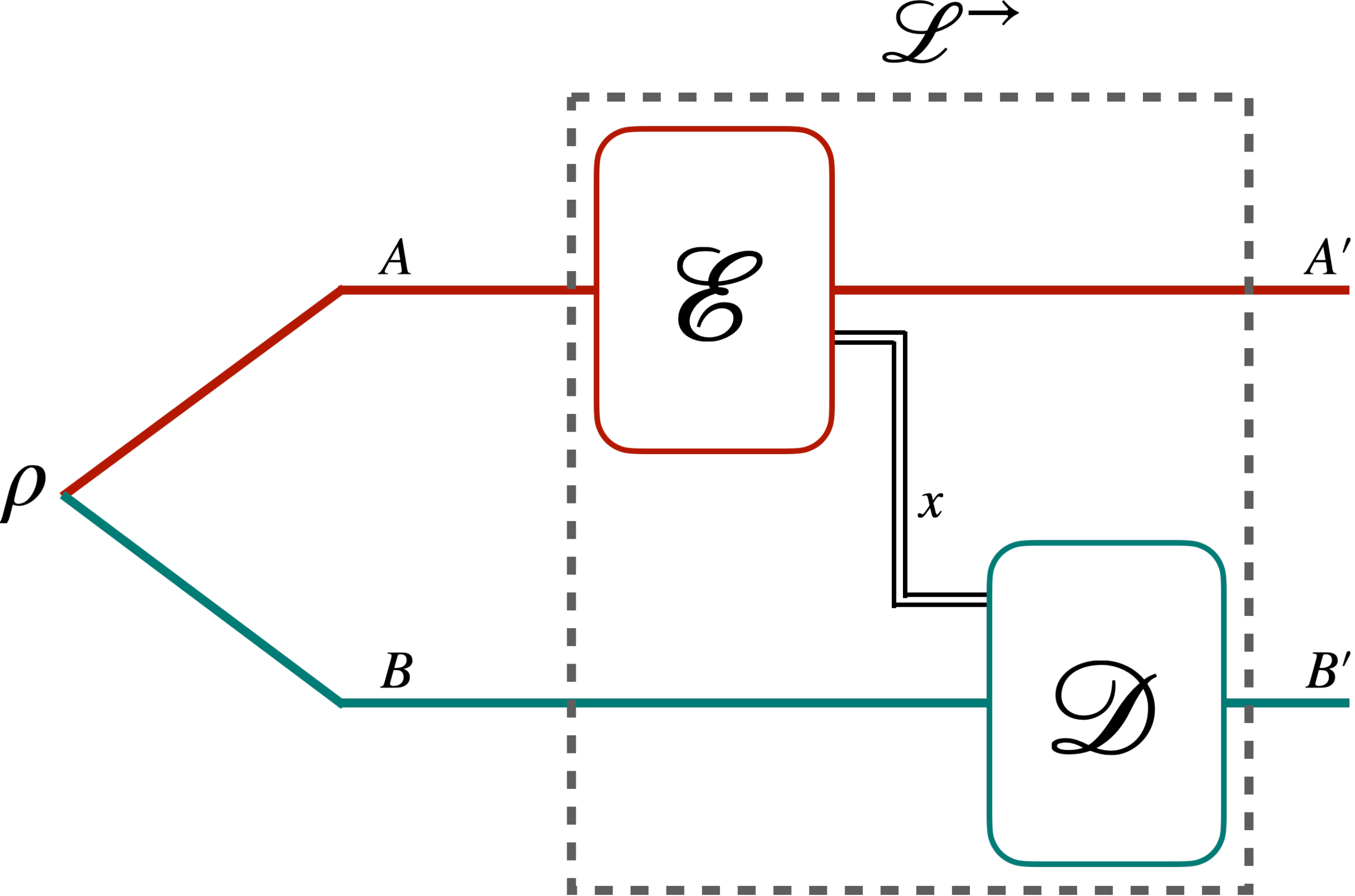}
        \caption{\centering\label{fig:one_WL_channel} One-way LOCC channel}
    \end{subfigure}
    \begin{subfigure}{0.45\textwidth}
        \includegraphics[width=0.9\linewidth]{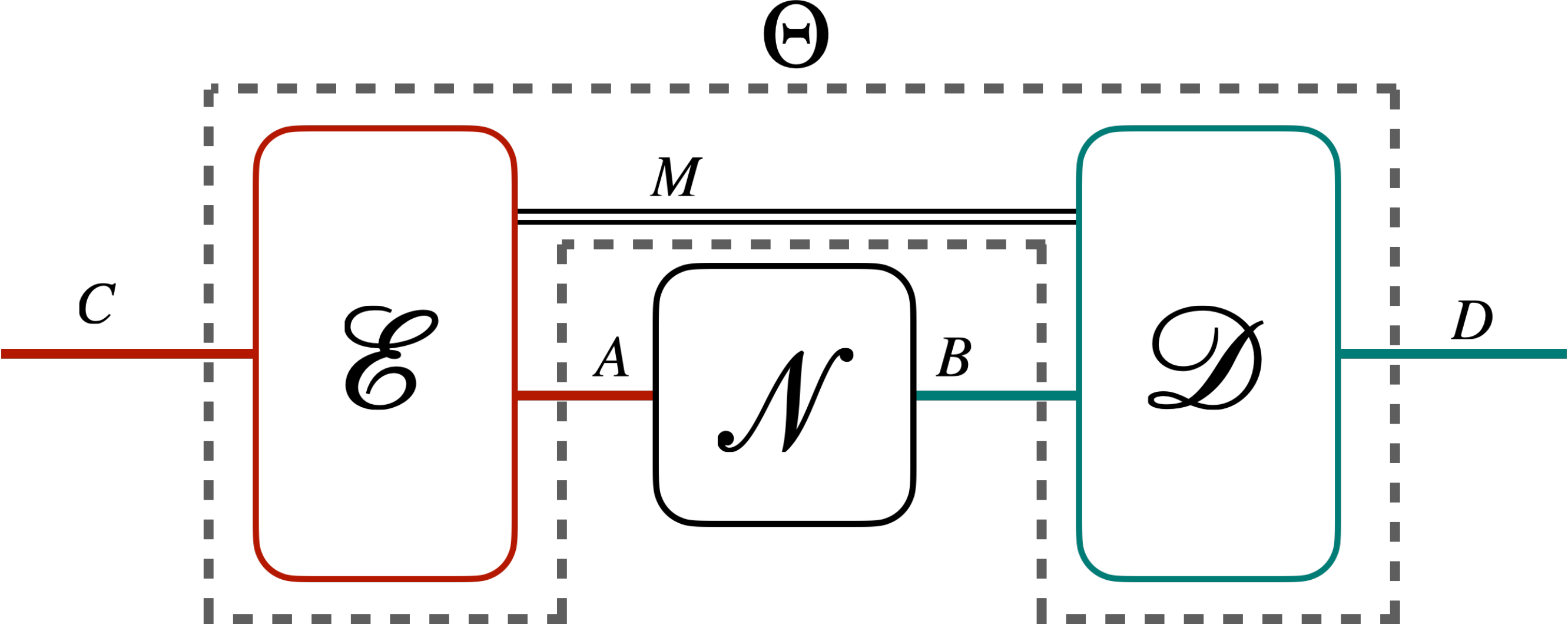}
        \caption{\centering\label{fig:one_WL_superchannel} One-way LOCC superchannel}
    \end{subfigure}
    \caption{(a) Schematic diagram of a one-way LOCC channel $\mathcal{L}^{\to}_{AB\to A'B'}$, as defined in~\eqref{eq:one_WL_ch_defn}, acting on a bipartite state $\rho_{AB}$. (b) Schematic diagram of a one-way LOCC superchannel $\Theta_{(A\to B)\to (C\to D)}$, defined in~\eqref{eq:superchannel_fund_theo} with $M$ being a classical system, acting on a channel $\mathcal{N}_{A\to B}$.}
\end{figure}

A one-way LOCC channel is a quantum channel that acts on a bipartite state, and it can be physically described by the following sequence of operations: Say Alice and Bob share a bipartite state $\rho_{AB}$. Alice applies a quantum instrument $\left\{\mathcal{E}^x_{A\to A'}\right\}_{x\in \mathcal{X}}$ on her system, where $x$ is a classical label corresponding to the outcome of the instrument. She sends the classical label $x$ to Bob through an ideal classical channel. Bob then applies a quantum channel $\mathcal{D}^x_{B\to B'}$ based on the label $x$ that he received from Alice (see Figure~\ref{fig:one_WL_channel}). A one-way LOCC channel can be mathematically described as follows:
\begin{equation}\label{eq:one_WL_ch_defn}
    \mathcal{L}^{\to}_{AB\to A'B'} = \sum_{x\in \mathcal{X}}\mathcal{D}^x_{A\to B}\otimes \mathcal{E}^x_{A\to A'},
\end{equation}
where $\left\{\mathcal{E}^x_{A\to A'}\right\}_{x \in \mathcal{X}}$ is a quantum instrument and $\left\{\mathcal{D}^x_{B\to B'}\right\}_{x\in \mathcal{X}}$ is a set of quantum channels.

\subsection{Quantum superchannels}\label{sec:superchannels}

A quantum superchannel $\Theta_{(A\to B)\to (C\to D)}$ is a linear map that transforms a quantum channel to another quantum channel. Since quantum channels are completely positive and trace-preserving maps, a superchannel is a completely CPTP-preserving map (see Definition~\ref{def:superch} for a formal definition). It can be perceived as a mathematical model for any physical transformation a quantum channel can undergo, as long as the resulting map is also a quantum channel. Quantum superchannels were introduced in~\cite{Chiribella_2008} and further investigated in~\cite{Gour_2019}, both of which provide a detailed discussion. Here, we include a brief discussion on superchannels relevant to this work.

\begin{definition}[Superchannel]
\label{def:superch}
Let $\mathcal{T}_{A\to B} : \mathcal{L}(A) \to \mathcal{L}(B)$ be a linear map. Let the space of all such maps be denoted by $\mathbb{L}^{AB}$. 
A linear map $\Theta_{(A\to B)\to (C\to D)}: \mathbb{L}^{AB} \to \mathbb{L}^{CD}$ is a superchannel if 
\begin{enumerate}
    \item It is completely CP preserving; i.e.,    $(\operatorname{id}_{(E)\to(E)}\otimes\Theta_{(A\to B)\to (C\to D)}) (\mathcal{T}_{EA \to EB})$ is a CP map if $\mathcal{T}_{EA \to E'B}$ is a CP map, for all possible dimensions of system E.
    
    \item It is TP preserving; i.e.,
        $\Theta_{(A\to B)\to (C\to D)}(\mathcal{T}_{A\to B})$  is a TP map if $\mathcal{T}_{A\to B}$ is a TP map.
\end{enumerate}
\end{definition}

According to the fundamental theorem of superchannels~\cite{Chiribella_2008}, every superchannel can be decomposed into a pre-processing channel $\mathcal{E}_{C\to MA}$ and a post-processing channel $\mathcal{D}_{MB\to D}$ connected by a memory system $M$. That is, for every superchannel $\Theta_{(A\to B)\to (C\to D)}$, there exist $\mathcal{E}_{C\to MA}$ and $\mathcal{D}_{MB\to D}$ such that
\begin{equation}\label{eq:superchannel_fund_theo}
    \Theta_{(A\to B)\to (C\to D)}(\mathcal{N}_{A\to B}) = \mathcal{D}_{MB\to D}\circ\mathcal{N}_{A\to B}\circ\mathcal{E}_{C\to MA}.
\end{equation}
Quantum superchannels are a powerful tool in analyzing communication tasks over a quantum channel, as any communication protocol can be modeled as a superchannel.

A special class of superchannels that is relevant to this work is the class of one-way LOCC superchannels. This is the set of superchannels that can be simulated by local operations and forward classical communication (see Figure~\ref{fig:one_WL_superchannel}). In particular, if system $M$ in~\eqref{eq:superchannel_fund_theo} is set to be a classical system, then every superchannel $\Theta_{(A\to B)\to (C\to D)}$ that has the form given in~\eqref{eq:superchannel_fund_theo} is a one-way LOCC superchannel.

\section{One-way secret-key distillation}\label{sec:key_dist_setup}

In principle, the existence of a quantum network ensures unconditional secret key distribution~\cite{BB84, Ekert91}. Alice and Bob can often manipulate a shared entangled state by means of local operations to obtain a maximally classically-correlated state that is completely independent of the system of any eavesdropper. The maximally-classically correlated state can then be used as a key to encrypt some classical data that Alice intends to send to Bob using the one-time-pad scheme. Since the eavesdropper is independent of the key shared between Alice and Bob, it is impossible for them to decode the encrypted data irrespective of their computational power. The state thus established between Alice, Bob, and an eavesdropper is called a tripartite key state, and it can be expressed in the following form:
\begin{equation}\label{eq:tripartite_key_defn}
    \tau_{ABE} = \frac{1}{K}\sum_{i=0}^{k-1} |i\rangle\!\langle i|_A\otimes |i\rangle\!\langle i|_B\otimes \sigma_E,
\end{equation}
where $\sigma_E$ is an arbitrary quantum state.

In general, the task of secret-key distillation is a three party problem due to the involvement of the eavesdropper. However, a crucial discovery was made in~\cite{HHHO05,HHHO09}, establishing an equivalence between the tripartite scenario and a bipartite scenario involving the concept of a \textit{private state}. In the next section, we briefly review the structure of bipartite private states, which plays an important role in this work.

\subsection{Bipartite private states}

A bipartite private state $\gamma^k_{ABA'B'}$ is the most general form of a quantum state that furnishes a secret key of $\log_2 k$ bits upon local measurements of systems $A$ and $B$. Therefore, to establish a secret key whose secrecy is ensured by the laws of quantum mechanics, one needs to establish a bipartite private state.

It was shown in~\cite{HHHO05, HHHO09} that a private state $\gamma^k_{ABA'B'}$ holding $\log_2 k$ secret key bits can always be written in the following form:
\begin{equation}\label{eq:priv_st_defn}
	\gamma^k_{ABA'B} = V_{ABA'B'}\left(\Phi^k_{AB}\otimes\tau_{A'B'}\right)V^{\dagger}_{ABA'B'},
\end{equation}
where $\Phi^k_{AB}$ is a maximally entangled state of Schmidt rank $k$, the operator $\tau_{A'B'}$ is an arbitrary bipartite state, and $V_{ABA'B'}$ is called a twisting unitary, defined as follows:
\begin{equation}\label{eq:twisting_unitary_defn}
	V_{ABA'B'} = \sum_{i=0}^{k-1} |i\rangle\!\langle i|_A\otimes I_B \otimes U^{i}_{A'B'},
\end{equation}
with $U^i_{A'B'}$ being some unitary operator. The private state in~\eqref{eq:priv_st_defn} can then be written more explicitly as follows:
\begin{equation}
    \gamma^k_{ABA'B} = \sum_{i,j=0}^{k-1} |i\rangle\!\langle j|_A\otimes |i\rangle\!\langle j|_B\otimes U^i_{A'B'}\tau_{A'B'}\left(U^j_{A'B'}\right)^{\dagger}.
\end{equation}
Systems $A$ and $B$ are said to be the key systems, and systems $A'$ and $B'$ are said to be the shield systems.

\subsection{One-shot, one-way distillable key of a state}

Let us now consider the task of distilling secret keys from a bipartite state shared between two parties using local operations and one-way classical communication. Since the distillation of a secret key is equivalent to the distillation of a bipartite private state, we consider the task of distilling a private state from bipartite resource state using one-way LOCC channels. However, distilling private states exactly, or even probabilistically, is a very restrictive task~\cite{SW24}, and one must relax this setting to allow for any practical distillation of secret keys.

The task of distilling approximate secret keys using one-way LOCC channels has been a subject of significant interest in several prior works~\cite{DW05, RR12, KKGW21} (see Figure~\ref{fig:key_distillation} for a schematic diagram). The error in distillation of secret keys from a state $\rho_{AB}$ using a one-way LOCC channel $\mathcal{L}^{\to}_{AB\to A'B'A''B''}$ is measured by the infidelity, defined as
\begin{equation}\label{eq:key_distillation_error}
    p_{\operatorname{err}}\!\left(\mathcal{L}^{\to};\rho_{AB}\right) \coloneqq \inf_{\gamma^k_{A'B'A''B''}}\left(1-F\!\left(\gamma^k_{A'B'A''B''},\mathcal{L}^{\to}_{AB\to A'B'A''B''}\!\left(\rho_{AB}\right)\right)\right),
\end{equation}
where the infimum is over all bipartite private states holding $\log_2 k$ secret bits and $F(\cdot,\cdot)$ denotes the fidelity between two states, which is defined as follows:
\begin{equation}
    F\!\left(\rho,\sigma\right) \coloneqq \left(\operatorname{Tr}\!\left[\sqrt{\sqrt{\sigma}\rho\sqrt{\sigma}}\right]\right)^2.
\end{equation}
The reason for choosing infidelity to be the metric of error is motivated from the \lq$\gamma^k$-privacy test\rq~\cite{HHHLO08_QP, HHHLO08} (see~\cite[Section~15.1.3]{KW24} for a detailed discussion.
The $\gamma^k$-privacy test is a POVM $\left\{\Pi^{\gamma}_{ABA'B'}, I_{ABA'B'}-\Pi^{\gamma}_{ABA'B'}\right\}$, where 
\begin{equation}\label{eq:priv_test_defn}
    \Pi^{\gamma}_{ABA'B'} \coloneqq V_{ABA'B'}\left(\Phi^k_{AB}\otimes I_{A'B'}\right)V^{\dagger}_{ABA'B'},
\end{equation}

\begin{figure}
    \centering
    \includegraphics[width=0.65\linewidth]{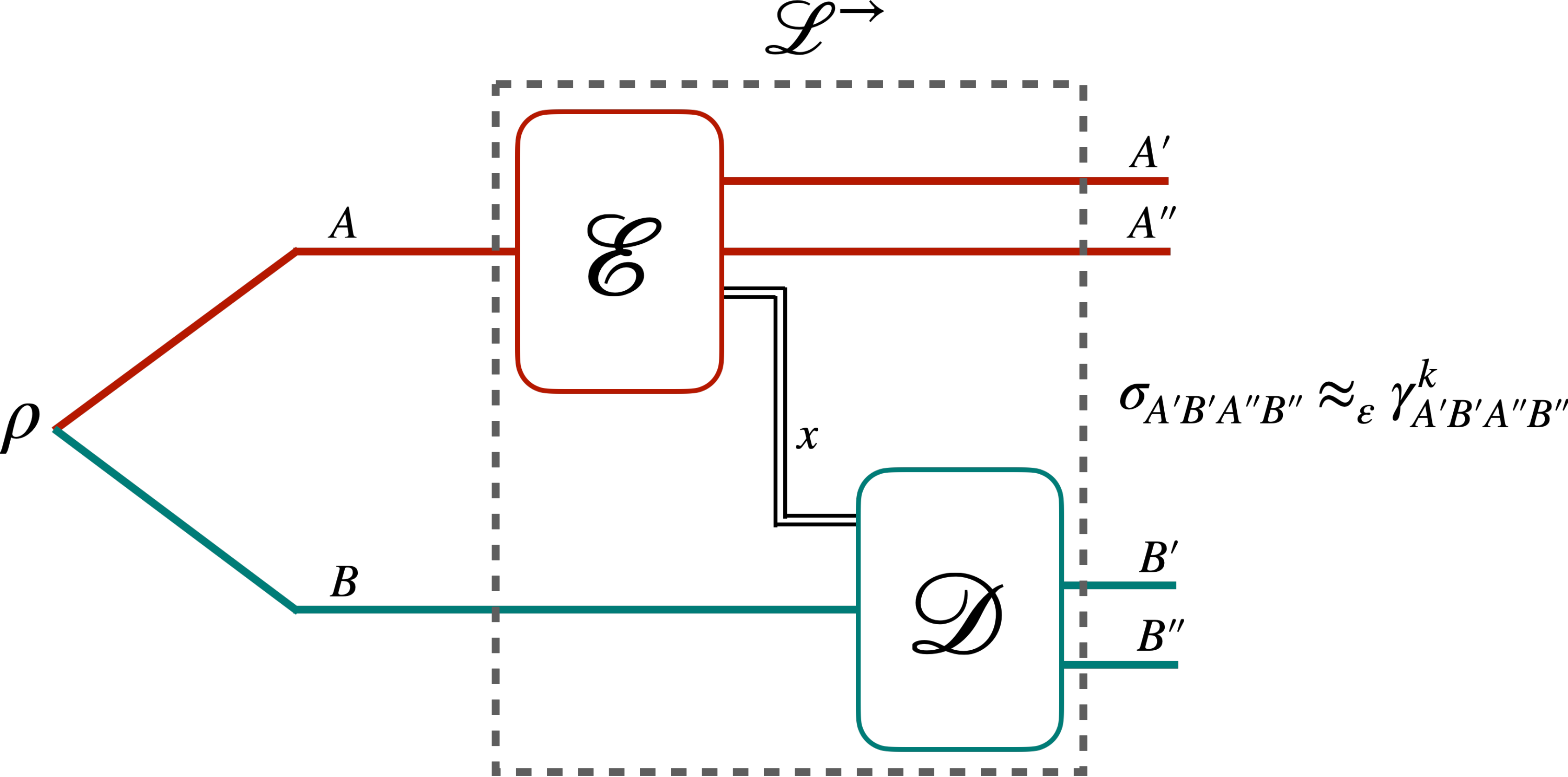}
    \caption{Schematic diagram of approximate distillation of a bipartite private state $\gamma^k_{A'B'A''B''}$ from a state $\rho_{AB}$ using a one-way LOCC channel $\mathcal{L}^{\to}_{AB\to A'B'A''B''}$, where the error in distillation, denoted by $\varepsilon$, is defined in~\eqref{eq:key_distillation_error}.}
    \label{fig:key_distillation}
\end{figure}

The one-shot, one-way distillable key of a state is the quantity that describes the number of bits of  secret key that can be established between two parties holding a resource state $\rho_{AB}$, with some error tolerance $\varepsilon$, when using one-way LOCC channels.
\begin{definition}[One-shot, one-way distillable key]
    For $\varepsilon\in[0,1]$, the one-shot, one-way distillable key of a state $\rho_{AB}$ is defined as follows:
\begin{equation}\label{eq:1W_dist_key_st_defn}
	K^{\varepsilon,\to}_D\!\left(\rho_{AB}\right) \coloneqq \sup_{\substack{k\in \mathbb{N}, \gamma^k_{A'B'A''B''},\\ \mathcal{L}^{\to}\in \operatorname{1WL}}}\left\{\begin{array}{c}
	\log_2 k:\\
	F\!\left(\mathcal{L}^{\to}_{AB\to A'B'A''B''}\!\left(\rho_{AB}\right),\gamma^k_{A'B'A''B''}\right) \ge 1-\varepsilon
	\end{array} \right\},
\end{equation}
where $\operatorname{1WL}$ stands for the set of all one-way LOCC channels.
\end{definition}

In the above definition, the supremum is over every positive integer $k$, every private state $\gamma^k_{A'B'A''B''}$ holding $\log_2 k$ secret key bits, and every one-way LOCC channel $\mathcal{L}^{\to}_{AB\to A'B'A''B''}$.

\section{Two-extendibility}\label{sec:two_extendibility}

In this section we review the concepts of two-extendibility for states and channels. The resource theory of $k$-extendibility was developed in~\cite{KDWW19, KDWW21} as a semidefinite relaxation of the resource theory of entanglement, and the resource theory of two-extendibility is a special case when $k = 2$. A state-dependent resource theory of extendibility was developed in~\cite{WWW24}, which is the framework that we employ to study the task of secret-key distillation with one-way LOCC channels.

\subsection{Two-extendible states and channels}

Let us first discuss two-extendible states~\cite{DPS04}, also known as \textit{symmetrically extendible states}~\cite{Wer89}, \textit{two-shareable states}~\cite{Yang06}, and \textit{anti-degradable states}~\cite{LDS18}. 

\begin{definition}[Two-extendible state]
    A bipartite state $\rho_{AB}$ is said to be two-extendible if there exists a state $\omega_{ABE}$ such that the following conditions hold:
    \begin{equation}\label{eq:2_ext_st_ext}
        \operatorname{Tr}_{E}\!\left[\omega_{ABE}\right] = \rho_{AB},
    \end{equation}
    and
    \begin{equation}\label{eq:2_ext_st_perm}
        W_{BE}\!\left(\omega_{ABE}\right)\!W^{\dagger}_{BE} = \omega_{ABE},
    \end{equation}
    where the unitary swap operator $W$ is defined as follows:
    \begin{equation}\label{eq:SWAP_defn}
        W_{BE} \coloneqq \sum_{k,k' = 0}^{d-1} |k\rangle\!\langle k'|_B\otimes |k'\rangle\!\langle k|_E.
    \end{equation}
\end{definition}
Note that the system $E$ should be isomorphic to the system $B$ for~\eqref{eq:2_ext_st_ext} and~\eqref{eq:2_ext_st_perm} to hold. The state $\omega_{ABE}$ is said to be a two-extension of the state $\rho_{AB}$ if the conditions in~\eqref{eq:2_ext_st_ext} and~\eqref{eq:2_ext_st_perm} are met. 

\begin{remark}
    All bipartite separable states are two-extendible. Consider an arbitrary bipartite separable state $\rho_{AB}\coloneqq \sum_{x\in \mathcal{X}} p(x)\sigma^x_A\otimes \tau^x_{B}$, where $\{p(x)\}_{x\in \mathcal{X}}$ is a probability distribution and $\left\{\sigma^x_A\right\}_{x\in \mathcal{X}}$ and $\left\{\tau^x_B\right\}_{x\in \mathcal{X}}$ are sets of quantum states. One can always construct the following two-extension of $\rho_{AB}$:
    \begin{equation}
        \omega_{ABE} \coloneqq \sum_{x\in \mathcal{X}} p(x)\sigma^x_A\otimes \tau^x_{B}\otimes \tau^x_{E},
    \end{equation}
    which shows that all bipartite separable states are two-extendible. However, all two-extendible states are not separable. A simple example is the following isotropic state~\cite{HH99}:
    \begin{equation}
        \zeta_{AB} = \frac{5}{8}\Phi_{AB} + \frac{1}{8}\left(I_{AB}-\Phi_{AB}\right) = \frac{1}{2}\Phi_{AB} + \frac{1}{8}I_{AB},
    \end{equation}
    where $A$ and $B$ are two-dimensional systems, $\Phi_{AB}$ is a two-qubit maximally entangled state, and $I_{AB}$ is the identity operator. This state is two-extendible with the following two-extension:
    \begin{equation}\label{eq:2_ext_st_example}
        \omega_{ABE} = \frac{1}{4}\Phi_{AB}\otimes I_{E} + \frac{1}{4}\Phi_{AE}\otimes I_{B},
    \end{equation}
    but $\zeta_{AB}$ has non-zero distillable entanglement~\cite{HH99}, and hence, it is not a separable state.
\end{remark}

A family of semidefinite relaxations of one-way LOCC channels was developed in~\cite{KDWW19,KDWW21}, called $k$-extendible channels. The set of $k$-extendible channels serves as the set of free channels in~\cite{KDWW19,KDWW21}. Setting $k=2$, we obtain the set of two-extendible channels, which serves as the set of free operations in the state-dependent resource theory of unextendibility developed in~\cite{WWW24}. Here we briefly discuss the idea of two-extendible channels.

\begin{definition}[Two-extendible channel]
    A bipartite channel $\mathcal{N}_{AB\to A'B'}$ is said to be two-extendible if there exists a channel $\mathcal{P}_{ABE\to A'B'E'}$ such that the following conditions hold:
    \begin{equation}\label{eq:2_ext_ch_ext}
        \operatorname{Tr}_{E'}\circ\mathcal{P}_{ABE\to A'B'E'} = \mathcal{N}_{AB\to A'B'}\otimes \operatorname{Tr}_{E},
    \end{equation}
    and
    \begin{equation}\label{eq:2_ext_ch_perm}
        \mathcal{W}_{B'E'}\circ\mathcal{P}_{ABE\to A'B'E'} = \mathcal{P}_{ABE\to A'B'E'}\circ \mathcal{W}_{BE},
    \end{equation}
    where $\mathcal{W}_{BE} \coloneqq W_{BE}(\cdot)\!W_{BE}^{\dagger}$ with $W_{BE}$ defined in~\eqref{eq:SWAP_defn}. The conditions in~\eqref{eq:2_ext_ch_ext} and~\eqref{eq:2_ext_ch_perm} are known as the channel extension condition and the permutation covariance condition, respectively.
\end{definition}
The channel $\mathcal{P}_{ABE\to A'B'E'}$ is said to be a two-extension of  $\mathcal{N}_{AB\to A'B'}$ if the channel extension and permutation covariance conditions, mentioned in~\eqref{eq:2_ext_ch_ext} and~\eqref{eq:2_ext_ch_perm} respectively, hold.

If a channel $\mathcal{N}_{AB\to A'B'}$ is two-extendible, then it is non-signaling from $B$ to $A$~\cite[Appendix A]{HSW23}; that is,
\begin{equation}
    \operatorname{Tr}_{B'}\circ\mathcal{N}_{AB\to A'B'} = \operatorname{Tr}_{B'}\circ\mathcal{N}_{AB\to A'B'}\circ \mathcal{R}^{\pi}_B,
\end{equation}
where $\mathcal{R}^{\pi}_B$ is a channel that traces out the input and replaces it with a maximally mixed state. Moreover, all one-way LOCC channels are two-extendible, as can be seen from a simple construction. An arbitrary one-way LOCC channel can be written in the following form:
    \begin{equation}
        \mathcal{N}_{AB\to A'B'} = \sum_{x\in \mathcal{X}} \mathcal{E}^x_{A\to A'}\otimes \mathcal{F}^x_{B\to B'},
    \end{equation}
    where $\left\{\mathcal{E}^x_{A\to A'}\right\}_{x \in \mathcal{X}}$ is a quantum instrument and $\left\{\mathcal{F}^x_{B\to B'}\right\}_{x\in \mathcal{X}}$ is a set of quantum channels. A two-extension of this channel can be constructed as follows:
    \begin{equation}
        \mathcal{P}_{ABE\to A'B'E'} = \sum_{x\in \mathcal{X}} \mathcal{E}^x_{A\to A'}\otimes \mathcal{F}^x_{B\to B'}\otimes \mathcal{F}^x_{E\to E'}.
    \end{equation}
    Hence, every one-way LOCC channel is two-extendible.

    On the contrary, all two-extendible channels cannot be simulated by local operations and one-way classical communication. Consider the example of a bipartite channel that traces out the input and replaces it with the state mentioned in~\eqref{eq:2_ext_st_example}, which can be mathematically represented as follows:
    \begin{equation}\label{eq:2_ext_ch_example}
        \mathcal{N}_{AB\to A'B'}\!\left(\cdot\right) = \operatorname{Tr}\!\left[\cdot\right]\left(\frac{1}{2}\Phi_{AB} + \frac{1}{2}\frac{I_{AB}}{4}\right).
    \end{equation}
    Since this channel is capable of taking a separable state as input and establishing an entangled state, it is not a one-way LOCC channel. However, one can construct the following two-extension of the channel:
    \begin{equation}
        \mathcal{P}_{ABE\to A'B'E'}\!\left(\cdot\right) \coloneqq \operatorname{Tr}\!\left[\cdot\right]\left(\frac{1}{4}\Phi_{A'B'}\otimes I_{E'} + \frac{1}{4}\Phi_{A'E'}\otimes I_{B'}\right).
    \end{equation}
    Therefore, the channel defined in~\eqref{eq:2_ext_ch_example} is an example of a two-extendible channel that is not a one-way LOCC channel.

\subsection{Unextendible entanglement of states}

Let $\mathbb{R}$ denote the field of real numbers. A generalized divergence~\cite{PV10} is a functional $\mathbf{D}\colon \mathcal{S}(A)\times \mathcal{S}(A) \to \mathbb{R}\cup \{+\infty\}$, such that, for arbitrary states $\rho_A,\sigma_A\in \mathcal{S}(A)$ and an arbitrary channel $\mathcal{N}_{A\to B}$, the data-processing inequality holds
\begin{equation}
    \mathbf{D}\!\left(\rho_A\Vert\sigma_A\right) \ge \mathbf{D}\!\left(\mathcal{N}_{A\to B}(\rho_A)\Vert\mathcal{N}_{A\to B}(\sigma_A)\right).
\end{equation}

Some examples of divergences that commonly appear in quantum information theory are the quantum relative entropy~\cite{Ume62}, Petz-R\'enyi relative entropies~\cite{Petz86}, sandwiched R\'enyi relative entropies~\cite{MDSST13, WWY14}, and geometric R\'enyi relative entropies~\cite{Mat13, FF21}.

The generalized unextendible entanglement of a bipartite  state has been defined in~\cite{WWW24}. We include a short discussion on the topic for necessary development.

\begin{definition}[\cite{WWW24}]
\label{def:unext-ent}
    The generalized unextendible entanglement of a bipartite state $\rho_{AB}$, induced by a generalized divergence $\mathbf{D}$ between states, is defined as
        \begin{equation}
        \label{eq:gen_unext_ent_states}
        \mathbf{E}^u(\rho_{AB}) \coloneqq  \inf_{\omega_{ABE}\in \mathcal{S}\left(ABE\right)} \frac{1}{2}\Big\{\mathbf{D}\!\left(\rho_{AB}\Vert\operatorname{Tr}_{B}\!\left[\omega_{ABE}\right]\right) \colon \operatorname{Tr}_{E}\!\left[\omega_{ABE}\right] = \rho_{AB}\Big\},
    \end{equation}    
    where the optimization is over every state $\rho_{ABE}$ that is an extension of the state $\rho_{AB}$. We also adopt the following alternative notations sometimes because they can be helpful to make the bipartition $A|B$ clear:
    \begin{equation}
        \mathbf{E}^u(A;B)_{\rho} \equiv \mathbf{E}^u(\rho_{A:B}) \equiv \mathbf{E}^u(\rho_{AB}).
    \end{equation}
\end{definition}

Let us define the following set of extensions of a bipartite state~$\rho_{AB}$:
\begin{equation}\label{eq:extensions_state}
    \operatorname{Ext}\!\left(\rho_{AB}\right) \coloneqq \left\{\omega_{ABE}: \operatorname{Tr}_{BE}\!\left[\omega_{ABE}\right] = \rho_{AB}\right\},
\end{equation}
where $E$ is isomorphic to $B$. This allows us to write the generalized unextendible entanglement of  $\rho_{AB}$, induced by the generalized divergence $\mathbf{D}$, as
\begin{equation}
    \mathbf{E}^u(\rho_{AB}) = \inf_{\omega_{ABE}\in \operatorname{Ext}\left(\rho_{AB}\right)}\frac{1}{2}\mathbf{D}\!\left(\rho_{AB}\Vert\operatorname{Tr}_{E}\!\left[\omega_{ABE}\right]\right).
\end{equation}
Alternatively, one can define the set of state-dependent free states as follows:
\begin{equation}\label{eq:free_states}
    \mathcal{F}\!\left(\rho_{AB}\right) \coloneqq \left\{\operatorname{Tr}_{B}\!\left[\omega_{ABE}\right]:\omega_{ABE}\in \operatorname{Ext}\!\left(\rho_{AB}\right)\right\}.
\end{equation}
The generalized unextendible entanglement of the state $\rho_{AB}$ can then be written as follows:
\begin{equation}
    \mathbf{E}^u(\rho_{AB}) = \inf_{\sigma_{AB}\in \mathcal{F}\left(\rho_{AB}\right)}\frac{1}{2}\mathbf{D}\!\left(\rho_{AB}\Vert\sigma_{AB}\right).
\end{equation}
\begin{theorem}[\cite{WWW24}]\label{theo:unext_ent_st_monotonicity}
    The generalized unextendible entanglement of bipartite state does not increase under the action of a two-extendible channel. That is,
    \begin{equation}
        \mathbf{E}^u\!\left(\rho_{AB}\right) \ge \mathbf{E}^u\!\left(\mathcal{N}_{AB\to A'B'}\!\left(\rho_{AB}\right)\right),
    \end{equation}
    where $\mathcal{N}_{AB\to A'B'}$ is a two-extendible channel.
\end{theorem}
\begin{proof}
	See Theorem 2 in~\cite{WWW24}.
\end{proof}

\medskip

A direct consequence of Theorem~\ref{theo:unext_ent_st_monotonicity} is that the generalized unextendible entanglement of a bipartite state does not increase under the action of one-way LOCC channels.

The generalized unextendible entanglement provides a framework for quantifying the unextendibility of a bipartite state $\rho_{AB}$ with respect to the system $B$. A different measure for unextendibility was considered in~\cite{KDWW19, KDWW21}, where the divergence was measured from the fixed set of two-extendible states. However, in Definition~\ref{def:unext-ent}, the divergence is measured by means of a set of states that depend on the input state itself. Although both measures are equal to the minimal possible value of $\mathbf{D}$ when $\rho_{AB}$ is two-extendible, they are not equal in general.

The unextendible entanglement is a measure of entanglement between two systems, and as such, it is expected to obtain its maximum value for the maximally entangled state. This is indeed the case, as is evident from the following argument. An arbitrary bipartite state $\rho_{AB}$ can be established between Alice and Bob with the help of a maximally entangled state $\Phi^d_{A_0B_0}$ of sufficiently large Schmidt rank and a one-way LOCC channel, where $\operatorname{dim}(A_0) = \operatorname{dim}(B_0) = \min\{\operatorname{dim}(A),\operatorname{dim}(B)\}$. A simple protocol to perform this transformation is as follows: Alice prepares the state $\rho_{AA'}$ locally. We can assume $\operatorname{dim}(A')\le \operatorname{dim}(A)$ without loss of generality. She uses the maximally entangled state $\Phi^d_{A_0B_0}$ to implement the teleportation protocol and send the state on system $A'$ to Bob, thus establishing the state $\rho_{AB}$ between Alice and Bob. The aforementioned protocol can be mathematically represented as the following one-way LOCC channel~\cite[Chapter~5]{KW24}:
    \begin{equation}\label{eq:teleportation_protocol}
        \mathcal{L}^{\rho,\to}_{A_0B_0\to AB}\!\left(\Phi^d_{A_0B_0}\right) = \sum_{x,z=0}^{d-1} \operatorname{Tr}_{A_0A'}\!\left[\Phi^{z,x}_{A_0A'}W^{z,x}_{B_0}\left(\rho_{AA'}\otimes\Phi_{A_0B_0}\right)\left(W^{z,x}_{B_0}\right)^{\dagger}\right],
    \end{equation}
    where $\{W^{z,x}\}_{z,x}$ is the set of Heisenberg--Weyl operators and $\Phi^{z,x}_{A_0A'} \coloneqq W^{z,x}_{A_0}\Phi^d_{A_0A'}(W^{z,x}_{A_0})^\dag$. Since the generalized unextendible entanglement of a bipartite state does not increase under the action of a one-way LOCC channel, the following inequality holds for every state $\rho_{AB}$:
    \begin{equation}\label{eq:unext_ent_le_max_ent}
        \mathbf{E}^u\!\left(\rho_{AB}\right) = \mathbf{E}^u\!\left(\mathcal{L}^{\rho,\to}_{A_0B_0\to AB}\!\left(\Phi^d_{A_0B_0}\right)\right) \le \mathbf{E}^u\!\left(\Phi^d_{A_0B_0}\right),
    \end{equation}
    where $\mathcal{L}^{\rho,\to}_{A_0B_0\to AB}$ is the channel defined in~\eqref{eq:teleportation_protocol} and $d\coloneqq \min\{\operatorname{dim}(A),\operatorname{dim}(B)\}$.

\subsubsection{Smooth-min unextendible entanglement}

The unextendible entanglement of states was studied in detail in~\cite{WWW24} for several different underlying divergences, with applications in finding efficiently computable upper bounds on the probabilistic and exact one-way distillable entanglement and key  of a state. A stronger no-go theorem for probabilistic key distillation was obtained in~\cite{SW24} using the min-relative entropy as the underlying divergence for the unextendible entanglement. In this work, we are specifically interested in the unextendible entanglement of a state induced by the smooth min-relative entropy to understand the limits of one-shot approximate distillable key of a bipartite state using one-way LOCC channels. As a special case of~\eqref{eq:gen_unext_ent_states}, we define the smooth min-unextendible entanglement as follows:
\begin{equation}
    E^{u,\varepsilon}_{\min}\!\left(\rho_{AB}\right) \coloneqq \inf_{\sigma_{AB}\in \mathcal{F}\left(\rho_{AB}\right)} \frac{1}{2}D^{\varepsilon}_{\min}\!\left(\rho_{AB}\Vert\sigma_{AB}\right),
\end{equation}
where the set $\mathcal{F}\!\left(\rho_{AB}\right)$ was defined in~\eqref{eq:free_states} and
\begin{equation}\label{eq:smooth_min_rel_ent_defn}
    D^{\varepsilon}_{\min}\!\left(\rho\Vert\sigma\right) \coloneqq -\log_2 \inf_{0\le \Lambda \le I}\left\{\operatorname{Tr}\!\left[\Lambda\sigma\right]: \operatorname{Tr}\!\left[\Lambda\rho\right] \ge 1-\varepsilon\right\}
\end{equation}
is the smooth min-relative entropy~\cite{BD10,BD11}, also known as the hypothesis testing relative entropy~\cite{WR12}.

We will often use the following quantity in our discussions:
\begin{equation}\label{eq:J_hypo_test_defn}
    J^{\varepsilon}_{\min}\!\left(\rho_{AB}\right) \coloneqq 2^{-2E^{u,\varepsilon}_{\min}\!\left(\rho_{AB}\right)},
\end{equation}
which we will abbreviate as $J^{\varepsilon}_{\min}$ when the state it acts upon is obvious from the context. The quantity $J^{\varepsilon}_{\min}\left(\rho_{AB}\right)$ can alternatively be written as follows:
\begin{equation}\label{eq:J_eps_interpret_1}
    J^{\varepsilon}_{\min}\!\left(\rho_{AB}\right) = \sup_{\sigma_{AB}\in \mathcal{F}(\rho_{AB})} \inf_{0\le \Lambda \le I}\left\{\operatorname{Tr}\!\left[\Lambda\sigma\right] : \operatorname{Tr}\!\left[\Lambda\rho\right] \ge 1-\varepsilon\right\}.
\end{equation}

The set $\mathcal{F}(\rho_{AB})$ is a convex set of quantum states, and the set $\{\Lambda: 0\le \Lambda \le I: \operatorname{Tr}\!\left[\Lambda\rho\right]\ge 1-\varepsilon\}$ is a convex set of measurement operators for a fixed state $\rho$. Therefore, using Sion's minimax theorem~\cite{Sion58}, we can interchange the supremum and infimum to arrive at the following equality:
\begin{equation}
    J^{\varepsilon}_{\min}\!\left(\rho_{AB}\right) = \inf_{0\le \Lambda \le I}\sup_{\sigma_{AB}\in \mathcal{F}(\rho_{AB})} \left\{\operatorname{Tr}\!\left[\Lambda\sigma\right] : \operatorname{Tr}\!\left[\Lambda\rho\right] \ge 1-\varepsilon\right\}.
\end{equation}
The above equality gives an interpretation for the quantity $J^{\varepsilon}_{\min}(\rho_{AB})$ in the hypothesis testing setting. Given a quantum state $\rho_{AB}$ and an arbitrary state $\sigma_{AB}\in \mathcal{F}(\rho_{AB})$, the quantity $J^{\varepsilon}_{\min}(\rho_{AB})$ denotes the minimum type-II error probability in the worst case when the type-I error probability is guaranteed to be less than $\varepsilon$. Note that the quantity $J^{\varepsilon}_{\min}$ decreases monotonically with increasing smooth-min unextendible entanglement. As such, $J^{\varepsilon}_{\min}$ is smaller for highly entangled states and larger for weakly entangled states, and the minimum value is achieved for the maximally entangled state due to~\eqref{eq:unext_ent_le_max_ent}.

\begin{proposition}\label{prop:max_ent_unext_ent_hypo_test}
    The unextendible entanglement of a maximally entangled state with Schmidt rank $d$ is equal to the following:
    \begin{equation}
        E^{u,\varepsilon}_{\min}\!\left(\Phi^d_{AB}\right) = \log_2d -\frac{1}{2}\log_2(1-\varepsilon).
    \end{equation}
\end{proposition}
\begin{proof}
    See Appendix~\ref{app:max_ent_unext_ent_hypo_test}.
\end{proof}

\begin{proposition}\label{prop:J_eps_range}
    The smooth-min unextendible entanglement of a state $\rho_{AB}$ is bounded as follows:
    \begin{equation}
        -\frac{1}{2}\log_2(1-\varepsilon) \le E^{u,\varepsilon}_{\min}\!\left(\rho_{AB}\right) \le \log_2 d -\frac{1}{2}\log_2(1-\varepsilon),
    \end{equation}
    where $d \coloneqq \min\{\operatorname{dim}(A),\operatorname{dim}(B)\}$ with $\operatorname{dim}(A)$ and $\operatorname{dim}(B)$ being the dimensions of system $A$ and $B$, respectively.

    Consequently,
    \begin{equation}\label{eq:J_eps_range}
        \frac{1-\varepsilon}{d^2} \le J^{\varepsilon}\!\left(\rho_{AB}\right) \le 1-\varepsilon.
    \end{equation}
\end{proposition}
\begin{proof}
    See Appendix~\ref{app:J_eps_range}.
\end{proof}

\medskip

The smooth-min unextendible entanglement of a bipartite state can be computed using a semidefinite program (see Appendix~\ref{app:semidefinite_programs}).

\subsubsection{\texorpdfstring{$\alpha$}{alpha}-Sandwiched unextendible entanglement}

Another quantity that is relevant to this work is the $\alpha$-sandwiched unextendible entanglement~\cite{WWW24}, which is defined as follows:
\begin{equation}
    \widetilde{E}^u_{\alpha}\!\left(\rho_{AB}\right) \coloneqq \inf_{\sigma_{AB}\in \mathcal{F}(\rho_{AB})}\frac{1}{2}\widetilde{D}_{\alpha}\!\left(\rho_{AB}\Vert\sigma_{AB}\right) \qquad \forall \alpha \in \left(1,\infty\right),
\end{equation}
where 
\begin{equation}
    \widetilde{D}_{\alpha}\!\left(\rho\Vert\sigma\right) \coloneqq \frac{1}{\alpha - 1}\log_2 \operatorname{Tr}\!\left[\left(\sigma^{(1-\alpha)/2\alpha}\rho\sigma^{(1-\alpha)/2\alpha}\right)^{\alpha}\right]
\end{equation}
is the $\alpha$-sandwiched R\'enyi relative entropy~\cite{MDSST13, WWY14}. The $\alpha$-sandwiched unextendible entanglement was defined for all $\alpha \in (0,1)\cup(1,\infty)$ in~\cite{WWW24}, but we restrict our development here to $\alpha \in \left(1,\infty\right)$, due to technical reasons that will become apparent in Section~\ref{sec:distillable_key_results}.

Consider the special case when $\alpha \to \infty$. It was shown in~\cite[Theorem 5]{MDSST13}  that the $\alpha$-sandwiched relative entropy is equal to the max-relative entropy~\cite{Dat09}  when $\alpha \to \infty$ ; that is,
\begin{equation}\label{eq:max_re_ent_defn}
    D_{\max}\!\left(\rho\Vert\sigma\right) = \lim_{\alpha \to \infty}\widetilde{D}_{\alpha}\!\left(\rho\Vert\sigma\right).
\end{equation}
Corresponding to the max-relative entropy, the max-unextendible entanglement is defined as follows:
\begin{equation}
    E^u_{\max}\!\left(\rho_{AB}\right) \coloneqq \frac{1}{2}\inf_{\sigma_{AB}\in \mathcal{F}(\rho_{AB})} D_{\max}\!\left(\rho_{AB}\Vert\sigma_{AB}\right).
\end{equation}

Besides monotonicity under two-extendible channels, the $\alpha$-sandwiched unextendible entanglement has several other properties desirable in a resource monotone. We state some of the relevant properties below.
\begin{itemize}
    \item \textbf{Subadditivity:} The $\alpha$-sandwiched unextendible entanglement obeys the following subadditivity inequality~\cite[Proposition 13]{WWW24}
    \begin{equation}
        \widetilde{E}^u_{\alpha}\!\left(\rho_{A_1B_1}\otimes \sigma_{A_2B_2}\right) \le \widetilde{E}^u_{\alpha}\!\left(\rho_{A_1B_1}\right) + \widetilde{E}^u_{\alpha}\!\left(\sigma_{A_2B_2}\right) \qquad \forall \alpha \in \left[\frac{1}{2},1\right)\cup\left(1,\infty\right).
    \end{equation}
    \item \textbf{Additivity of max-unextendible entanglement:} The max-unextendible entanglement is additive under tensor product of states~\cite[Proposition 14]{WWW24}; that is,
    \begin{equation}
        E^u_{\max}\!\left(\rho_{A_1B_1}\otimes \sigma_{A_2B_2}\right) = E^u_{\max}\!\left(\rho_{A_1B_1}\right) + E^u_{\max}\!\left(\sigma_{A_2B_2}\right).
    \end{equation}
    \item \textbf{Monotonicity in $\alpha$:} The $\alpha$-sandwiched unextendible entanglement of a state increases monotonically with increasing $\alpha$, which follows from the fact that the $\alpha$-sandwiched R\'enyi relative entropy between two states increases monotonically with $\alpha$~\cite[Theorem 7]{MDSST13}.
    \item \textbf{Semidefinite representation:} The max-unextendible entanglement can be computed using a semidefinite program~\cite{WWW24} (see Appendix~\ref{app:semidefinite_programs} for a review).
\end{itemize}

The max-unextendible entanglement of a state can be written as a limiting case of the $\alpha$-sandwiched unextendible entanglement as follows:
\begin{align}
    \lim_{\alpha\to \infty}\widetilde{E}^u_{\alpha}\!\left(\rho_{AB}\right) &= \sup_{\alpha\in (1,\infty)}\frac{1}{2}\inf_{\sigma_{AB}\in\mathcal{F}\!\left(\rho_{AB}\right)}\widetilde{D}_{\alpha}\!\left(\rho_{AB}\Vert\sigma_{AB}\right)\\
    &= \frac{1}{2}\inf_{\sigma_{AB}\in\mathcal{F}\!\left(\rho_{AB}\right)}\sup_{\alpha\in (1,\infty)}\widetilde{D}_{\alpha}\!\left(\rho_{AB}\Vert\sigma_{AB}\right)\\
    &= \frac{1}{2}\inf_{\sigma_{AB}\in\mathcal{F}\!\left(\rho_{AB}\right)} D_{\max}\!\left(\rho_{AB}\Vert\sigma_{AB}\right)\\
    &= E^u_{\max}\!\left(\rho_{AB}\right)\label{eq:Emax_eq_lim_inf_sandwiched},
\end{align}
where the first equality follows from the monotonicity of the $\alpha$-sandwiched unextendible entanglement in $\alpha$. The $\alpha$-sandwiched R\'enyi relative entropy $\widetilde{D}_{\alpha}\!\left(\rho\Vert\sigma\right)$ is lower-semicontinuous with respect to $\sigma$~\cite[Lemma IV.8]{MO21} (see also \cite[Remark~38]{DKQSWW23}), and it increases monotonically with $\alpha \in (1,\infty)$. Therefore, we can use the Mosonyi--Hiai minimax theorem from~\cite[Corollary~A.2]{MH11} to arrive at the second equality above. The last two equalities follow from the equality in~\eqref{eq:max_re_ent_defn} and the definition of the max-unextendible entanglement, respectively.

The subadditivity of $\alpha$-sandwiched unextendible entanglement is useful in the analysis of one-shot, one-way secret-key distillation from independent and identically distributed (i.i.d.)~copies of a bipartite state, as we shall see in Section~\ref{sec:sandwich_st_ub}.

\section{Limits on secret-key distillation from bipartite states}\label{sec:distillable_key_results}

In this section we use the framework of unextendible entanglement discussed in Section~\ref{sec:two_extendibility} to obtain upper bounds on the one-shot, one-way distillable key of a state defined in~\eqref{eq:1W_dist_key_st_defn}.

Consider a quantum state $\psi^{\gamma}_{ABA'B'E}$, which is an extension of a private state $\gamma^k_{ABA'B'}$ with system~$E$ held by an eavesdropper. The reduced state of $\psi^{\gamma}_{ABA'B'E}$ on systems $AE$ is a product state, with the reduced state on system $A$ being the maximally mixed state. This is evident from the equivalence between a bipartite private state and a tripartite secret-key state. It was further shown in~\cite[Appendix K]{WWW24} that applying a twisting unitary on the joint systems of Alice and the eavesdropper is also insufficient to establish any correlations between Alice's key system, $A$, and the eavesdropper's system, $E$. We state this formally in Lemma~\ref{lem:twist_priv_marginal}, which we later use to establish the main results of this work.

\begin{lemma}[\cite{WWW24}]\label{lem:twist_priv_marginal}
    Let $\psi^{\gamma}_{ABA'B'EE'R}$ be a purification of a bipartite private state $\gamma^k_{ABA'B'}$.
    For every purification for which system $E$ is isomorphic to $B$ and system $E'$ is isomorphic to $B'$, the following equality holds:
    \begin{equation}\label{eq:twist_priv_marginal}
        \operatorname{Tr}_{A'BB'E'R}\!\left[W_{AEA'E'}^{\dagger}\psi^{\gamma}W_{AEA'E'}\right] = \pi_A\otimes\sigma_E,
    \end{equation}
    where $\tau_E$ is a quantum state, $\pi_A$ is the maximally mixed state, and $W_{AEA'E'}$ is a twisting unitary of the form given in~\eqref{eq:twisting_unitary_defn}.
\end{lemma}
\begin{proof}
    See~\cite[Appendix K]{WWW24}.
\end{proof}

\medskip

The fact that Alice's system is in a product state with the eavesdropper ensures that the state shared by Alice and the eavesdropper does not pass the privacy test with a probability greater than~$\frac{1}{k}$. We generalize this statement to approximate private states in Lemma~\ref{lem:priv_marg_phi_fid_bnd} below.
\begin{lemma}\label{lem:priv_marg_phi_fid_bnd}
	Fix $k \in \mathbb{N}$ and $\varepsilon \in \left[0, 1-\frac{1}{k^2}\right]$. Let $\sigma_{ABA'B'}$ be a quantum state such that
 \begin{equation}
     F(\sigma_{ABA'B'},\gamma^k_{ABA'B'}) \ge 1-\varepsilon,
 \end{equation}
 where $\gamma^k_{ABA'B'}$ is a bipartite private state holding $\log_2 k$ secret key bits. Let $V_{ABA'B'}$ be the twisting unitary corresponding to the private state; that is, there exists a state $\tau_{A'B'}$ such that
	\begin{equation}
		\gamma^k_{ABA'B'} = V_{ABA'B'}\!\left(\Phi^k_{AB}\otimes \tau_{A'B'}\right)V^{\dagger}_{ABA'B'}.
	\end{equation}
	The probability of an arbitrary state $\omega_{AEA'E'} \in \mathcal{F}(\sigma_{ABA'B'})$ passing the $\gamma^k$-privacy test is bounded from above by the following quantity:
    \begin{equation}
        \operatorname{Tr}\!\left[\Pi^{\gamma}_{AEA'E'}\omega_{AEA'E'}\right] \le \varsigma(\varepsilon, k), 
    \end{equation}
    where 
 \begin{equation}
     \varsigma(\varepsilon, k) \coloneqq \varepsilon + \frac{1-2\varepsilon}{k^2} +\frac{2\sqrt{(k^2-1)\varepsilon(1-\varepsilon)} }{k^2},
     \label{eq:varsig-func}
 \end{equation}
$\Pi^{\gamma}_{ABA'B'}$ is the privacy test, system $B$ is isomorphic to $E$, and system $B'$ is isomorphic to $E'$. The set $\mathcal{F}\!\left(\sigma_{ABA'B'}\right)$ was defined in~\eqref{eq:free_states}.
\end{lemma}

\begin{proof}
    	Let $\phi^{\sigma}_{ABA'B'EE'R}$ be an arbitrary  purification of $\sigma_{ABA'B'}$, such that system $E$ is isomorphic to $B$ and system $E'$ is isomorphic to $B'$. Then the quantum state $\operatorname{Tr}_{BB'R}\!\left[\phi^{\sigma}\right]$ is in the set $\mathcal{F}(\sigma_{ABA'B'})$. Let $\psi^{\gamma}_{ABA'B'EE'R}$ be a purification of $\gamma^k_{ABA'B'}$. An arbitrary purification of the private state $\gamma^k_{ABA'B'}$ is of the following form:
\begin{equation}\label{eq:priv_st_arb_purification}
	\psi^{\gamma}_{ABA'B'X} = V_{ABA'B'}\left(\Phi^k_{AB}\otimes\psi^{\tau}_{A'B'X}\right)V^{\dagger}_{ABA'B'},
\end{equation}
where $\psi^{\tau}_{A'B'X}$ is a pure state and $X$ is a purifying system. 

As a consequence of Uhlmann's theorem, consider that
\begin{align}
    F\!\left(\sigma_{ABA'B'},\gamma^k_{ABA'B'}\right) &= \max_{\psi^{\gamma}} \left|\langle \psi^{\gamma}|\phi^{\sigma}\rangle\right|^2\\
    &= \max_{\psi^{\gamma}} F\!\left(\psi^{\gamma}_{ABA'B'EE'R},\phi^{\sigma}_{ABA'B'EE'R}\right)\\
    &= \max_{\psi^{\tau}_{A'B'EE'R}} F\!\left( V_{ABA'B'}\left(\Phi^k_{AB}\otimes\psi^{\tau}_{A'B'EE'R}\right)V^{\dagger}_{ABA'B'}, \phi^{\sigma}_{ABA'B'EE'R}\right).
\end{align}
The maximizations in the first and second equalities are over every purification $\psi^{\gamma}$ of $\gamma^k_{ABA'B'}$ on the systems $ABA'B'EE'R$. Since every purification of the state $\gamma^k_{ABA'B'}$ can be written in the form mentioned in~\eqref{eq:priv_st_arb_purification}, the maximization over every purification $\psi^{\gamma}_{ABA'B'EE'R}$ is equivalent to a maximization over every purification $\psi^{\tau}_{A'B'EE'R}$ of the state $\tau_{A'B'}$.
Therefore, there exists a pure state $\psi^{\tau}_{A'B'EE'R}$ and a corresponding purification $\psi^{\gamma}_{ABA'B'EE'R}$ such that the following equality holds:
\begin{equation}\label{eq:uhlmann_application}
	F\!\left(\sigma_{ABA'B'},\gamma^k_{ABA'B'}\right) = |\langle \psi^{\gamma}|\phi^{\sigma}\rangle|^2 = F(\psi^{\gamma},\phi^{\sigma}).
\end{equation} 

Using the data-processing inequality for fidelity of states, consider that
\begin{align}
	F(\psi^{\gamma},\phi^{\sigma})
 & \leq F\!\left(\operatorname{Tr}_{A'BB'E'R}\!\left[V_{AEA'E'}^{\dagger}\psi^{\gamma}V_{AEA'E'}\right],\operatorname{Tr}_{A'BB'E'R}\!\left[V_{AEA'E'}^{\dagger}\phi^{\sigma}V_{AEA'E'}\right]\right)\\
&= F\left(\pi_A\otimes \tau_E, \operatorname{Tr}_{A'E'}\!\left[V_{AEA'E'}^{\dagger}\omega_{AEA'E'} V_{AEA'E'}\right]\right),\label{eq:fid_purifications_le_fid_marginals}
\end{align}
where we have used Lemma~\ref{lem:twist_priv_marginal} to arrive at the final equality and $\omega_{AEA'E'} $ is defined in the statement of the lemma.

Now consider the twirling channel,  defined as follows:
\begin{equation}\label{eq:twirl_defn}
	\mathcal{T}_{AB}\!\left(\cdot\right) = \int dU \left(U_A\otimes \overline{U}_B\right)\left(\cdot\right)\left(U_A\otimes \overline{U}_B\right)^{\dagger},
\end{equation} 
where the integral is with respect to the Haar measure. The action of this channel on an arbitrary quantum state $\rho_{AB}$, with $\operatorname{dim}(A) = \operatorname{dim}(B) = d$, results in the following isotropic state~\cite{HH99,Watrous2018}:
\begin{equation}\label{eq:twirl_to_isotropic}
	\mathcal{T}_{AB}\!\left(\rho_{AB}\right) = \operatorname{Tr}\!\left[\Phi^d_{AB}\rho_{AB}\right]\Phi^d_{AB} + \left(1-\operatorname{Tr}\!\left[\Phi^d_{AB}\rho_{AB}\right]\right)\frac{I_{AB}-\Phi^d_{AB}}{d^2-1}.
\end{equation}
Consider that
\begin{align}
	\operatorname{Tr}\!\left[\Phi^k_{AE}\left(\pi_A\otimes\tau_E\right)\right] &= \frac{1}{k}\operatorname{Tr}\!\left[\Phi^k_{AE}(I_A\otimes\tau_E)\right]\\
	&= \frac{1}{k^2}\operatorname{Tr}\!\left[\tau_E\right]\\
	&= \frac{1}{k^2}.
\end{align}
Therefore,
\begin{equation}\label{eq:twirled_max_mix}
	\mathcal{T}_{AE}\!\left(\pi_A\otimes \tau_E\right) = \frac{1}{k^2}\Phi^k_{AE} + \left(1-\frac{1}{k^2}\right)\frac{I_{AE}-\Phi^k_{AE}}{k^2-1} = \frac{I_{AE}}{k^2}.
\end{equation}
The action of the twirling channel $\mathcal{T}_{AE}$ on the state $\operatorname{Tr}_{A'E'}\!\left[V^{\dagger}\omega_{AEA'E'}V\right]$ is given by the following expression:
\begin{equation}\label{eq:twirled_priv_marg}
	\mathcal{T}_{AE}\!\left(\operatorname{Tr}_{A'E'}\!\left[V^{\dagger}\omega_{AEA'E'}V\right]\right) = q\,\Phi^k_{AE} + \left(1-q\right)\frac{I_{AE}-\Phi^k_{AE}}{k^2-1},
\end{equation}
where
\begin{equation}\label{eq:q_expression}
	q \coloneqq \operatorname{Tr}\!\left[\Phi^k_{AE}\left(\operatorname{Tr}_{A'E'}\!\left[V^{\dagger}\omega_{AEA'E'}V\right]\right)\right].
\end{equation}
The above expression can be rewritten as follows:
\begin{equation}\label{eq:q_priv_test_interpret}
    q = \operatorname{Tr}\!\left[V_{AEA'E'}\left(\Phi^k_{AE}\otimes I_{A'E'}\right)V^{\dagger}_{AEA'E'}\omega_{AEA'E'}\right] = \operatorname{Tr}\!\left[\Pi^{\gamma}_{AEA'E'}\omega_{AEA'E'}\right],
\end{equation}
where $\Pi^{\gamma}_{AEA'E'}$ is the $\gamma$-privacy test defined in~\eqref{eq:priv_test_defn}. Therefore, $q$ can be interpreted as the probability of the state $\omega_{AEA'E'}$ passing the $\gamma$-privacy test.

Going back to~\eqref{eq:fid_purifications_le_fid_marginals} and using the data-processing inequality for fidelity, we arrive at the following inequality:
\begin{align}
	F(\psi^{\gamma},\phi^{\sigma}) &\le F\!\left(\mathcal{T}_{AE}\!\left(\pi_A\otimes \tau_E\right), \mathcal{T}_{AE}\!\left(\operatorname{Tr}_{A'E'}\!\left[V^{\dagger}\omega V\right]\right)\right)\\
	&= F\!\left(\pi_{AE},\mathcal{T}_{AE}\!\left(\operatorname{Tr}_{A'E'}\!\left[V^{\dagger}\omega V\right]\right)\right).\label{eq:fid_purifications_le_fid_twirled}
\end{align}
Since the above inequality holds for all $\sigma_{ABA'B'}$ such that $F(\sigma_{ABA'B'},\gamma^k_{ABA'B'})\ge 1-\varepsilon$ for some private state $\gamma^k_{ABA'B'}$ and every $\omega_{AEA'E'}$ that is in the set $\mathcal{F}(\sigma_{ABA'B'})$, we can combine~\eqref{eq:uhlmann_application},~\eqref{eq:fid_purifications_le_fid_marginals}, and~\eqref{eq:fid_purifications_le_fid_twirled} to arrive at the following inequality:
\begin{align}
	1-\varepsilon &\le F(\sigma_{ABA'B'},\gamma^k_{ABA'B'})\\
 &\le F\!\left(\pi_{AE},\mathcal{T}_{AE}\!\left(\operatorname{Tr}_{A'E'}\!\left[V^{\dagger}\omega V\right]\right)\right)\label{eq:fid_twirl_priv_max_mix}\\
 &= F\!\left(\pi_{AE},q~\Phi^k_{AE} + (1-q)\frac{I_{AE}-\Phi^k_{AE}}{k^2}\right)\\
 &= \left(\sqrt{\frac{q}{k^2}} + \sqrt{\left(1-q\right)\left(1-\frac{1}{k^2}\right)}\right)^2,\label{eq:q_eq_main}
\end{align}
where the first equality follows from~\eqref{eq:twirled_priv_marg} and the last equality follows by evaluating the fidelity between the two isotropic states.

The inequality in~\eqref{eq:q_eq_main} is satisfied for all $q \in [0,1]$ if $\varepsilon \ge 1-\frac{1}{k^2}$. If $\varepsilon \in \left[0,1-\frac{1}{k^2}\right]$ then $q$ must lie in the following range for~\eqref{eq:q_eq_main} to hold:
\begin{equation}\label{eq:q_final_expression}
    0\le q \le \varepsilon + \frac{1-2\varepsilon}{k^2} + \frac{2\sqrt{(k^2-1)\varepsilon(1-\varepsilon)}}{k^2}.
\end{equation}
See Appendix~\ref{app:priv_marg_phi_fid_bnd} for a detailed proof of the aforementioned statement. Since we have identified $q$ to be the probability of the state $\omega_{AEA'E'}$ passing the $\gamma$-privacy test, we conclude the statement of the lemma.
\end{proof}

\medskip

Lemma~\ref{lem:priv_marg_phi_fid_bnd} plays a central role in obtaining lower bounds on the unextendible entanglement of an approximate private state $\sigma_{ABA'B'}$ such that $F(\sigma_{ABA'B'},\gamma^k_{ABA'B'}) \ge 1-\varepsilon$ for some private state~$\gamma^k_{ABA'B'}$. 

\begin{remark}
    Note that $F\!\left(\pi_{AB},\Phi^k_{AB}\right) = \frac{1}{k^2}$, where $\pi_{AB}$ is the maximally mixed state. Therefore, if $\varepsilon \ge 1-\frac{1}{k^2}$ for some $k\in \mathbb{N}$, then the following inequality holds for the one-shot, one-way distillable key of an arbitrary bipartite state $\rho_{AB}$:
    \begin{equation}
        K^{\varepsilon,\to}_{D}\!\left(\rho_{AB}\right) \ge \log_2 k.
    \end{equation}
    As such, choosing $\varepsilon \ge 1-\frac{1}{k^2}$ allows the one-shot, one-way distillable key of a separable state to be non-zero, which makes this regime  uninteresting from a practical perspective.
\end{remark}

\subsection{Smooth-min unextendible entanglement upper bound on one-shot distillable secret key of bipartite states}\label{sec:smooth_min_st_ub}

The $\gamma$-privacy test is a special POVM that can be used to distinguish a bipartite state $\sigma_{ABA'B'}$ from any state $\omega_{AEA'E'} \in \mathcal{F}(\sigma_{ABA'B'})$. Recall that the smooth-min unextendible entanglement quantifies the ability to distinguish between a state $\rho_{AB}$ from an arbitrary state $\sigma_{AB} \in \mathcal{F}(\rho_{AB})$. Therefore, Lemma~\ref{lem:priv_marg_phi_fid_bnd} allows us to obtain a bound on the unextendible entanglement of a state $\sigma_{ABA'B'}$ satisfying $F(\sigma_{ABA'B'},\gamma^k_{ABA'B'})\ge 1-\varepsilon$. This is stated as Proposition~\ref{prop:unext_ent_hyp_test_lb} below.
\begin{proposition}\label{prop:unext_ent_hyp_test_lb}
	Fix $k\in \mathbb{N}$ and $\varepsilon \in \left[0,1-\frac{1}{k^2}\right]$. Consider a quantum state $\sigma_{ABA'B'}$ such that $F(\sigma_{ABA'B'},\gamma^k_{ABA'B'})\ge 1-\varepsilon$, where $\gamma^k_{ABA'B}$ is a bipartite private state. The smooth-min unextendible entanglement of the state $\sigma_{ABA'B'}$ is bounded from below by the following quantity:
	\begin{equation}
		E^{u,\varepsilon}_{\min}\!\left(\sigma_{ABA'B'}\right) \ge -\frac{1}{2}\log_2\!\left(\varsigma(\varepsilon, k)\right),
	\end{equation}
 where $\varsigma(\varepsilon, k)$ is defined in~\eqref{eq:varsig-func}.
\end{proposition}

\begin{proof}
	Let $\sigma_{ABA'B'}$ be a quantum state such that $F(\sigma_{ABA'B'},\gamma^k_{ABA'B'})\ge 1-\varepsilon$ for some bipartite private state $\gamma^k_{ABA'B'}$ and some $\varepsilon \in \left[0,1-\frac{1}{k^2}\right]$. The following inequality holds for every such state $\sigma_{ABA'B'}$~\cite[Lemma~9]{WTB17}:
\begin{equation}
	\operatorname{Tr}\!\left[\Pi^{\gamma}_{ABA'B'}\sigma_{ABA'B'}\right] \ge 1-\varepsilon,
\end{equation}
where the projector $\Pi^{\gamma}_{ABA'B'}$ is the $\gamma^k$-privacy test defined in~\eqref{eq:priv_test_defn}.

Now consider the hypothesis testing relative entropy between the state $\sigma_{ABA'B'}$, for which $F(\sigma_{ABA'B'},\gamma^k_{ABA'B'})\ge 1-\varepsilon$ for some private state $\gamma^k_{ABA'B'}$, and an arbitrary state $\omega_{AEA'E'}$ which lies in the set $\mathcal{F}(\sigma_{ABA'B'})$:
\begin{align}
	D^{\varepsilon}_{\min}\!\left(\sigma_{ABA'B'}\Vert \omega_{AEA'E'}\right) = \sup_{0\le \Lambda \le I} -\log_2  \left\{\operatorname{Tr}\!\left[\Lambda \omega\right] : \operatorname{Tr}\!\left[\Lambda \sigma\right] \ge 1-\varepsilon\right\}.\label{eq:hypo_test_priv_defn_2}
\end{align}
It is understood here that system $E$ is isomorphic to $B$ and system $E'$ is isomorphic to $B'$.

Since $\Pi^{\gamma}_{ABA'B'}$ is a specific measurement operator that satisfies the constraints in~\eqref{eq:hypo_test_priv_defn_2}, the following inequality holds:
\begin{equation}\label{eq:hypo_test_priv_test_bnd}
	D^{\varepsilon}_{\min}\!\left(\sigma_{ABA'B'}\Vert\omega_{AEA'E'}\right) \ge -\log_2 \operatorname{Tr}\!\left[\Pi^{\gamma}_{AEA'E'}\omega_{AEA'E'}\right],
\end{equation}
which leads to the following inequality after applying Lemma~\ref{lem:priv_marg_phi_fid_bnd}:
\begin{equation}\label{eq:hypo_test_priv_close_marg_lb}
	D^{\varepsilon}_{\min}\!\left(\sigma_{ABA'B'}\Vert\omega_{AEA'E'}\right) \ge -\log_2\!\left(\varsigma(\varepsilon, k)\right).
\end{equation}
Since~\eqref{eq:hypo_test_priv_close_marg_lb} holds for all $\omega_{AEA'E'}$ in the set $\mathcal{F}(\sigma_{ABA'B'})$, we conclude that the smooth-min unextendible entanglement of the state $\sigma_{ABA'B'}$ is bounded from below by the following quantity:
\begin{align}
	E^{u,\varepsilon}_{\min}\!\left(\sigma_{ABA'B'}\right) &= \frac{1}{2}\inf_{\omega \in \mathcal{F}(\sigma)} D^{\varepsilon}_{\min}\!\left(\sigma_{ABA'B'}\Vert\omega_{AEA'E'}\right)\\
    &\ge -\frac{1}{2}\log_2\!\left(\varsigma(\varepsilon, k)\right),
\end{align}
thus completing the proof.
\end{proof}

\medskip

\begin{remark}\label{rem:nogo_key_distill}
    For a fixed $\varepsilon \in \left[0,\frac{3}{4}\right]$, if $E^{u,\varepsilon}_{\min}\!\left(\sigma_{ABA'B'}\right)<-\frac{1}{2}\log_2\!\left(\varsigma\!\left(\varepsilon,2\right)\right)$, then there does not exist a private state $\gamma^k_{ABA'B'}$ such that $F\!\left(\sigma_{ABA'B'},\gamma^k_{ABA'B'}\right) \ge 1-\varepsilon$. 
\end{remark}

The unextendible entanglement of a state does not increase under the action of a one-way LOCC channel. Therefore, for any one-way LOCC channel $\mathcal{L}^{\to}_{AB\to A'B'A''B''}$ that is used to distill an $\varepsilon$-approximate private state $\sigma_{A'B'A''B''}$ from a bipartite resource state $\rho_{AB}$, the unextendible entanglement of $\rho_{AB}$ must be larger than the unextendible entanglement of $\sigma_{A'B'A''B''}$, which in turn is bounded from below by $-\frac{1}{2}\log_2 (\varsigma(\varepsilon,k))$ as stated in Proposition~\ref{prop:unext_ent_hyp_test_lb}. That is,
\begin{equation}
    E^{u,\varepsilon}_{\min}\!\left(\rho_{AB}\right) \ge E^{u,\varepsilon}_{\min}\!\left(\sigma_{A'B'A''B''}\right) \ge -\frac{1}{2}\log_2(\varsigma(\varepsilon,k)).
\end{equation}
Rewriting the above inequality as an upper bound on $\log_2 k$, which is the maximum number of secret bits that can be distilled from $\rho_{AB}$ using a two-extendible channel, yields an upper bound on the one-shot, one-way distillable key of the state $\rho_{AB}$. 
\begin{theorem}[Unextendibility bound on one-shot distillable key] \label{theo:distill_key_st_ub_hypo_test}
    Fix $\varepsilon \in (0,1)$. Let $\rho_{AB}$ be a quantum state  such that the following inequality holds:
    \begin{equation}
        \varepsilon < J^{\varepsilon}_{\min}\!\!\left(\rho_{AB}\right) \le \frac{1}{4} + \frac{\varepsilon}{2} + \frac{\sqrt{3\varepsilon(1-\varepsilon)}}{2},
    \end{equation}
   where $J^{\varepsilon}_{\min}\!\left(\rho_{AB}\right)$ is defined in~\eqref{eq:J_hypo_test_defn}. Then the one-shot, one-way distillable key of a bipartite state $\rho_{AB}$ is bounded from above by the following quantity:
	\begin{equation}\label{eq:dist_key_hypo_test_ub}
		K^{\varepsilon,\to}_D\!\left(\rho_{AB}\right) \le \frac{1}{2}\log_2\!\left[\left(\frac{\sqrt{J^{\varepsilon}_{\min}\!\left(\rho_{AB}\right)\!\left(1-J^{\varepsilon}_{\min}\!\left(\rho_{AB}\right)\right)}+\sqrt{\varepsilon(1-\varepsilon)}}{J^{\varepsilon}_{\min}\!\left(\rho_{AB}\right)-\varepsilon}\right)^2+1\right].
	\end{equation}
 If $J^{\varepsilon}_{\min}(\rho_{AB}) > \frac{1}{4} + \frac{\varepsilon}{2} + \frac{\sqrt{3\varepsilon(1-\varepsilon)}}{2}$, then the one-shot, one-way distillable key of the state is equal to zero.
\end{theorem}
\begin{proof}
	See Appendix~\ref{app:distill_key_st_ub_hypo_test}.
\end{proof}

\medskip

When $\varepsilon = 0$, the upper bound from Theorem~\ref{theo:distill_key_st_ub_hypo_test} simplifies to the min-unextendible entanglement bound on the exact one-way distillable key of a state obtained in~\cite[Theorem~24]{WWW24}.

The one-shot, one-way distillable key of a state does not increase under the action of one-way LOCC channels, and we expect the same from any reasonable bound on the quantity. Similarly, we expect that the bound does not increase with decreasing $\varepsilon$ because demanding a higher fidelity between the distilled state and the target state should only lead to a lower yield from the distillation process. To examine if the upper bound on the one-shot, one-way distillable key of a state given in Theorem~\ref{theo:distill_key_st_ub_hypo_test} satisfies the aforementioned criteria, we invoke Lemma~\ref{lem:monotoncity_obj_func} stated below.
\begin{lemma}\label{lem:monotoncity_obj_func}
    For all $J \in (\varepsilon,1-\varepsilon]$ and $ \varepsilon \in [0,1]$, the following function
    \begin{equation}
        f(J,\varepsilon) \coloneqq \frac{1}{2}\log_2\!\left[\left(\frac{\sqrt{J(1-J)}+\sqrt{\varepsilon(1-\varepsilon)}}{J-\varepsilon}\right)^2+1\right]
    \end{equation}
    decreases monotonically with $J$ and increases monotonically with $\varepsilon$.
\end{lemma}
\begin{proof}
    See Appendix~\ref{app:monotonicity_obj_func}.
\end{proof}

\medskip

Recall that the quantity $J^{\varepsilon}_{\min}(\rho_{AB})$ increases under the action of one-way LOCC channels. Therefore, Lemma~\ref{lem:monotoncity_obj_func} implies that the upper bound on the one-shot, one-way distillable key of a state given in Theorem~\ref{theo:distill_key_st_ub_hypo_test} decreases monotonically under the action of one-way LOCC channels, and it increases monotonically with $\varepsilon$.

The smooth-min unextendible entanglement of a state can be written as a semidefinite program (see Appendix~\ref{app:semidefinite_programs}), which allows us to compute the upper bound on the one-shot, one-way distillable key given in Theorem~\ref{theo:distill_key_st_ub_hypo_test} using a semidefinite program. In Figure~\ref{fig:iso_state_dist_key} we demonstrate some numerical results for the smooth-min unextendible entanglement upper bound on the one-shot, one-way distillable key of an isotropic state~\cite{HH99}, which is parameterized as follows:
\begin{equation}\label{eq:iso_state_parameterization}
    \zeta^{F,d}_{AB} = F\Phi^{d}_{AB} + (1-F)\frac{I_{AB}-\Phi^d_{AB}}{d^2-1},
\end{equation}
where $F$ is a parameter in the interval $[0,1]$ and $d=d_A=d_B$.

\begin{figure}
    \centering
    \begin{subfigure}{0.4\textwidth}
        \includegraphics[width=\linewidth]{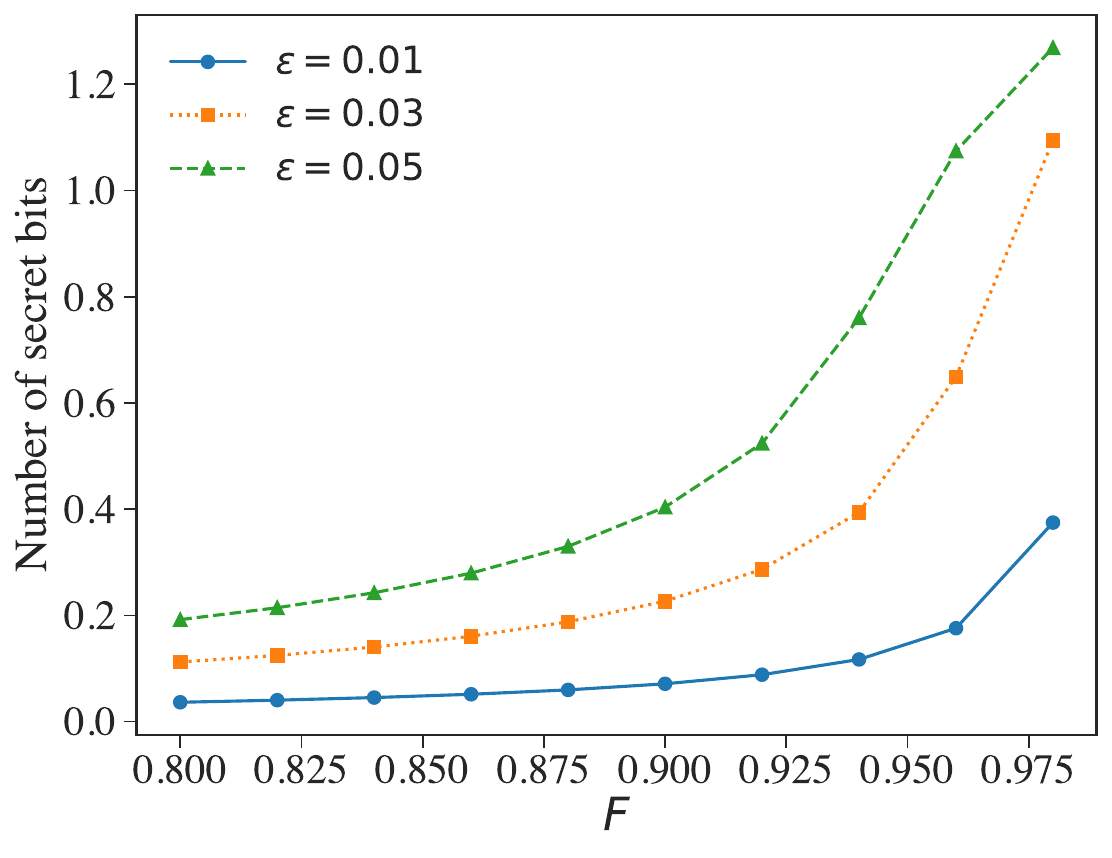}
        \caption{\centering\label{fig:iso_dist_key_2d}}
    \end{subfigure}
    \begin{subfigure}{0.4\textwidth}
        \includegraphics[width=\linewidth]{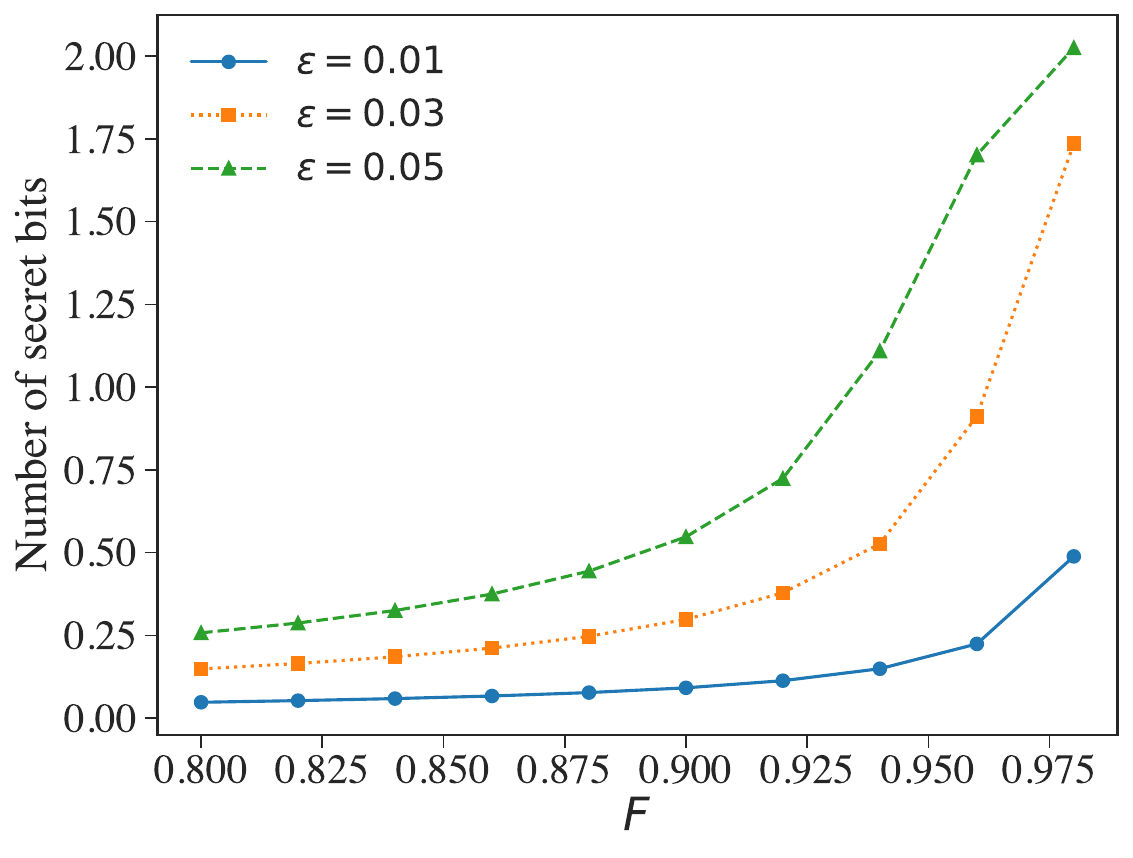}
        \caption{\centering\label{fig:iso_dist_key_3d}}
    \end{subfigure}
    \begin{subfigure}{0.4\textwidth}
        \includegraphics[width=\linewidth]{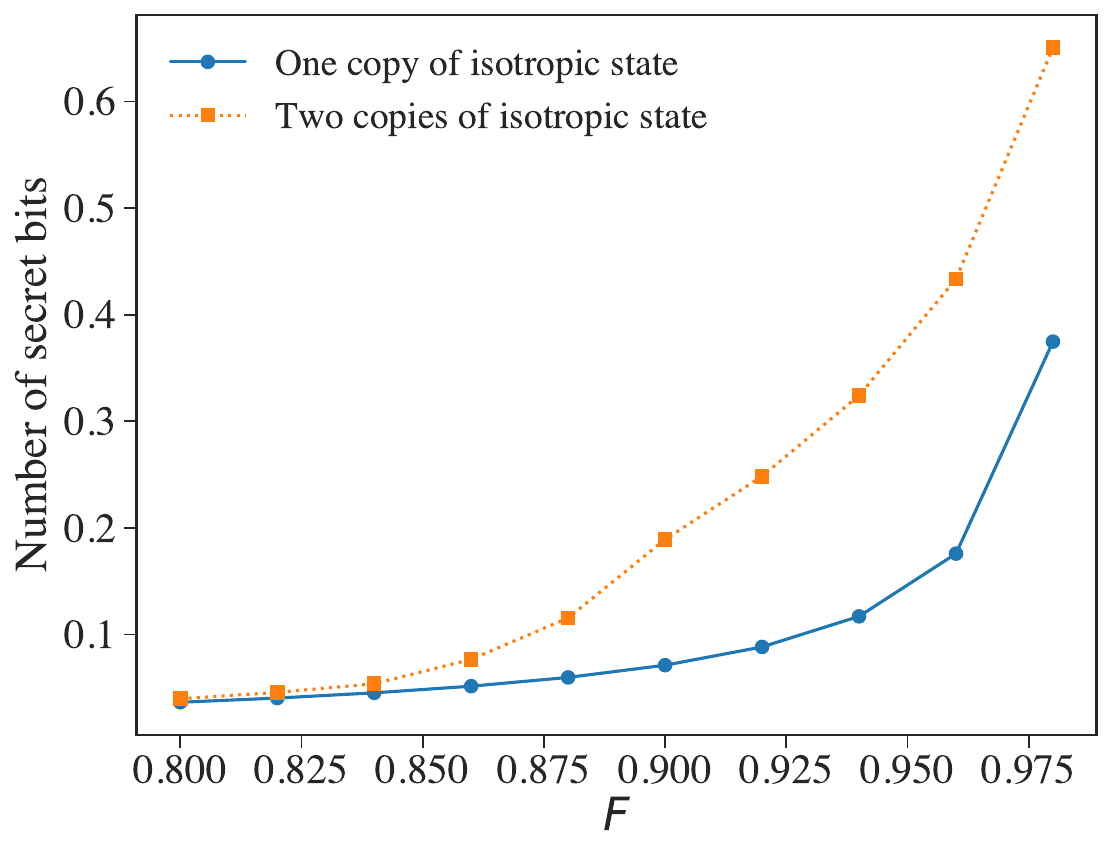}
        \caption{\centering\label{fig:iso_dist_key_1v2_qubit}}
    \end{subfigure}
    \caption{Upper bounds on the one-shot one-way distillable key of an isotropic state, described in~\eqref{eq:iso_state_parameterization}, using~\eqref{eq:dist_key_hypo_test_ub}. (a) Upper bounds for a two-dimensional isotropic state for different values of $\varepsilon$ plotted against the parameter $F$. (b) Upper bounds for a three-dimensional isotropic state for different values of $\varepsilon$ plotted against the parameter $F$. (c) Comparison between the upper bounds obtained for a single copy of a two-dimensional isotropic state with the upper bound obtained for two copies of a two-dimensional isotropic state, plotted against the parameter $F$.}
    \label{fig:iso_state_dist_key}
\end{figure}

\subsubsection{Simplified upper bounds}

The upper bound on the one-shot, one-way distillable key of a bipartite state obtained in Theorem~\ref{theo:distill_key_st_ub_hypo_test} is difficult to analyze, due to its complicated form. Weaker but simpler bounds can be obtained by using the smooth-min unextendible entanglement of states.

\textbf{Relaxation 1:} We first consider a relaxation of the upper bound in~\eqref{eq:dist_key_hypo_test_ub} by finding an algebraic inequality for the function $\varsigma(\varepsilon,k)$ defined in~\eqref{eq:varsig-func}:
\begin{align}
    \varsigma(\varepsilon,k) &= \varepsilon + \frac{1-2\varepsilon}{k^2} + \frac{2\sqrt{(k^2-1)\varepsilon(1-\varepsilon)}}{k^2}\\
    &\le \frac{1}{k^2} + \varepsilon\left(1-\frac{2}{k^2}\right) + 2\frac{\sqrt{k^2\varepsilon}}{k^2}\\
    &\le \frac{1}{k^2} + \varepsilon + 2\frac{\sqrt{\varepsilon}}{k}\\
    &= \left(\frac{1}{k} + \sqrt{\varepsilon}\right)^2,
\end{align}
where the first inequality follows from the fact that $(k^2-1)(1-\varepsilon)\le k^2$ for all $k\in \mathbb{N}$ and $\varepsilon \in [0,1]$ and the second inequality follows from the fact that $1-\frac{2}{k^2} < 1$ for all $k\in \mathbb{N}$.

We can now use this upper bound on $\varsigma(\varepsilon,k)$ along with the statement of Proposition~\ref{prop:unext_ent_hyp_test_lb} to obtain a lower bound on the smooth-min unextendible entanglement of a state $\sigma_{ABA'B'}$ such that $F(\sigma_{ABA'B'},\gamma^k_{ABA'B'})\ge 1-\varepsilon$ for some private state $\gamma^k_{ABA'B'}$. In particular,
\begin{equation}\label{eq:smooth_min_unext_ent_alg_lb}
    E^{u,\varepsilon}_{\min}\!\left(\sigma_{ABA'B'}\right) \ge -\frac{1}{2}\log_2\!\left(\frac{1}{k} + \sqrt{\varepsilon}\right)^2 = -\log_2\!\left(\frac{1}{k} + \sqrt{\varepsilon}\right).
\end{equation}
This leads to the following simplified but weaker upper bound on the one-shot, one-way distillable key of a bipartite state:
\begin{corollary}\label{cor:dist_key_st_alg_ub}
        Fix $\varepsilon \in (0,1)$. Let $\rho_{AB}$ be a quantum state  such that the following inequality holds:
    \begin{equation}
        \frac{1}{4} + \frac{\varepsilon}{2} + \frac{\sqrt{3\varepsilon(1-\varepsilon)}}{2} \ge J^{\varepsilon}_{\min}(\rho_{AB})  > \varepsilon,
    \end{equation}
    where $J^{\varepsilon}_{\min}$ is defined in~\eqref{eq:J_hypo_test_defn}. Then the one-shot, one-way distillable key of a bipartite state $\rho_{AB}$ is bounded from above as follows:
	\begin{equation}\label{eq:dist_key_hypo_test_alg_ub}
		K^{\varepsilon,\to}_D\!\left(\rho_{AB}\right) \le -\log_2\!\left(\sqrt{J^{\varepsilon}_{\min}(\rho_{AB})} - \sqrt{\varepsilon}\right).
	\end{equation}
\end{corollary}

\begin{proof}
    The proof is similar to the proof of Theorem~\ref{theo:distill_key_st_ub_hypo_test} and follows directly from~\eqref{eq:smooth_min_unext_ent_alg_lb}.
\end{proof}

\medskip

\textbf{Relaxation 2:} Another way to obtain a simpler bound on the one-shot, one-way distillable key of a state is by considering the trace norm. First we will find a simpler but weaker statement of Lemma~\ref{lem:priv_marg_phi_fid_bnd}. Consider a state $\sigma_{ABA'B'}$ such that $F(\gamma^k_{ABA'B'},\sigma_{ABA'B'})\ge 1-\varepsilon$ for some private state $\gamma^k_{ABA'B'}$. Combining~\eqref{eq:twirled_priv_marg} and~\eqref{eq:fid_twirl_priv_max_mix}, we find that the following inequality holds for any state $\omega_{AEA'E'} \in \mathcal{F}(\sigma_{ABA'B'})$:
\begin{equation}
    1-\varepsilon \le F\!\left(\pi_{AE}, q~\Phi^k_{AE} + (1-q)\frac{I_{AE}-\Phi^k_{AE}}{k^2-1}\right),
\end{equation}
where $q = \operatorname{Tr}\!\left[\Pi^{\gamma}_{AEA'E'}\omega_{AEA'E'}\right]$ as stated in~\eqref{eq:q_priv_test_interpret}. Using the Fuchs-van de Graaf inequality~\cite{FvdG98} (specifically, $F(\rho,\sigma) \leq 1-\frac{1}{4} \left \| \rho  - \sigma \right \|_1^2$), we arrive at the following inequality:
\begin{equation}
    1-\varepsilon \le F\!\left(\pi_{AE}, q~\Phi^k_{AE} + (1-q)\frac{I_{AE}-\Phi^k_{AE}}{k^2-1}\right) \le 1-\frac{1}{4}\left\Vert \pi_{AE} - q~\Phi^k_{AE} + (1-q)\frac{I_{AE}-\Phi^k_{AE}}{k^2-1}\right\Vert^2_1.
\end{equation}
The above inequality can be rewritten as follows:
\begin{align}
    \sqrt{\varepsilon} &\ge \frac{1}{2}\left\Vert \pi_{AE} - q~\Phi^k_{AE} + (1-q)\frac{I_{AE}-\Phi^k_{AE}}{k^2-1}\right\Vert_1\\
    &= \frac{1}{2}\left\Vert \frac{1}{k^2}\Phi^k_{AE} - \left(1-\frac{1}{k^2}\right)\frac{I_{AE}-\Phi^k_{AE}}{k^2-1} - q~\Phi^k_{AE} + (1-q)\frac{I_{AE}-\Phi^k_{AE}}{k^2-1}\right\Vert_1\\
    &= \frac{1}{2}\left\Vert \left(\frac{1}{k^2}-q\right)\Phi^k_{AE} + \left(\frac{1}{k^2}-q\right)\frac{I_{AE}-\Phi^k_{AE}}{k^2-1}\right\Vert_1\\
    &= \left|\frac{1}{k^2}-q\right|\left\Vert\frac{1}{2}\Phi^k_{AE} + \frac{1}{2}\frac{I_{AE}-\Phi^k_{AE}}{k^2-1} \right\Vert_1\\
    &= \left|\frac{1}{k^2}-q\right|,\label{eq:q_trace_dist_ineq_step_1}
\end{align}
where the first equality follows by writing the maximally mixed state as an isotropic state and the last equality follows by realizing that $\frac{1}{2}\Phi^k_{AE} + \frac{1}{2}\frac{I_{AE}-\Phi^k_{AE}}{k^2-1}$ is a quantum state for which the trace norm is equal to 1. The inequality in~\eqref{eq:q_trace_dist_ineq_step_1} is satisfied if and only if $q$ lies in the following range:
\begin{equation}
    \frac{1}{k^2}-\sqrt{\varepsilon} \le q \le \frac{1}{k^2} + \sqrt{\varepsilon}.
\end{equation}
Since $q = \operatorname{Tr}\!\left[\Pi^{\gamma}_{ABA'B'}\omega_{ABA'B'}\right]$, we have the following inequality:
\begin{equation}\label{eq:priv_test_pass_trace_dist}
    \operatorname{Tr}\!\left[\Pi^{\gamma}_{ABA'B'}\omega_{ABA'B'}\right] \le \frac{1}{k^2} + \sqrt{\varepsilon}.
\end{equation}

Using the inequality in~\eqref{eq:priv_test_pass_trace_dist} and the arguments used in the proof of Proposition~\ref{prop:unext_ent_hyp_test_lb}, we can bound the smooth-min unextendible entanglement of the state $\sigma_{ABA'B'}$ as follows:
\begin{equation}\label{eq:hypot_test_unext_ent_td_lb}
    E^{u,\varepsilon}_{\min}\!\left(\sigma_{ABA'B'}\right) \ge -\frac{1}{2}\log_2\!\left(\frac{1}{k^2} + \sqrt{\varepsilon}\right).
\end{equation}
A relaxed upper bound on the one-shot, one-way distillable key of a bipartite state can be found by using the inequality mentioned above, which we state as Corollary~\ref{cor:dist_key_st_td_ub}.
\begin{corollary}\label{cor:dist_key_st_td_ub}
    Fix $\varepsilon \in (0,1)$. Let $\rho_{AB}$ be a quantum state  such that the following inequality holds:
    \begin{equation}
        \frac{1}{4} + \frac{\varepsilon}{2} + \frac{\sqrt{3\varepsilon(1-\varepsilon)}}{2} \ge J^{\varepsilon}_{\min}(\rho_{AB})  > \sqrt{\varepsilon},
    \end{equation}
    where $J^{\varepsilon}_{\min}$ is defined in~\eqref{eq:J_hypo_test_defn}. Then the one-shot, one-way distillable key of a bipartite state $\rho_{AB}$ is bounded from above by the following quantity:
	\begin{equation}\label{eq:dist_key_hypo_test_td_ub}
		K^{\varepsilon,\to}_D\!\left(\rho_{AB}\right) \le -\frac{1}{2}\log_2\!\left(J^{\varepsilon}_{\min}(\rho_{AB}) - \sqrt{\varepsilon}\right).
	\end{equation}
\end{corollary}
\begin{proof}
    The proof is similar to the proof of Theorem~\ref{theo:distill_key_st_ub_hypo_test} and follows directly from~\eqref{eq:hypot_test_unext_ent_td_lb}. 
\end{proof}

\begin{remark}
    The upper bound on the one-shot, one-way distillable key given in Corollary~\ref{cor:dist_key_st_alg_ub} is tighter than the upper bound given in Corollary~\ref{cor:dist_key_st_td_ub} for highly entangled states, but the order is reversed for weakly entangled states. To be precise, 
    \begin{equation}
        K^{\varepsilon,\to}_D\!\left(\rho_{AB}\right) \le -\log_2\!\left(\sqrt{J^{\varepsilon}_{\min}} - \sqrt{\varepsilon}\right) \le -\frac{1}{2}\log_2\!\left(J^{\varepsilon}_{\min} - \sqrt{\varepsilon}\right) ,
    \end{equation}
    for all  $J^{\varepsilon}_{\min} \in \left[\sqrt{\varepsilon},\frac{1}{4}\!\left(1 + \varepsilon + 2\sqrt{\varepsilon}\right)\right]$,
    and
    \begin{equation}
        K^{\varepsilon,\to}_D\!\left(\rho_{AB}\right)  \le -\frac{1}{2}\log_2\!\left(J^{\varepsilon}_{\min} - \sqrt{\varepsilon}\right)\le -\log_2\!\left(\sqrt{J^{\varepsilon}_{\min}} - \sqrt{\varepsilon}\right) ,
    \end{equation}
    for all  $J^{\varepsilon}_{\min} \in \left[\frac{1}{4}\!\left(1 + \varepsilon + 2\sqrt{\varepsilon}\right), 1-\varepsilon\right]$.
\end{remark}

The bounds obtained in Corollary~\ref{cor:dist_key_st_alg_ub} and Corollary~\ref{cor:dist_key_st_td_ub} can be computed using a semidefinite program. In Figure~\ref{fig:isotropic_st_relaxed_bnd} we plot the three different upper bounds on the one-shot, one-way distillable key of isotropic states, defined in~\eqref{eq:iso_state_parameterization}, obtained from Theorem~\ref{theo:distill_key_st_ub_hypo_test}, Corollary~\ref{cor:dist_key_st_alg_ub}, and Corollary~\ref{cor:dist_key_st_td_ub}. Notice that the bound from Corollary~\ref{cor:dist_key_st_alg_ub} is tighter than the bound from Corollary~\ref{cor:dist_key_st_td_ub} when the resource state is highly entangled.

\begin{figure}
    \centering
    \begin{subfigure}{0.4\linewidth}
       \includegraphics[width=\linewidth]{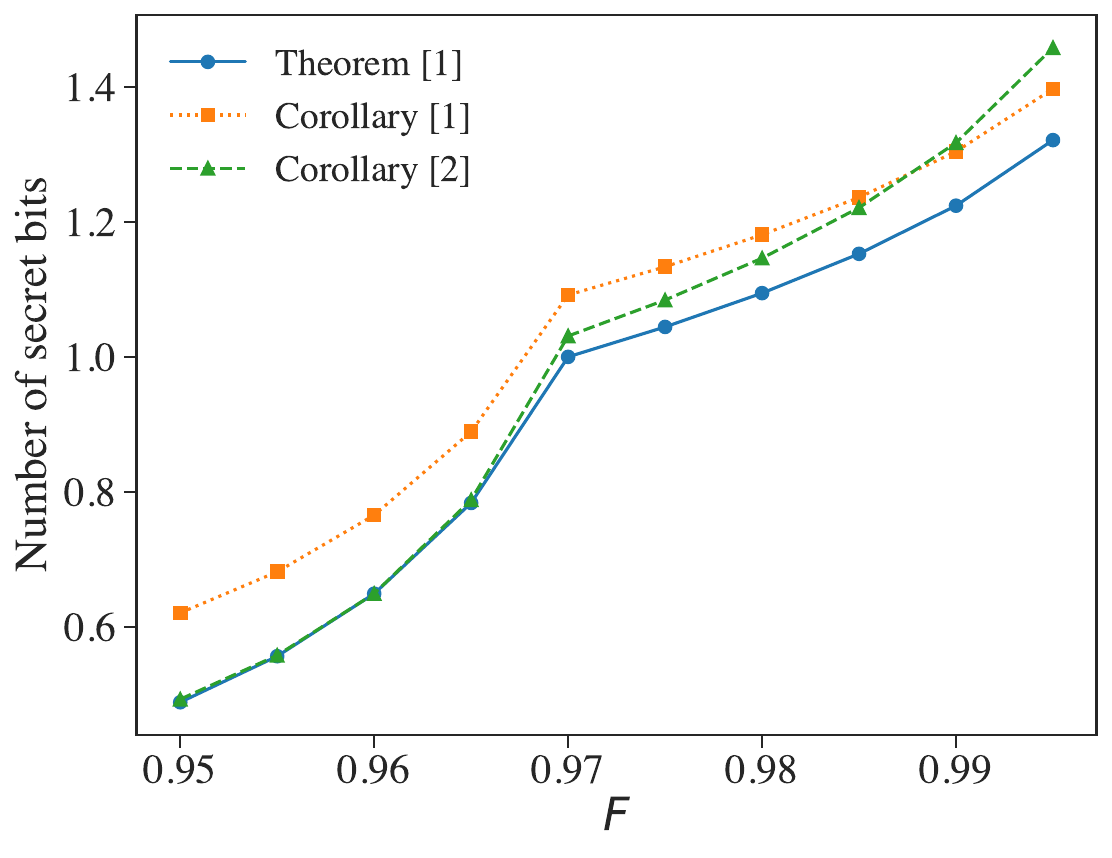} 
       \caption{\centering $\varepsilon = 0.03$}
    \end{subfigure}
    \begin{subfigure}{0.4\linewidth}
        \includegraphics[width = \linewidth]{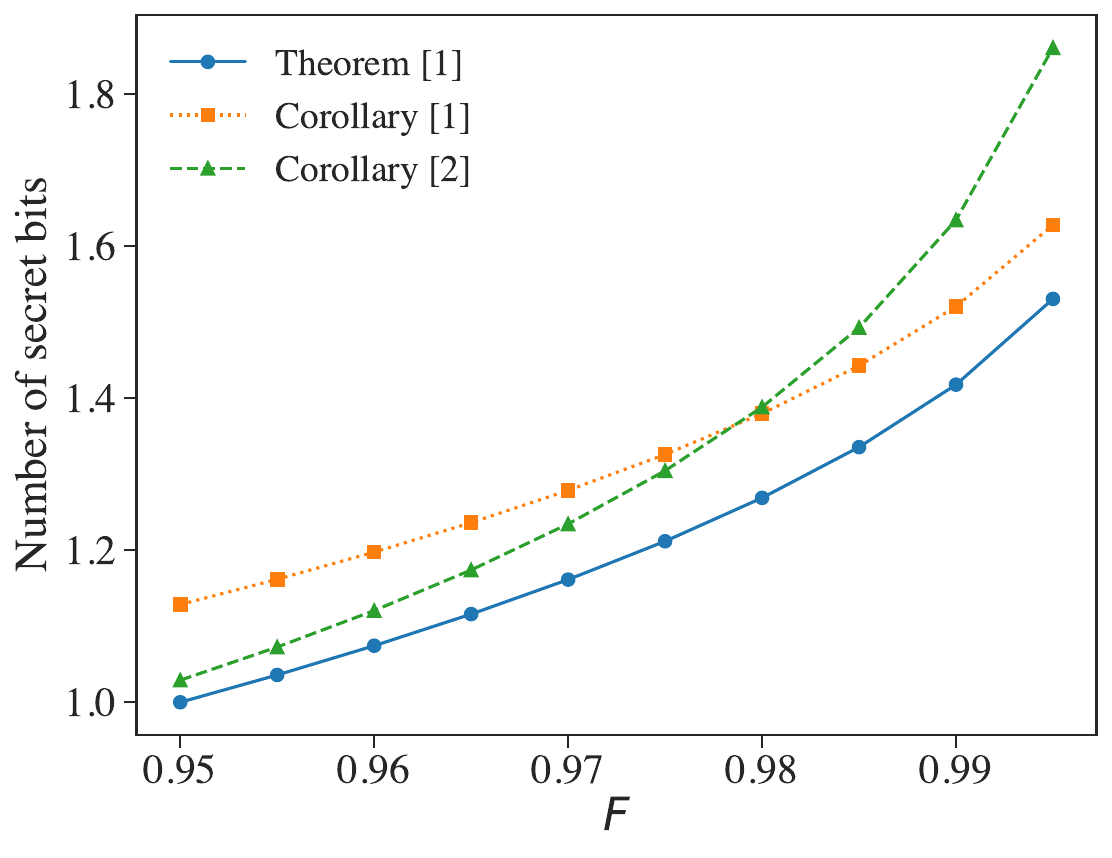}
        \caption{\centering $\varepsilon = 0.05$}
    \end{subfigure}
    
    \caption{Comparison between the upper bounds on the one-shot, one-way distillable key of two-dimensional isotropic states, parameterized as given in~\eqref{eq:iso_state_parameterization}, obtained using Theorem~\ref{theo:distill_key_st_ub_hypo_test}, Corollary~\ref{cor:dist_key_st_alg_ub}, and Corollary~\ref{cor:dist_key_st_td_ub}.}
    \label{fig:isotropic_st_relaxed_bnd}
\end{figure}

\subsection{One-way secret-key distillation from i.i.d.~copies of a state}\label{sec:sandwich_st_ub}

Resource distillation from independent and identically distributed (i.i.d.) copies of a state is often considered a physically relevant setting. As such, the rate at which secret bits can be distilled from $n$~i.i.d.~copies of a state, which is equal to $\frac{1}{n}K^{\varepsilon,\to}_D\!\left(\rho^{\otimes n}_{AB}\right)$, is an important quantity from both information-theoretic and practical perspectives. 

In principle, the bounds obtained for the one-shot, one-way distillable key of a state can be used to obtain an upper bound on the rate of one-way distillable key rate from $n$ copies of $\rho_{AB}$, with $\varepsilon$ error tolerance, by simply calculating the upper bounds for the state $\rho^{\otimes n}_{AB}$. However, the complexity of computing these bounds scales exponentially with the number of copies $n$, rendering the computation of these bounds intractable for large enough $n$. The smooth-min relative entropy is not subadditive; that is, there exists a choice of states $\rho^1,  \rho^2, \sigma^2$, and $\sigma^2$, such that the following inequality does not hold:
\begin{equation}
    D^{\varepsilon}_{\min}\!\left(\rho^2\otimes \rho^2\middle\Vert\sigma^1\otimes \sigma^2\right) \le D^{\varepsilon}_{\min}\!\left(\rho^1\middle\Vert\sigma^1\right) + D^{\varepsilon}_{\min}\!\left( \rho^2\middle\Vert\sigma^2\right),
\end{equation}
which makes it difficult to obtain a single-letter bound on the one-way distillable key rate with our approach. We turn to the $\alpha$-sandwiched unextendible entanglement of bipartite states to address this problem.

The smooth-min relative entropy is related to the $\alpha$-sandwiched R\'enyi relative entropy by the following inequality~\cite[Lemma~5]{CMW16}:
\begin{equation}\label{eq:smooth_min_vs_sandwich_ineq}
    D^{\varepsilon}_{\min}\!\left(\rho\Vert\sigma\right) \le \widetilde{D}_{\alpha}\!\left(\rho\Vert\sigma\right) + \frac{\alpha}{\alpha -1}\log_2\!\left(\frac{1}{1-\varepsilon}\right)
\end{equation}
for all $\varepsilon \in [0,1)$ and $\alpha \in (1,\infty)$. By taking an infimum over all states $\sigma \in \mathcal{F}(\rho)$, we arrive at an inequality relating the smooth-min unextendible entanglement of a state and the $\alpha$-sandwiched unextendible entanglement of the state, which we state in Proposition~\ref{prop:smooth_min_vs_sandwich_unext_ent} below.
\begin{proposition}\label{prop:smooth_min_vs_sandwich_unext_ent}
    Let $\alpha \in (1,\infty)$ and $\varepsilon \in [0,1)$. Then the following inequality holds between the smooth-min unextendible entanglement of a state and the $\alpha$-sandwiched unextendible entanglement of the state:
    \begin{equation}
        E^{u,\varepsilon}_{\min}\!\left(\rho\right) \le \widetilde{E}^u_{\alpha}\!\left(\rho\right) + \frac{1}{2}\cdot\frac{\alpha}{\alpha -1}\log_2\!\left(\frac{1}{1-\varepsilon}\right).
    \end{equation}
\end{proposition}

As a counterpart of $J^{\varepsilon}_{\min}(\rho_{AB})$, we define the following quantity for mathematical simplicity:
\begin{equation}
    \widetilde{J}^{\varepsilon}_{\alpha}\!\left(\rho_{AB}\right) \coloneqq 2^{-2\widetilde{E}^u_{\alpha}(\rho_{AB})}\left(1-\varepsilon\right)^{\frac{\alpha}{\alpha - 1}} \qquad \forall \alpha \in (1,\infty).
\end{equation}
It is straightforward to verify from Proposition~\ref{prop:smooth_min_vs_sandwich_unext_ent} that
\begin{equation}\label{eq:J_eps_J_alpha_reln}
    J^{\varepsilon}_{\min}\!\left(\rho_{AB}\right) \ge \widetilde{J}^{\varepsilon}_{\alpha}\!\left(\rho_{AB}\right) \qquad \forall \alpha \in (1,\infty).
\end{equation}
One can simply use Theorem~\ref{theo:distill_key_st_ub_hypo_test} and Lemma~\ref{lem:monotoncity_obj_func} to obtain an upper bound on the one-shot, one-way distillable key of a state in terms of the $\alpha$-sandwiched unextendible entanglement, as follows:
\begin{equation}\label{eq:1shot_1w_1copy_sandwich_bnd}
    K^{\varepsilon,\to}_D\!\left(\rho_{AB}\right) \le f\!\left(J^{\varepsilon}_{\min}\!\left(\rho_{AB}\right),\varepsilon\right) \le f\!\left(\widetilde{J}^{\varepsilon}_{\alpha}\!\left(\rho_{AB}\right),\varepsilon\right) \qquad \forall \alpha \in (1,\infty),
\end{equation}
where the function $f$ is defined in Lemma~\ref{lem:monotoncity_obj_func}. The first inequality follows from Theorem~\ref{theo:distill_key_st_ub_hypo_test}, and the second inequality follows from~\eqref{eq:J_eps_J_alpha_reln} and Lemma~\ref{lem:monotoncity_obj_func}.

While the $\alpha$-sandwiched unextendible entanglement bound is clearly worse than the smooth-min unextendible entanglement bound, it gives a single-letter upper bound on the one-shot, one-way distillable key from $n$ i.i.d.~copies of a state $\rho_{AB}$. The subadditivity of the $\alpha$-sandwiched unextendible entanglement implies the following for all $\alpha \in (1,\infty)$:
\begin{equation}
    \widetilde{J}^{\varepsilon}_{\alpha}\!\left(\rho^{\otimes n}_{AB}\right) =  2^{-2\widetilde{E}^u_{\alpha}(\rho^{\otimes n}_{AB})}\left(1-\varepsilon\right)^{\frac{\alpha}{\alpha - 1}} \ge 2^{-2n\widetilde{E}^u_{\alpha}(\rho_{AB})}\left(1-\varepsilon\right)^{\frac{\alpha}{\alpha - 1}}.
\end{equation}
Let us define the following quantity:
\begin{equation}\label{eq:J_sandwich_defn}
    \widetilde{J}^{\varepsilon,n}_{\alpha}\!\left(\rho_{AB}\right) \coloneqq 2^{-2n\widetilde{E}^u_{\alpha}(\rho_{AB})}\left(1-\varepsilon\right)^{\frac{\alpha}{\alpha - 1}}.
\end{equation}
Then the fact that $\widetilde{J}^{\varepsilon}_{\alpha}\left(\rho^{\otimes n}_{AB}\right)\ge \widetilde{J}^{\varepsilon,n}_{\alpha}\left(\rho_{AB}\right)$, combined with Lemma~\ref{lem:monotoncity_obj_func} and~\eqref{eq:1shot_1w_1copy_sandwich_bnd}, leads to a single-letter upper bound on the one-shot, one-way distillable key of $n$ i.i.d.~copies of a state, which we state formally in Corollary~\ref{cor:dist_key_sandwich_ub_n_copies}.
\begin{corollary}\label{cor:dist_key_sandwich_ub_n_copies}
    Fix $\varepsilon \in (0,1)$ and $\alpha \in (1,\infty)$. Let $\rho_{AB}$ be a quantum state  such that the following inequality holds:
    \begin{equation}
        \varepsilon < \widetilde{J}^{\varepsilon,n}_{\alpha}\!\left(\rho_{AB}\right) \le \frac{1}{4} + \frac{\varepsilon}{2} + \frac{\sqrt{3\varepsilon(1-\varepsilon)}}{2},
    \end{equation}
   where $\widetilde{J}^{\varepsilon,n}_{\alpha}\left(\rho_{AB}\right)$ is defined in~\eqref{eq:J_sandwich_defn}. Then the $n$-shot, one-way distillable key of a state $\rho_{AB}$ is bounded from above by the following quantity:
	\begin{equation}\label{eq:dist_key_sandwich_ub_n_copies}
		K^{\varepsilon,\to}_D\!\left(\rho^{\otimes n}_{AB}\right) \le \frac{1}{2}\log_2\!\left[\left(\frac{\sqrt{\widetilde{J}^{\varepsilon,n}_{\alpha}\!\left(\rho_{AB}\right)\!\left(1-\widetilde{J}^{\varepsilon,n}_{\alpha}\!\left(\rho_{AB}\right)\right)}+\sqrt{\varepsilon(1-\varepsilon)}}{\widetilde{J}^{\varepsilon,n}_{\alpha}\!\left(\rho_{AB}\right)-\varepsilon}\right)^2+1\right],
	\end{equation}
  If $\widetilde{J}^{\varepsilon,n}_{\alpha}\!\left(\rho_{AB}\right) > \frac{1}{4} + \frac{\varepsilon}{2} + \frac{\sqrt{3\varepsilon(1-\varepsilon)}}{2}$, then the $n$-shot, one-way distillable key of the state is equal to zero.
\end{corollary}

A special case of the bound stated above arises when $\alpha \to \infty$. Let us define the following quantity:
\begin{align}
    J^{\varepsilon}_{\max}\!\left(\rho_{AB}\right) &\coloneqq \lim_{\alpha \to \infty}\widetilde{J}^{\varepsilon}_{\alpha}\!\left(\rho_{AB}\right)\\
    &= \left(\lim_{\alpha \to \infty}2^{-2\widetilde{E}^u_{\alpha}(\rho_{AB})}\right)\!\left(\lim_{\alpha \to \infty}\left(1-\varepsilon\right)^{\frac{\alpha}{\alpha - 1}}\right)\\
    &= 2^{-2E^u_{\max}(\rho_{AB})}(1-\varepsilon),
\end{align}
where the second equality follows from the definition of $\widetilde{J}^{\varepsilon}_{\alpha}(\rho_{AB})$ and the last equality follows from~\eqref{eq:Emax_eq_lim_inf_sandwiched}. The max-unextendible entanglement of a state can be computed using an SDP, which implies that $J^{\varepsilon}_{\max}(\rho_{AB})$ can be computed using an SDP. Thus, setting $\alpha\to \infty$ in Corollary~\ref{cor:dist_key_sandwich_ub_n_copies} leads to a single-letter, computable bound on the $n$-shot, one-way distillable key of a state. 

In Figure~\ref{fig:Emax_st_bounds} we plot the upper bounds on the $n$-shot, one-way distillable key of isotropic states calculated using Corollary~\ref{cor:dist_key_sandwich_ub_n_copies} with $\alpha \to \infty$. In Figure~\ref{fig:Emax_d_2} we plot the upper bounds for two-dimensional isotropic states, and in Figure~\ref{fig:Emax_d_3} we plot the upper bounds for three-dimensional isotropic states, with $\varepsilon = 0.01$ in all the cases.

\begin{figure}
    \centering
    \begin{subfigure}{0.45\linewidth}
    \centering
        \includegraphics[width=\linewidth]{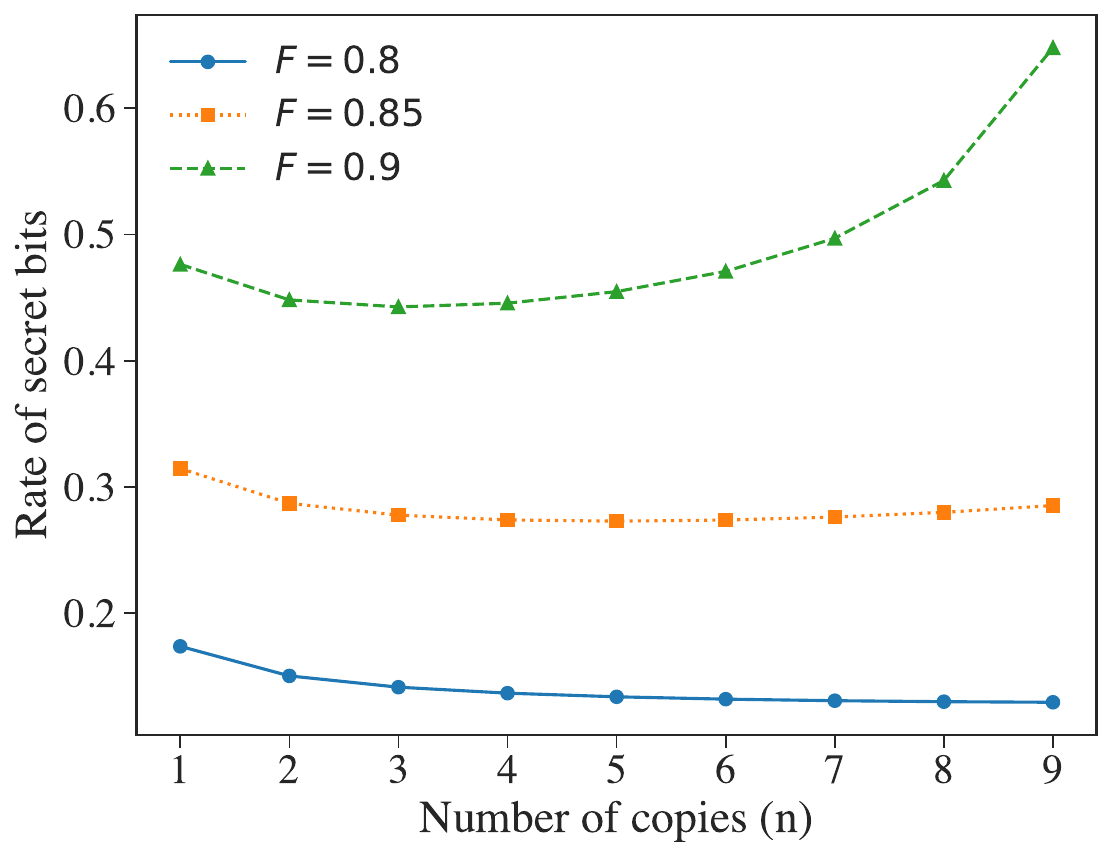}
    \caption{\centering Two-dimensional isotropic state}
    \label{fig:Emax_d_2}
    \end{subfigure}
    \begin{subfigure}{0.45\linewidth}
    \centering
        \includegraphics[width=\linewidth]{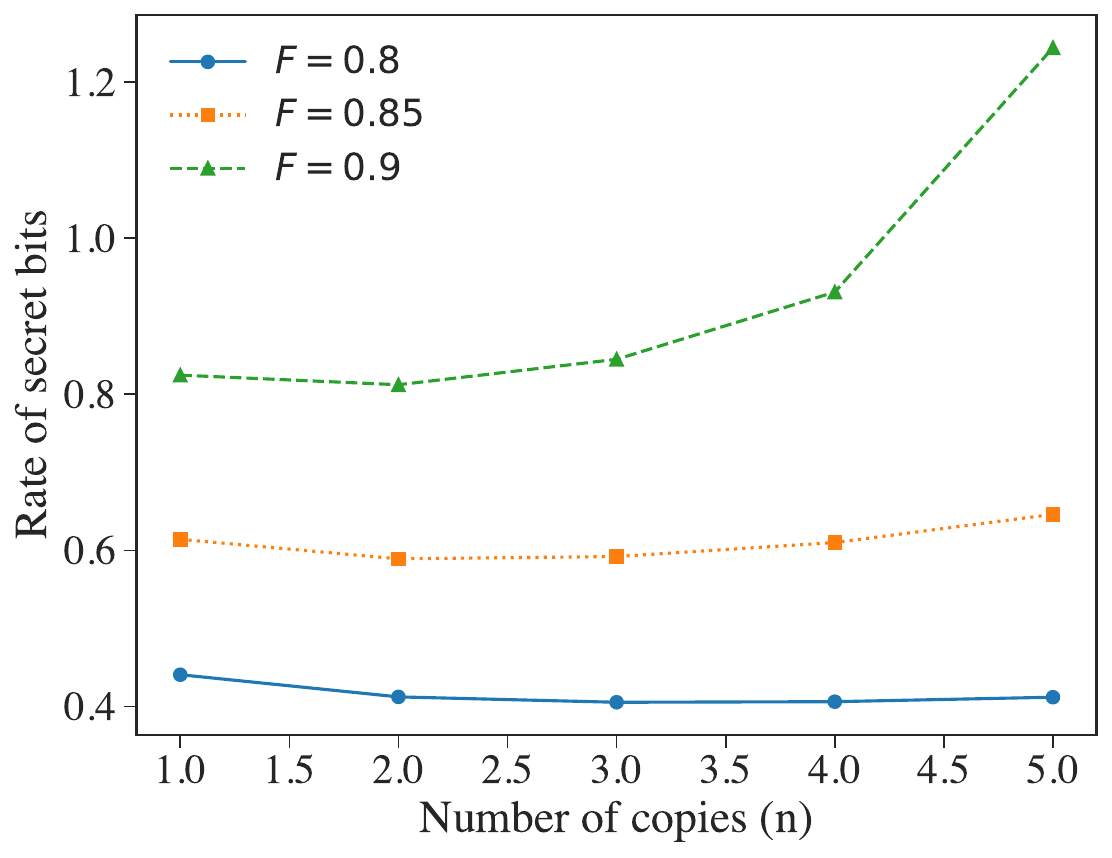}
    \caption{\centering Three-dimensional isotropic state}
    \label{fig:Emax_d_3}
    \end{subfigure}
    \caption{Upper bounds on the $n$-shot, one-way distillable-key rate of isotropic states using Corollary~\ref{cor:dist_key_sandwich_ub_n_copies} and setting $\alpha \to \infty$. The upper bounds are plotted for different values of the parameter $F$ of the isotropic state, with respect to the parameterization given in~\eqref{eq:iso_state_parameterization}, against the number of copies of the isotropic state, and $\varepsilon$ is set equal to $0.01$.  }
    \label{fig:Emax_st_bounds}
\end{figure}

\begin{remark}
    The techniques used to arrive at Corollaries~\ref{cor:dist_key_st_alg_ub} and~\ref{cor:dist_key_st_td_ub} can be used to find simpler bounds on the one-shot, one-way distillable key of $n$ i.i.d.~copies of a state $\rho_{AB}$ by using the $\alpha$-sandwiched unextendible entanglement. As such, for all $\varepsilon \in (0,1)$, $\alpha \in (1,\infty)$, $n\in \mathbb{N}$, and a state $\rho_{AB}$, if $\widetilde{J}^{\varepsilon,n}_{\alpha}\!\left(\rho_{AB}\right) > \varepsilon$, then
    \begin{equation}
        K^{\varepsilon,\to}_D\!\left(\rho^{\otimes n}_{AB}\right) \le -\log_2\!\left(\sqrt{\widetilde{J}^{\varepsilon,n}_{\alpha}\!\left(\rho_{AB}\right)} - \sqrt{\varepsilon}\right),
    \end{equation}
    and if $\widetilde{J}^{\varepsilon,n}_{\alpha}\!\left(\rho_{AB}\right) > \sqrt{\varepsilon}$, then
    \begin{equation}
        K^{\varepsilon,\to}_D\!\left(\rho^{\otimes n}_{AB}\right) \le -\frac{1}{2}\log_2\!\left(\widetilde{J}^{\varepsilon,n}_{\alpha}\!\left(\rho_{AB}\right) - \sqrt{\varepsilon}\right).
    \end{equation}
\end{remark}

\subsection{One-way secret-key distillation in the asymptotic setting}\label{sec:asymp_key_dist_states}

Now let us study the asymptotic setting of one-way secret-key distillation by using the framework of unextendibility.

In the asymptotic setting, we are interested in the maximum rate at which an arbitrarily large number of i.i.d.~copies of a state $\rho_{AB}$ can be transformed into a state that is arbitrarily close to an ideal secret key. In this setting, a one-way secret-key distillation protocol is given by a sequence of one-way LOCC channels $\left\{\mathcal{L}^{n,\to}_{A^nB^n\to A'B'A''B''}\right\}_{n\in \mathbb{N}}$, a sequence of bipartite private states $\left\{\gamma^{k_n}_{A'B'A''B''}\right\}_{n\in \mathbb{N}}$, and a sequence of real numbers $\{\varepsilon_n\}_{n\in \mathbb{N}}$ corresponding to the error in distillation. The joint system $A^n$ refers to $n$ systems, each of which are identical to the system $A$. The $n^{\operatorname{th}}$ element of this sequence acts on $n$ copies of the resource state $\rho_{AB}$ such that the infidelity between the output state and $\gamma^{k_n}_{A'B'A''B''}$ is less than or equal to $\varepsilon_n$. That is,
\begin{equation}
    F\!\left(\gamma^{k_n}_{A'B'A''B''},\mathcal{L}^{n,\to}_{A^nB^n\to A'B'A''B''}\!\left(\rho^{\otimes n}_{AB}\right)\right)\ge 1-\varepsilon_n \qquad \forall n\in \mathbb{N}.
\end{equation}
To achieve arbitrary precision in distilling secret keys, we demand that $\varepsilon_n \to 0$ as $n\to \infty$. The maximum achievable rate of distilling secret keys is then given by the one-way distillable key of the state, which can be mathematically defined in terms of the one-shot, one-way distillable key as follows~\cite[Chapter~15]{KW24}:
\begin{equation}
    K^{\to}_{D}\!\left(\rho_{AB}\right) \coloneqq \inf_{\varepsilon \in (0,1]}\liminf_{n\to \infty}\frac{1}{n}K^{\varepsilon,\to}_D\!\left(\rho^{\otimes n}_{AB}\right).
\end{equation}

The upper bounds on the one-shot, one-way distillable key, obtained in Sections~\ref{sec:smooth_min_st_ub} and~\ref{sec:sandwich_st_ub}, do not provide any new insight into the one-way distillable key of the state. This is because, for most states of interest, there exists an $n \in \mathbb{N}$ such that $J^{\varepsilon}_{\min}\!\left(\rho^{\otimes n}_{AB}\right) \le \varepsilon$ for any $\varepsilon \in (0,1]$, rendering the bound useless. 

However, consider a further restricted setting where the sequence $\{\varepsilon_n\}_{n\in \mathbb{N}}$ is required to decrease exponentially fast. The maximum rate of key distillation from an arbitrarily large number of copies of a resource state, for a fixed error exponent $a$, which we call the $a$-exponential one-way distillable key of a state, can be mathematically defined in the following manner.
\begin{definition}
    Fix $a>0$. We define the $a$-exponential one-way distillable key of a state $\rho_{AB}$ as follows:
    \begin{equation}
        K^{\to}_{D,a}\!\left(\rho_{AB}\right) \coloneqq \sup_{\substack{\left\{k_n\right\}_{n\in \mathbb{N}},\left\{\gamma^{k_n}_{A'B'A''B''}\right\}_{n\in \mathbb{N}}\\\left\{\mathcal{L}^{n,\to}_{A^nB^n\to A'B'A''B''}\right\}_{n\in \mathbb{N}} }} \liminf_{n\to \infty}\left\{\frac{\log_2 k_n}{n}: F\!\left(\gamma^{k_n}_{A'B'A''B''},\mathcal{L}^{\to}\!\left(\rho^{\otimes n}_{AB}\right)\right)\ge 1-2^{-an}\right\},
    \end{equation}
    where the supremum is over all sequences of bipartite states $\left\{\gamma^{k_n}_{ A'B'A''B''}\right\}_{n\in \mathbb{N}}$ and all sequences of one-way LOCC channels $\left\{\mathcal{L}^{n,\to}_{A^nB^n\to A'B'A''B''}\right\}_{n\in \mathbb{N}}$. The bipartite state $\gamma^{k_n}_{A'B'A''B''}$ holds $\log_2 k_n$ secret bits.

    We define the converse of $a$-exponential one-way distillable key of $\rho_{AB}$ as follows:
    \begin{equation}
        \widetilde{K}^{\to}_{D,a}\!\left(\rho_{AB}\right) \coloneqq \sup_{\substack{\left\{k_n\right\}_{n\in \mathbb{N}},\left\{\gamma^{k_n}_{A'B'A''B''}\right\}_{n\in \mathbb{N}}\\\left\{\mathcal{L}^{n,\to}_{A^nB^n\to A'B'A''B''}\right\}_{n\in \mathbb{N}} }} \limsup_{n\to \infty}\left\{\frac{\log_2 k_n}{n}: F\!\left(\gamma^{k_n}_{A'B'A''B''},\mathcal{L}^{\to}\!\left(\rho^{\otimes n}_{AB}\right)\right)\ge 1-2^{-an}\right\}.
    \end{equation}
\end{definition}

\begin{theorem}\label{theo:dist_key_asymptotic_bnd}
    Consider an arbitrary bipartite state $\rho_{AB}$. Let $d \coloneqq \min\{\operatorname{dim}(A),\operatorname{dim}(B)\}$ with $\operatorname{dim}(A)$ and $\operatorname{dim}(B)$ being the dimensions of systems $A$ and $B$, respectively. Fix $a\in (2\log_2d , \infty)$. Then the following bound holds:
    \begin{equation}
        \widetilde{K}^{\to}_{D,a}\!\left(\rho_{AB}\right) \le E^u\!\left(\rho_{AB}\right),
    \end{equation}
    where $E^u(\rho_{AB})$ is the relative-entropy-induced unextendible entanglement of the state $\rho_{AB}$ (i.e., defined as in~\eqref{eq:gen_unext_ent_states} with $\boldsymbol{D}$ replaced by the quantum relative entropy $D$). 
\end{theorem}

\begin{proof}
   Let $\rho_{AB}$ be an arbitrary bipartite state from which we wish to distill secret keys, and let $\operatorname{dim}(A)$ and $\operatorname{dim}(B)$ be the dimensions of systems $A$ and $B$, respectively. Let $\left\{\mathcal{L}^{n,\to}_{A^nB^n\to A'B'A''B''}\right\}_{n\in \mathbb{N}}$ be a sequence of one-way LOCC channels, and let $\left\{\gamma^{k_n}_{A'B'A''B''}\right\}_{n\in \mathbb{N}}$ be a sequence of bipartite private states such that the following condition holds for all $a > 2\log_2 d$ and $n\in \mathbb{N}$:
   \begin{equation}
       F\!\left(\gamma^{k_n}_{A'B'A''B''},\mathcal{L}^{n,\to}\!\left(\rho^{\otimes n}_{AB}\right)\right)\ge 1-2^{-an},
   \end{equation}
   where $d\coloneqq \min\{\operatorname{dim}(A),\operatorname{dim}(B)\}$. 
   
   Let us define $\varepsilon_n \coloneqq 2^{-an}$ for convenience. Let us also define $J^{\varepsilon_n,n}(\rho_{AB})\coloneqq 2^{-E^{u,\varepsilon_n}_{\min}(\rho^{\otimes n}_{AB})}$. From Proposition~\ref{prop:J_eps_range} we know that $J^{\varepsilon,n}$ is bounded from below as follows:
   \begin{equation}\label{eq:J_eps_n_lb}
       J^{\varepsilon_n,n}\!\left(\rho\right) \ge \frac{1-\varepsilon_n}{d^{2n}}.
   \end{equation}
   Corollary~\ref{cor:dist_key_st_alg_ub} implies that the following inequality holds for all one-way LOCC channels $\mathcal{L}^{n,\to}_{A^nB^n\to A'B'A''B''}$ and all private states $\gamma^{k_n}_{A'B'A''B''}$:
   \begin{align}
       \log_2 k_n &\le -\log_2\!\left(\sqrt{J^{\varepsilon_n,n}\!\left(\rho\right)}-\sqrt{\varepsilon_n}\right)\\
       &= -\log_2\!\left(\sqrt{J^{\varepsilon_n,n}\!\left(\rho\right)}\!\left(1-\sqrt{\frac{\varepsilon_n}{J^{\varepsilon_n,n}\!\left(\rho\right)}}\right)\right)\\
       &= -\frac{1}{2}\log_2\!\left(J^{\varepsilon_n,n}\!\left(\rho\right)\right)-\log_2\!\left(1-\sqrt{\frac{\varepsilon_n}{J^{\varepsilon_n,n}\!\left(\rho\right)}}\right).
   \end{align}
   Dividing both sides by $n$ and taking the limit superior as $n\to \infty$ leads to the following inequality:
   \begin{equation}\label{eq:sbits_rate_limsup_ineq}
       \limsup_{n\to \infty} \frac{\log_2 k_n}{n} \le \limsup_{n\to \infty}\left\{-\frac{1}{2n}\log_2 J^{\varepsilon_n,n}\!\left(\rho\right) - \frac{1}{n}\log_2\!\left(1-\sqrt{\frac{\varepsilon_n}{J^{\varepsilon_n,n}\!\left(\rho\right)}}\right)\right\}
   \end{equation}
   
   Using the lower bound on $J^{\varepsilon_n,n}\!\left(\rho\right)$ from~\eqref{eq:J_eps_n_lb}, we arrive at the following inequality:
   \begin{align}
       \frac{\varepsilon_n}{J^{\varepsilon,n}} &\le \varepsilon_n\cdot\frac{d^{2n}}{1-\varepsilon_n}\label{eq:eps_J_ratio_ub_step_1}\\
       &\le \frac{2^{-an}d^{2n}}{1-2^{-an}}\\
       &= \frac{d^{2n}}{2^{an}-1}\\
       &= \frac{2^{2n\log_2 d}}{2^{an}-1} \\
       & = \frac{2^{-n(a-2\log_2 d)}}{1- 2^{-an}},
       \label{eq:eps_J_ratio_ub}
   \end{align}
   where the second inequality follows from the fact that the function $\varepsilon_n/(1-\varepsilon_n)$ increases monotonically with $\varepsilon_n \in [0,1)$ and the fact that $\varepsilon_n \le 2^{-an}$. Thus, for sufficiently large $n$, since $a > 2 \log_2 d$ by assumption, it follows that $\frac{2^{-n(a-2\log_2 d)}}{1- 2^{-an}} \leq 1$ and thus that $\frac{\varepsilon_n}{J^{\varepsilon,n}} \leq 1$.

Then we find that
   \begin{align}
       \limsup_{n\to \infty}-\log_2\!\left(1-\sqrt{\frac{\varepsilon_n}{J^{\varepsilon_n,n}\!\left(\rho\right)}}\right) & = -\log_2\!\left(1-\limsup_{n\to \infty}\sqrt{\frac{\varepsilon_n}{J^{\varepsilon_n,n}\!\left(\rho\right)}}\right) \\
       & \leq -\log_2\!\left(1-\limsup_{n\to \infty}\sqrt{\frac{2^{-n(a-2\log_2 d)}}{1- 2^{-an}}}\right)\\
       & = -\log_2\!\left(1-0\right) = 0.\label{eq:eps_J_ratio_log_final}
   \end{align}
Now let us go back to~\eqref{eq:sbits_rate_limsup_ineq}. Substituting~\eqref{eq:eps_J_ratio_log_final} in~\eqref{eq:sbits_rate_limsup_ineq} leads to the following inequality:
\begin{align}
    \limsup_{n\to \infty} \frac{\log_2 k_n}{n} &\le \limsup_{n\to \infty}-\frac{1}{2n}\log_2 J^{\varepsilon_n,n}\!\left(\rho\right)\\
    &= \limsup_{n\to \infty}\frac{1}{n}E^{u,\varepsilon_n}_{\min}\!\left(\rho^{\otimes n}_{AB}\right)\\
    &= \limsup_{n\to \infty}\frac{1}{2n}\inf_{\sigma_{A^nB^n}\in \mathcal{F}(\rho^{\otimes n}_{AB})}D^{\varepsilon_n}_{\min}\!\left(\rho^{\otimes n}_{AB}\Vert \sigma_{A^nB^n}\right)\label{eq:sbits_rate_limsup_hypo_test_rel_ent_ub}
\end{align}
where we have used the definition of $J^{\varepsilon_n,n}\!\left(\rho\right)$ to arrive at the first equality and the definition of $E^{u,\varepsilon}_{\min}$ to arrive at the second equality. Note that if $\sigma_{AB} \in \mathcal{F}(\rho_{AB})$ then $\sigma^{\otimes n}_{AB}\in \mathcal{F}\left(\rho^{\otimes n}_{AB}\right)$. Therefore,
\begin{equation}
    \inf_{\sigma_{A^nB^n}\in \mathcal{F}(\rho^{\otimes n}_{AB})}D^{\varepsilon_n}_{\min}\!\left(\rho^{\otimes n}_{AB}\Vert \sigma_{A^nB^n}\right) \le \inf_{\sigma_{AB}\in \mathcal{F}(\rho_{AB})}D^{\varepsilon_n}_{\min}\!\left(\rho^{\otimes n}_{AB}\Vert \sigma^{\otimes n}_{AB}\right).
\end{equation}
Substituting the above inequality in~\eqref{eq:sbits_rate_limsup_hypo_test_rel_ent_ub}, we arrive at the following inequality:
\begin{align}
    \limsup_{n\to \infty}\frac{\log_2 k_n}{n} &\le \limsup_{n\to \infty} \frac{1}{2n} \inf_{\sigma_{AB}\in \mathcal{F}(\rho_{AB})}D^{\varepsilon_n}_{\min}\!\left(\rho^{\otimes n}_{AB}\Vert \sigma^{\otimes n}_{AB}\right)\\
    &\leq  \inf_{\sigma_{AB}\in \mathcal{F}(\rho_{AB})}\limsup_{n\to \infty}\frac{1}{2n}D^{\varepsilon_n}_{\min}\!\left(\rho^{\otimes n}_{AB}\Vert \sigma^{\otimes n}_{AB}\right),\label{eq:sbits_rate_hypo_test_limsup_ub_2}
\end{align}
where the second inequality from an asymptotic version of the max-min inequality.

Note that $D^{\varepsilon}_{\min}(\rho\Vert\sigma)$ increases monotonically with increasing $\varepsilon$. Since $\varepsilon_n \le 2^{-an}$, for every $\varepsilon^* \in (0,1)$, there exists an $N \in \mathbb{N}$ such that $\varepsilon_n \le \varepsilon^*$ for all $n \ge N$. Consequently, 
\begin{equation}
    \frac{1}{n}D^{\varepsilon_n}_{\min}\!\left(\rho^{\otimes n}_{AB}\Vert\sigma^{\otimes n}_{AB}\right) \le \frac{1}{n}D^{\varepsilon^*}_{\min}\!\left(\rho^{\otimes n}_{AB}\Vert\sigma^{\otimes n}_{AB}\right) \qquad \forall n\ge N.
\end{equation}
Substituting the above inequality in~\eqref{eq:sbits_rate_hypo_test_limsup_ub_2}, we arrive at the following inequality:
\begin{equation}
    \limsup_{n\to \infty}\frac{\log_2 k_n}{n} \le \inf_{\sigma_{AB}\in \mathcal{F}(\rho_{AB})} \frac{1}{2}\limsup_{n\to \infty} \frac{1}{n}D^{\varepsilon^*}_{\min}\!\left(\rho^{\otimes n}_{AB}\Vert\sigma^{\otimes n}_{AB}\right).
\end{equation}
For all $\varepsilon \in (0,1)$, the following inequality holds~\cite{nagaoka2000StrongConverseSteins}:
\begin{equation}
    \limsup_{n\to \infty}\frac{1}{n}D^{\varepsilon}_{\min}\!\left(\rho^{\otimes n}\Vert\sigma^{\otimes n}\right) \le D\!\left(\rho\Vert\sigma\right),
\end{equation}
where $D(\cdot\Vert\cdot)$ is the Umegaki relative entropy~\cite{Ume62}. Therefore, we conclude the following:
\begin{equation}
    \limsup_{n\to \infty}\frac{\log_2 k_n}{n} \le \inf_{\sigma_{AB}\in \mathcal{F}(\rho_{AB})} \frac{1}{2} D(\rho_{AB}\Vert\sigma_{AB}) = E^u\!\left(\rho_{AB}\right).
\end{equation}
Since the above inequality holds for all sequences of one-way LOCC channels $\left\{\mathcal{L}^{n,\to}_{A^nB^n\to A'B'A''B''}\right\}_{n\in \mathbb{N}}$ and all sequences of private states $\left\{\gamma^{k_n}_{A'B'A''B''}\right\}_{n\in \mathbb{N}}$ such that $F\left(\gamma^{k_n}_{A'B'A''B''},\mathcal{L}^{\to}(\rho^{\otimes n}_{AB})\right) \ge 1-2^{-an}$ for all $a > 2\log_2 d$, we conclude the statement of theorem.
\end{proof}

\begin{remark}
    The relative-entropy-induced unextendible entanglement of a state can be computed using a semidefinite program. See~\cite{KS24} for a semidefinite representation of the relative entropy between two states that can be used directly to estimate the relative-entropy-induced unextendible entanglement of a state to arbitrary precision.  
\end{remark}

\section{Forward-assisted private communication from channels}\label{sec:priv_comm_results}

In this section we extend the results obtained in Section~\ref{sec:distillable_key_results} to understand the limitations of private communication over channels. We begin with a brief discussion on secret-key distillation from a channel with local operations and forward classical communication in the one-shot setting. In this setting, where forward classical communication can be performed with no cost, the task of secret-key distillation from a channel is equivalent to the task of private communication from the channel.

\subsection{One-shot, one-way distillable key of a channel}

To distill a secret key from a channel using one-way LOCC, Alice locally prepares a state $\psi_{A'A''\hat{A}}$, and encodes one share of this state using a quantum instrument $\left\{\mathcal{E}^x_{\hat{A}\to A}\right\}_{x\in \mathcal{X}}$. She then sends the system $A$ to Bob through the quantum channel $\mathcal{N}_{A\to B}$ along with the classical label $x$. Bob then decodes the received state by applying a quantum channel $\mathcal{D}^x_{B\to B'B''}$, which he can choose based on the classical label $x$ that he received from Alice. The state established at the end of the protocol can be mathematically described as follows:
\begin{equation}
    \sigma_{A'B'A''B''} \coloneqq \sum_{x\in \mathcal{X}}\left(\mathcal{D}^x_{B\to B'B''}\circ\mathcal{N}_{A\to B}\circ\mathcal{E}^x_{\hat{A}\to A}\right)\!\left(\psi_{A'A''\hat{A}}\right).
\end{equation}
For the secret-key distillation task to be successful in distilling $\log_2k$ secret bits with an error tolerance $\varepsilon$, we require the following inequality to hold:
\begin{equation}
    F\!\left(\sigma_{A'B'A''B''},\gamma^k_{A'B'A''B''}\right) \ge 1-\varepsilon
\end{equation}
for some private state $\gamma^k_{A'B'A''B''}$ holding $\log_2 k$ secret bits.  Figure~\ref{fig:key_distill_ch_diagram} depicts a schematic diagram of the task of one-way secret-key distillation from a channel.

\begin{figure}
    \centering
    \includegraphics[width=0.7\linewidth]{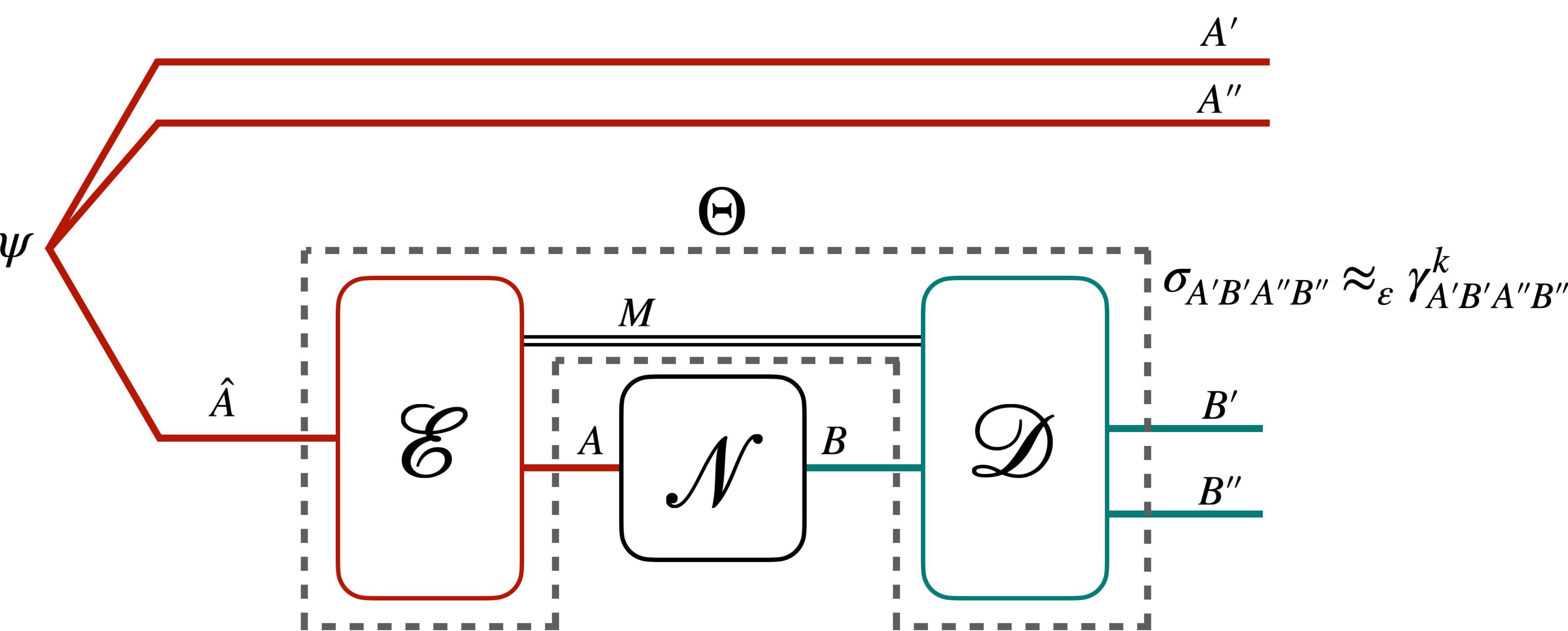}
    \caption{Schematic diagram of secret-key distillation from a channel $\mathcal{N}_{A\to B}$ using a one-way LOCC superchannel $\Theta$. The error in the distillation process, denoted by $\varepsilon$, is given by the infidelity between the state $\sigma_{A'B'A''B''}$ established at the end of the protocol and a private state $\gamma^k_{A'B'A''B''}$ that holds $\log_2 k$ secret bits.}
    \label{fig:key_distill_ch_diagram}
\end{figure}

The task of secret-key distillation from a quantum channel can be expressed more concisely using the language of superchannels~\cite{Chiribella_2008, Gour_2019}. To distill $\log_2 k$ secret bits from a channel $\mathcal{N}_{A\to B}$ with an error tolerance $\varepsilon$ and only using one-way LOCC, Alice and Bob apply a one-way LOCC superchannel $\Theta_{(A\to B)\to (\hat{A}\to B'B'')}$ on the channel $\mathcal{N}_{A\to B}$ such that the following inequality holds:
\begin{equation}
    F\!\left(\left(\Theta\!\left(\mathcal{N}\right)\right)\!\left(\psi_{A'A''\hat{A}}\right),\gamma^k_{A'B'A''B''}\right) \ge 1-\varepsilon,
\end{equation}
for some locally prepared state $\psi_{A'A''\hat{A}}$ and some private state $\gamma^k_{A'B'A''B''}$ holding $\log_2 k$ secret bits. The ability to distill secret keys from a channel in the one-shot setting can then be quantified by the one-shot, one-way distillable key of the channel, which we define below.
\begin{definition}
    The one-shot, one-way distillable key of a channel $\mathcal{N}_{A\to B}$ is defined as follows:
    \begin{equation}
        K^{\varepsilon,\to}_D\!\left(\mathcal{N}_{A\to B}\right) \coloneqq \sup_{\substack{k\in \mathbb{N,}\gamma^k_{A'B'A''B''},\\ \psi_{A'A''\hat{A}} \in \mathcal{S}(A'A''\hat{A}),\\\Theta \in \operatorname{1WL}}} \left\{\log_2 k: F\!\left(\left(\Theta\!\left(\mathcal{N}\right)\right)\!\left(\psi_{A'A''\hat{A}}\right),\gamma^k_{A'B'A''B''}\right) \ge 1-\varepsilon\right\},
    \end{equation}
    where the supremum is over every natural number $k$, every quantum state $\psi_{A'A''\hat{A}}$, every private state $\gamma^k_{A'B'A''B''}$, and every one-way LOCC superchannel $\Theta_{(A\to B)\to (\hat{A}\to B'B'')}$.
\end{definition}

The one-shot, one-way distillable key of a channel can be written in terms of the one-shot, one-way distillable key of a state as follows:
\begin{equation}\label{eq:ch_1shot_1w_key_wrt_st}
    K^{\varepsilon,\to}_D\!\left(\mathcal{N}_{A\to B}\right) = \sup_{\substack{\psi_{A'A''\hat{A}} \in \mathcal{S}(A'A''\hat{A}),\\\Theta \in \operatorname{1WL}}} K^{\varepsilon,\to}_D\!\left(\left(\Theta_{(A\to B)\to (\hat{A}\to B'B'')}\!\left(\mathcal{N}_{A\to B}\right)\right)\!\left(\psi_{A'A''\hat{A}}\right)\right).
\end{equation}

\begin{figure}
    \centering
    \includegraphics[width=0.8\linewidth]{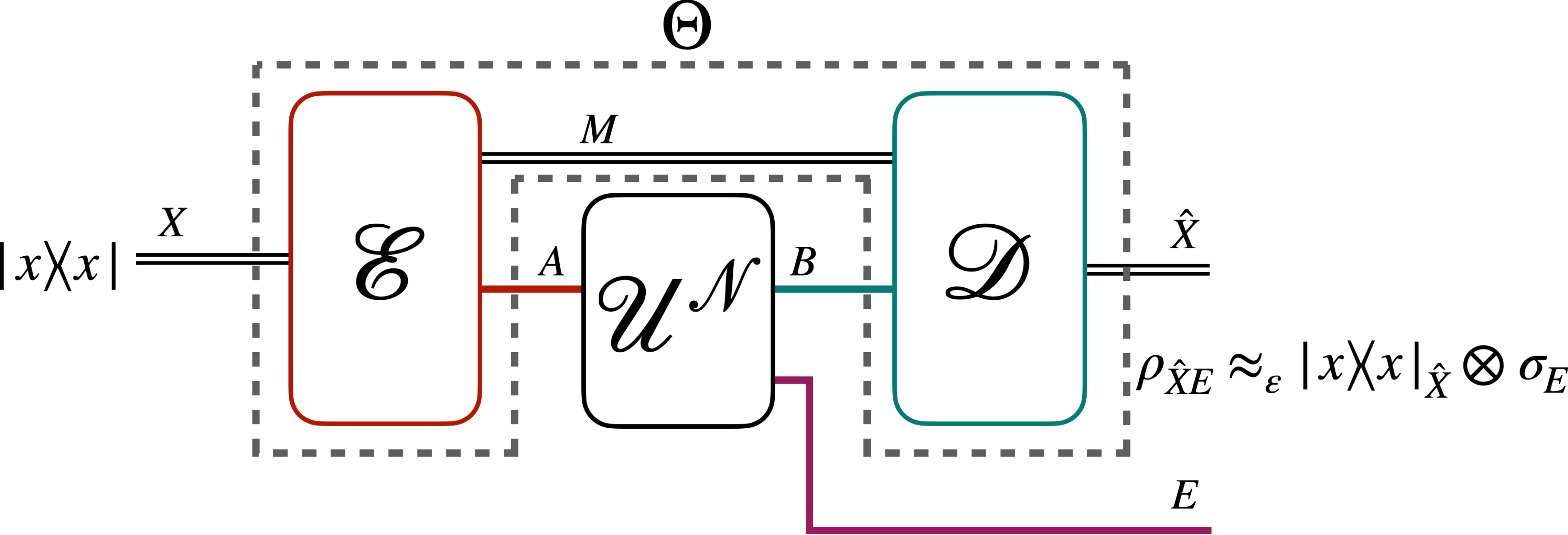}
    \caption{Schematic diagram of private communication over a channel $\mathcal{N}_{A\to B}$, with an isometric extension $\mathcal{U}^{\mathcal{N}}_{A\to BE}$, using a one-way LOCC superchannel $\Theta$. The objective of this protocol is to send an arbitrary classical label $x$ from Alice to Bob such that the state of an eavesdropper, who holds the system $E$, is independent of the state Bob receives. The error in private communication, denoted by $\varepsilon$, is defined in~\eqref{eq:priv_comm_error}.}
    \label{fig:priv_comm}
\end{figure}

Establishing a secret key and using the one-time-pad scheme for private communication is not the only way to transmit private bits over a channel, and there may exist alternate protocols to realize private communication over a quantum channel~\cite{DLL03}. The notion of one-shot private capacity of channels is used to quantify the ability of a quantum channel to communicate data privately using local operations without making assumptions on the protocol. 

Suppose that Alice wants to send private data to Bob by using a channel $\mathcal{N}_{A\to B}$. To do this, Alice and Bob apply a superchannel $\Theta_{(A\to B)\to (X\to \hat{X})}$ on the channel $\mathcal{N}_{A\to B}$, where $X$ and $\hat{X}$ are classical systems. Let $\mathcal{U}^{\mathcal{N}}_{A\to BE}$ be an isometric extension of the channel $\mathcal{N}_{A\to B}$, where the eavesdropper has access to the system $E$. The error in private communication through this protocol is defined as follows:
\begin{equation}\label{eq:priv_comm_error}
    p_{\operatorname{err}}\!\left(\Theta,\mathcal{N}\right) \coloneqq \inf_{\sigma_E}\max_{x\in \mathcal{X}}\left(1-F\!\left(|x\rangle\!\langle x|_{\hat{X}}\otimes \sigma_E, \left(\Theta\!\left(\mathcal{U}^{\mathcal{N}}_{A\to BE}\right)\right)\!\left(|x\rangle\!\langle x|_{X}\right)\right)\right),
\end{equation}
where the infimum is over all quantum states $\sigma_E$ and the maximum is over all messages $x$ in the set $\mathcal{X}$. The one-shot private capacity of a channel $\mathcal{N}_{A\to B}$ is then defined as follows:
\begin{equation}
    P^{\varepsilon}\!\left(\mathcal{N}_{A\to B}\right) \coloneqq \sup_{\mathcal{X}, \Theta \in \operatorname{LO}} \left\{\log_2|\mathcal{X}|: p_{\operatorname{err}}\!\left(\Theta,\mathcal{N}\right)\le \varepsilon\right\},
\end{equation}
where $\operatorname{LO}$ refers to the set of all superchannels that can be realized by local operations only and $|\mathcal{X}|$ refers to the number of elements in the set $\mathcal{X}$. As such, the supremum is over all sets of messages $\mathcal{X}$ and all superchannels $\Theta_{(A\to B)\to (X\to \hat{X})}$ that can be realized by only local operations.

In the one-way LOCC setting considered throughout this work, Alice is allowed to send arbitrary amounts of classical data to Bob, which is publicly available to any eavesdropper as well. In this setting, the quantity of interest is the one-shot forward-assisted private capacity of a channel, which is defined as follows:
\begin{equation}
    P^{\varepsilon,\to}\!\left(\mathcal{N}_{A\to B}\right) \coloneqq \sup_{\mathcal{X}, \Theta \in \operatorname{1WL}} \left\{\log_2|\mathcal{X}|: p_{\operatorname{err}}\!\left(\Theta,\mathcal{N}\right)\le \varepsilon\right\},
\end{equation}
where $\operatorname{1WL}$ refers to the set of all one-way LOCC superchannels. Since $\operatorname{LO}\subseteq \operatorname{1WL}$, the following inequality holds for all quantum channels $\mathcal{N}_{A\to B}$:
\begin{equation}
\label{eq:basic-priv-cap-up-bnd}
    P^{\varepsilon,\to}\!\left(\mathcal{N}_{A\to B}\right) \ge P^{\varepsilon}\!\left(\mathcal{N}_{A\to B}\right).
\end{equation}
In the remainder of this paper, we will derive several upper bounds on the one-shot forward-assisted private capacity of channels, which, as a consequence of~\eqref{eq:basic-priv-cap-up-bnd}, also serve as upper bounds on the one-shot private capacity of the channel due to the inequality mentioned above.

In the presence of forward-classical assistance, the task of secret-key distillation is equivalent to the task of private communication. Suppose that a forward-assisted protocol allows Alice to send $n$ private bits to Bob through a channel $\mathcal{N}_{A\to B}$ with some error $\varepsilon$. Alice can send a secret key itself through this channel, hence, transforming the one-shot private communication protocol to a one-shot secret-key distillation protocol. Moreover, in the forward-classical assistance setting, a secret-key distillation protocol can be transformed into a private communication protocol by using the one-time-pad scheme, thus demonstrating the equivalence between the two tasks.

Due to the equivalence between the tasks of private communication and secret-key distillation in the presence of forward-classical assistance, the one-shot forward-assisted private capacity of a channel $P^{\varepsilon,\to}\!\left(\mathcal{N}_{A\to B}\right)$ is equal to the one-shot, one-way distillable key of the channel. That is,
\begin{equation}\label{eq:dist_key_eq_priv_cap}
    P^{\varepsilon,\to}\!\left(\mathcal{N}_{A\to B}\right) = K^{\varepsilon,\to}_D\!\left(\mathcal{N}_{A\to B}\right).
\end{equation}

The techniques used in Section~\ref{sec:distillable_key_results} for obtaining upper bounds on the one-shot, one-way distillable key of a state can be extended to obtain upper bounds on the one-shot, one-way distillable key of a channel, using the unextendible entanglement of channels. Furthermore, the equality in~\eqref{eq:dist_key_eq_priv_cap} allows us to obtain upper bounds on the one-shot forward-assisted private capacity, which are also upper bounds on the one-shot private capacity by definition.

\subsection{Unextendible entanglement of channels}

The generalized unextendible entanglement of channels was defined in~\cite{SW24_channels}. We briefly present the relevant properties of the quantity here.

Let us define the following set of channels with respect to a given channel $\mathcal{N}_{A\to B}$:
\begin{equation}\label{eq:free_channels}
    \mathcal{F}\!\left(\mathcal{N}_{A\to B}\right) \coloneqq \left\{\operatorname{Tr}_{B}\circ\mathcal{P}_{A\to BE}: \operatorname{Tr}_{E}\circ\mathcal{P}_{A\to BE} = \mathcal{N}_{A\to B}\right\},
\end{equation}
where systems $B$ and $E$ are isomorphic. The generalized unextendible entanglement of a channel $\mathcal{N}_{A\to B}$ is defined with respect to a generalized channel divergence $\mathbf{D}$~\cite{CMW16, LKDW18} as follows:
\begin{equation}
    \mathbf{E}^u\!\left(\mathcal{N}_{A\to B}\right) \coloneqq \frac{1}{2}\inf_{\mathcal{M}\in \mathcal{F}(\mathcal{N})}\mathbf{D}\!\left(\mathcal{N}_{A\to B}\Vert\mathcal{M}_{A\to E}\right),
\end{equation}
where the generalized divergence between channels is defined as
\begin{equation}
    \mathbf{D}\!\left(\mathcal{N}_{A\to B}\Vert\mathcal{M}_{A\to B}\right) = \sup_{\rho_{RA}\in \mathcal{S}(RA)}\mathbf{D}\!\left(\mathcal{N}_{A\to B}\!\left(\rho_{RA}\right)\Vert \mathcal{M}_{A\to B}\!\left(\rho_{RA}\right)\right).
\end{equation}

The generalized unextendible entanglement of a channel does not increase under the action of one-way LOCC superchannels, the latter defined in Section~\ref{sec:superchannels}.
\begin{lemma}[\cite{SW24_channels}]\label{lem:unext_ent_ch_monotonic}
    The generalized unextendible entanglement of a channel does not increase under the action of one-way LOCC superchannels. That is,
    \begin{equation}
        \mathbf{E}^u\!\left(\Theta\!\left(\mathcal{N}_{A\to B}\right)\right) \le \mathbf{E}^u\!\left(\mathcal{N}_{A\to B}\right) \qquad \forall \Theta \in \operatorname{1WL},
    \end{equation}
    where $\operatorname{1WL}$ represents the set of all one-way LOCC superchannels.
\end{lemma}

A more general statement of Lemma~\ref{lem:unext_ent_ch_monotonic} was presented in~\cite{SW24_channels}, where it was shown that the generalized unextendible entanglement of a channel does not increase under the action of two-extendible superchannels, which is a semidefinite relaxation of the set of one-way LOCC superchannels, and hence, contains the set of one-way LOCC superchannels. In this work we will only consider one-way LOCC superchannels and not two-extendible superchannels. We point the interested reader to~\cite{SW24_channels} for a more detailed discussion.

A direct consequence of Lemma~\ref{lem:unext_ent_ch_monotonic} is that the maximum value of the generalized unextendible entanglement of a channel $\mathcal{N}_{A\to B}$ is not larger than the generalized unextendible entanglement of the identity channel $\operatorname{id}_{A'\to B'}$, where $\operatorname{dim}(A') = \operatorname{dim}(B') = \min\!\left\{\operatorname{dim}(A), \operatorname{dim}(B)\right\}$. This can be seen from the following argument: Consider an arbitrary channel $\mathcal{N}_{A\to B}$. If $\operatorname{dim}(A) \ge \operatorname{dim}(B)$, then construct a superchannel $\Theta_{(B\to C)\to (A\to D)}$ that acts on an arbitrary channel $\mathcal{M}_{B\to C}$ as follows:
\begin{equation}
    \Theta_{(B\to C)\to (A\to D)}\!\left(\mathcal{M}_{B\to C}\right) = \operatorname{id}_{C\to D}\circ\mathcal{M}_{B\to C}\circ\mathcal{N}_{A\to B}.
\end{equation}
Lemma~\ref{lem:unext_ent_ch_monotonic} implies the following inequality:
\begin{align}
    \mathbf{E}^u\!\left(\operatorname{id}_{B\to C}\right)&\ge \mathbf{E}^u\!\left(\Theta\!\left(\operatorname{id}_{B\to C}\right)\right)\\
    &= \mathbf{E}^u\!\left(\operatorname{id}_{C\to D}\circ\operatorname{id}_{B\to C}\circ\mathcal{N}_{A\to B}\right)\\
    &= \mathbf{E}^u\!\left(\mathcal{N}_{A\to B}\right).\label{eq:unext_ent_id_ge_ch_in_ge_out}
\end{align}
Similarly, if $\operatorname{dim}(B)\ge \operatorname{dim}(A)$, then we can construct a superchannel $\Upsilon_{(D\to A)\to (C\to B)}$ that acts on an arbitrary channel $\mathcal{M}_{D\to A}$ as follows:
\begin{equation}
    \Upsilon_{(D\to A)\to (C\to B)}\!\left(\mathcal{M}_{D\to A}\right) = \mathcal{N}_{A\to B}\circ\mathcal{M}_{D\to A}\circ\operatorname{id}_{C\to D}.
\end{equation}
Once again, applying Lemma~\ref{lem:unext_ent_ch_monotonic} leads to the following inequality:
\begin{equation}\label{eq:unext_ent_id_ge_ch_out_ge_in}
    \mathbf{E}^u\!\left(\operatorname{id}_{D\to A}\right) \ge \mathbf{E}^u\!\left(\Upsilon\!\left(\operatorname{id}_{D\to A}\right)\right) = \mathbf{E}^u\!\left(\mathcal{N}_{A\to B}\right).
\end{equation}
The inequalities in~\eqref{eq:unext_ent_id_ge_ch_in_ge_out} and~\eqref{eq:unext_ent_id_ge_ch_out_ge_in} can be written together as the following inequality:
\begin{equation}
    \mathbf{E}^u\!\left(\mathcal{N}_{A\to B}\right) \le \min\!\left\{\mathbf{E}^u\!\left(\operatorname{id}_{A\to C}\right),\mathbf{E}^u\!\left(\operatorname{id}_{B\to D}\right)\right\},
\end{equation}
where $\operatorname{dim}(A) = \operatorname{dim}(C)$ and $\operatorname{dim}(B) = \operatorname{dim}(D)$.

Another important property of generalized unextendible entanglement of channels, which is relevant to our discussion, is its relation with the generalized unextendible entanglement of states. In particular, the generalized unextendible entanglement of a bipartite state that can be established between two distant parties using a quantum channel $\mathcal{N}_{A\to B}$ and one-way LOCC superchannels cannot be larger than the generalized unextendible entanglement of the channel $\mathcal{N}_{A\to B}$. We state this formally in Lemma~\ref{lem:unext_ent_state_le_unext_ent_ch_gen} below.
\begin{lemma}[\cite{SW24_channels}]\label{lem:unext_ent_state_le_unext_ent_ch_gen}
    The unextendible entanglement of a quantum state $\sigma_{RC'D}$, with respect to the partition $RC':D$, that can be established between two parties using a point-to-point quantum channel $\mathcal{N}_{A\to B}$ and a one-way LOCC superchannel $\Theta_{(A\to B)\to (C\to C'D)}$ is no greater than the unextendible entanglement of the quantum channel $\mathcal{N}_{A\to B}$; i.e., 
    \begin{equation}\label{eq:unext_ent_state_le_unext_ent_ch_gen}
        \sup_{\rho_{RC}}\mathbf{E}^u\!\left(\sigma_{RC':D}\right) \le  \mathbf{E}^u\!\left(\mathcal{N}_{A\to B}\right), 
    \end{equation}
    where 
    \begin{equation}
        \sigma_{RC'D} \coloneqq \left(\Theta_{(A\to B)\to (C\to C'D)}\!\left(\mathcal{N}_{A\to B}\right)\right)\left(\rho_{RC}\right),
    \end{equation}
    and $\rho_{RC}$ is a quantum state. The symbol $\mathbf{E}^u\!\left(\sigma_{RC':D}\right)$ denotes that the unextendible entanglement of the state $\sigma_{RC'D}$ is calculated with respect to the bipartition $RC':D$. 
\end{lemma}

Lemma~\ref{lem:unext_ent_state_le_unext_ent_ch_gen}, Lemma~\ref{lem:monotoncity_obj_func}, and~\eqref{eq:ch_1shot_1w_key_wrt_st} provide us with all the necessary tools to obtain an upper bound on the one-shot, one-way distillable key of a channel using unextendible entanglement of channels.

The two important quantities that we will use in this section are the smooth-min unextendible entanglement of a channel and the $\alpha$-geometric unextendible entanglement of the channel.

The smooth-min unextendible entanglement of a channel is defined in terms of the smooth-min relative entropy of channels as follows:
\begin{align}
    E^{u,\varepsilon}_{\min}\!\left(\mathcal{N}_{A\to B}\right) &\coloneqq \frac{1}{2} \inf_{\mathcal{M}\in \mathcal{F}(\mathcal{N})} D^{\varepsilon}_{\min}\!\left(\mathcal{N}_{A\to B}\Vert\mathcal{M}_{A\to E}\right)\\
    &=\frac{1}{2} \inf_{\mathcal{M}\in \mathcal{F}(\mathcal{N})} \sup_{\rho_{RA}\in \mathcal{S}(RA)}D^{\varepsilon}_{\min}\!\left(\mathcal{N}_{A\to B}\!\left(\rho_{RA}\right)\Vert\mathcal{M}_{A\to E}\!\left(\rho_{RA}\right)\right),
\end{align}
where the smooth-min relative entropy of states was defined in~\eqref{eq:smooth_min_rel_ent_defn}. The smooth-min relative entropy of channels can be written as a semidefinite program~\cite[Appendix~B-3]{WW19}, and the set $\mathcal{F}(\mathcal{N})$ can also be described by semidefinite constraints. This allows us to write the smooth-min unextendible entanglement of a channel as a semidefinite program (see Appendix~\ref{app:semidefinite_programs}).

\begin{proposition}\label{prop:smooth_min_ch_range}
    The smooth-min unextendible entanglement of a channel is bounded as follows:
    \begin{equation}
        - \frac{1}{2}\log_2(1-\varepsilon) \le E^{u,\varepsilon}_{\min}\!\left(\mathcal{N}_{A\to B}\right) \le \log_2d - \frac{1}{2}\log_2(1-\varepsilon),
    \end{equation}
    where $d \coloneqq \min\{\operatorname{dim}(A),\operatorname{dim}(B)\}$.
\end{proposition}
\begin{proof}
    See Appendix~\ref{app:smooth_min_ch_range}.
\end{proof}

\medskip

The $\alpha$-geometric unextendible entanglement of channels was explored in~\cite{SW24_channels} in the context of zero-error private communication. It is defined for a parameter $\alpha \in (0,1)\cup(1,2]$ as follows:
\begin{align}
    \widehat{E}^u_{\alpha}\!\left(\mathcal{N}_{A\to B}\right) &\coloneqq \inf_{\mathcal{M}\in\mathcal{F}(\mathcal{N})}\frac{1}{2}\widehat{D}_{\alpha}\!\left(\mathcal{N}_{A\to B}\Vert\mathcal{M}_{A\to B}\right)\\
    &=\inf_{\mathcal{M}\in\mathcal{F}(\mathcal{N})}\sup_{\rho_{RA}\in \mathcal{S}(RA)} \frac{1}{2}\widehat{D}_{\alpha}\!\left(\mathcal{N}_{A\to B}(\rho_{RA})\Vert\mathcal{M}_{A\to B}(\rho_{RA})\right),
\end{align}
where the $\alpha$-geometric R\'enyi relative entropy of states is defined for all $\alpha \in (0,1) \cup (1,\infty)$ as follows~\cite{Mat13}:
\begin{equation}\label{eq:alpha_geo_rel_ent_defn}
    \widehat{D}_{\alpha}\!\left(\rho\Vert\sigma\right) = \frac{1}{\alpha - 1}\log_2\operatorname{Tr}\!\left[\sigma\!\left(\sigma^{-\frac{1}{2}}\rho\sigma^{-\frac{1}{2}}\right)^{\alpha}\right].
\end{equation}

We list some properties of the $\alpha$-geometric unextendible entanglement of channels that are relevant to this work. We refer the interested reader to~\cite{SW24_channels} for a more detailed discussion of these properties:
\begin{enumerate}
    \item \textbf{Monotonicity in $\alpha$:} The $\alpha$-geometric unextendible entanglement of channels increases monotonically with increasing $\alpha$. That is,
    \begin{equation}
        \widehat{E}^u_{\alpha}\!\left(\mathcal{N}_{A\to B}\right) \ge \widehat{E}^u_{\beta}\!\left(\mathcal{N}_{A\to B}\right) \qquad \forall \alpha,\beta \in (0,1)\cup(1,2],\quad \alpha \ge \beta.
    \end{equation}
    \item \textbf{Subadditivity:} The $\alpha$-geometric unextendible entanglement of channels is subadditive under tensor products of channels. That is, the following inequality holds for all $\alpha\in (0,1)\cup(1,2]$:
    \begin{equation}
    \widehat{E}^u_{\alpha}\!\left(\mathcal{N}_{A\to B}\otimes \mathcal{M}_{A\to B}\right) \le \widehat{E}^u_{\alpha}\!\left(\mathcal{N}_{A\to B}\right) + \widehat{E}^u_{\alpha}\!\left( \mathcal{M}_{A\to B}\right),
    \label{eq:alpha-geo-subadd}
    \end{equation}
    where $\mathcal{N}_{A\to B}$ and $\mathcal{M}_{A\to B}$ are quantum channels.
    \item \textbf{Limiting case when $\alpha \to 1$:} The $\alpha$-geometric unextendible entanglement converges to the unextendible entanglement induced by the Belavkin--Staszewski relative entropy as $\alpha \to 1$. That is,
    \begin{equation}\label{eq:Bel_Stas_unext_ent}
        \widehat{E}^u\!\left(\mathcal{N}_{A\to B}\right)\coloneqq \lim_{\alpha \to 1}\widehat{E}^u_{\alpha}\!\left(\mathcal{N}_{A\to B}\right) = \inf_{\mathcal{M}\in \mathcal{F}(\mathcal{N})}\frac{1}{2}\widehat{D}\!\left(\mathcal{N}_{A\to B}\middle\Vert\mathcal{M}_{A\to B}\right),
    \end{equation}
    where the Belavkin--Staszewski relative entropy of states is defined as follows~\cite{BS82}:
    \begin{equation}
        \widehat{D}\!\left(\rho\Vert\sigma\right) \coloneqq \Bigg\{\begin{array}{cc}
            \operatorname{Tr}\!\left[\rho\!\log_2\left(\sqrt{\rho}\sigma^{-1}\sqrt{\rho}\right)\right] & \text{if }\operatorname{supp}(\rho) \subseteq \operatorname{supp}(\sigma)   \\
             +\infty & \text{otherwise} 
        \end{array},
    \end{equation}
    and where $\sigma^{-1}$ is taken on the support of $\sigma$ and the logarithm is evaluated on the support of $\rho$. The rightmost equality in~\eqref{eq:Bel_Stas_unext_ent} follows directly from~\cite[Proposition~36]{DKQSWW23}.
    \item \textbf{Semidefinite program:} The $\alpha$-geometric unextendible entanglement of a channel can be computed using a semidefinite program for rational values of $\alpha \in (1,2]$ (see Appendix~\ref{app:semidefinite_programs}).
\end{enumerate}

The $\alpha$-geometric unextendible entanglement of a channel can be related with the smooth-min unextendible entanglement using the inequality in~\eqref{eq:smooth_min_vs_sandwich_ineq}. The $\alpha$-geometric R\'enyi relative entropy of states is known to be larger than or equal to the $\alpha$-sandwiched R\'enyi relative entropy of states for all $\alpha \in (0,1)\cup(1,\infty)$~\cite{Tomamichel15, WWW24}. The inequality in~\eqref{eq:smooth_min_vs_sandwich_ineq} then implies the following inequality, which holds for all $\alpha \in (1,\infty)$ and $\varepsilon \in [0,1)$:
\begin{align}
    D^{\varepsilon}_{\min}\!\left(\rho\Vert\sigma\right) &\le \widetilde{D}_{\alpha}\!\left(\rho\Vert\sigma\right) + \frac{\alpha}{\alpha - 1}\log_2\!\left(\frac{1}{1-\varepsilon}\right)\\
    &\le \widehat{D}_{\alpha}\!\left(\rho\Vert\sigma\right) + \frac{\alpha}{\alpha - 1}\log_2\!\left(\frac{1}{1-\varepsilon}\right).\label{eq:smooth_min_vs_geo_ineq}
\end{align}
We will restrict our discussion to $\alpha \in (1,2]$, as this is the interval for which the $\alpha$-geometric R\'enyi relative entropy obeys the data-processing inequality. Setting $\rho \to (\operatorname{id}_R \otimes \mathcal{N})(\rho_{RA})$ and $\sigma\to (\operatorname{id}_R \otimes \mathcal{M})(\rho_{RA})$, where $\mathcal{N}$ and $\mathcal{M}$ are quantum channels, leads to the following inequality:
\begin{multline}
    D^{\varepsilon}_{\min}\!\left((\operatorname{id}_R \otimes \mathcal{N})(\rho_{RA})\Vert(\operatorname{id}_R \otimes \mathcal{M})(\rho_{RA})\right) \le \widehat{D}_{\alpha}\!\left((\operatorname{id}_R \otimes \mathcal{N})(\rho_{RA})\Vert(\operatorname{id}_R \otimes \mathcal{M})(\rho_{RA})\right) \\
    + \frac{\alpha}{\alpha - 1}\log_2\!\left(\frac{1}{1-\varepsilon}\right).
\end{multline}
Since the above inequality holds for every state $\rho$, we can take a supremum over all states and conclude the following inequality:
\begin{equation}
    D^{\varepsilon}_{\min}\!\left(\mathcal{N}\Vert\mathcal{M}\right) 
    \le \widehat{D}_{\alpha}\!\left(\mathcal{N}\Vert\mathcal{M}\right) + \frac{\alpha}{\alpha - 1}\log_2\!\left(\frac{1}{1-\varepsilon}\right).
\end{equation}
Now taking an infimum over all $\mathcal{M}\in \mathcal{F}(\mathcal{N})$, we arrive at the following inequality, which holds for all $\alpha \in (1,2]$ and $\varepsilon \in [0,1)$:
\begin{align}
    E^{u,\varepsilon}_{\min}\!\left(\mathcal{N}_{A\to B}\right) &= \inf_{\mathcal{M}\in \mathcal{F}(\mathcal{N})}\frac{1}{2}D^{\varepsilon}_{\min}\!\left(\mathcal{N}_{A\to B}\Vert\mathcal{M}_{A\to B}\right)\\
    &\le \inf_{\mathcal{M}\in \mathcal{F}(\mathcal{N})}\frac{1}{2}\widehat{D}_{\alpha}\!\left(\mathcal{N}_{A\to B}\Vert\mathcal{M}_{A\to B}\right) + \frac{1}{2}\cdot\frac{\alpha}{\alpha -1}\log_2\!\left(\frac{1}{1-\varepsilon}\right)\\
    &= \widehat{E}^u_{\alpha}\!\left(\mathcal{N}_{A\to B}\right) + \frac{1}{2}\cdot\frac{\alpha}{\alpha - 1}\log_2\!\left(\frac{1}{1-\varepsilon}\right).\label{eq:smooth_min_vs_geo_unext_ent_ineq}
\end{align}

\subsection{Upper bounds on the one-shot private capacity of a channel}

In this section, we discuss the application of the smooth-min unextendible entanglement and the max-unextendible entanglement of channels to obtaining upper bounds on the one-shot, one-way distillable key of a channel.

\subsubsection{Smooth-min unextendible entanglement upper bound}

Consider the following quantity:
\begin{equation}\label{eq:J_eps_ch_defn}
    J^{\varepsilon}_{\min}\!\left(\mathcal{N}_{A\to B}\right) \coloneqq 2^{-2E^{u,\varepsilon}_{\min}(\mathcal{N}_{A\to B})}.
\end{equation}
Lemma~\ref{lem:unext_ent_state_le_unext_ent_ch_gen} implies the following inequality:
\begin{equation}\label{eq:J_eps_ch_le_st}
    J^{\varepsilon}_{\min}\!\left(\mathcal{N}_{A\to B}\right) \le \sup_{\substack{\rho_{A'A''\hat{A}}\in \mathcal{S}(A'A''\hat{A})\\ \Theta \in \operatorname{1WL}}} J^{\varepsilon}_{\min}\!\left(\left(\Theta_{(A\to B)\to (\hat{A}\to B'B'')}\!\left(\mathcal{N}_{A\to B}\right)\right)\!\left(\rho_{A'A''\hat{A}}\right)\right),
\end{equation}
where $J^{\varepsilon}_{\min}(\cdot)$ for states was defined in~\eqref{eq:J_hypo_test_defn}. The supremum in the above inequality is over every state $\rho_{A'A''\hat{A}}$ and one-way LOCC superchannel $\Theta_{(A\to B)\to (\hat{A}\to B'B'')}$. The inequality in~\eqref{eq:J_eps_ch_le_st}, along with Theorem~\ref{theo:distill_key_st_ub_hypo_test} and Lemma~\ref{lem:monotoncity_obj_func}, yields an upper bound on the one-shot forward-assisted private capacity of a channel, which we state in Theorem~\ref{theo:dist_key_ch_hypo_test_bnd} below.

\begin{theorem}[Unextendibility bound on one-shot private capacity] \label{theo:dist_key_ch_hypo_test_bnd}
    Consider a quantum channel $\mathcal{N}_{A\to B}$ and a parameter $\varepsilon \in [0,1]$ such that
    \begin{equation}
        J^{\varepsilon}_{\min}\!\left(\mathcal{N}_{A\to B}\right) > \varepsilon,
    \end{equation}
    where $J^{\varepsilon}_{\min}(\mathcal{N}_{A\to B})$ is defined in~\eqref{eq:J_eps_ch_defn}. Then the one-shot, one-way distillable key of the channel, which is equal to the one-shot private capacity of the channel according to~\eqref{eq:dist_key_eq_priv_cap}, is bounded from above by the following quantity:
    \begin{equation}\label{eq:dist_key_ch_hypo_test_bnd}
        P^{\varepsilon,\to}\!\left(\mathcal{N}_{A\to B}\right) = K^{\varepsilon,\to}_{D}\!\left(\mathcal{N}_{A\to B}\right) \le \frac{1}{2}\log_2\!\left[\left(\frac{\sqrt{J^{\varepsilon}_{\min}(\mathcal{N})(1-J^{\varepsilon}_{\min}(\mathcal{N}))}+\sqrt{\varepsilon(1-\varepsilon)}}{J^{\varepsilon}_{\min}(\mathcal{N})-\varepsilon}\right)^2+1\right].
    \end{equation}
\end{theorem}
\begin{proof}
    Let $\mathcal{N}_{A\to B}$ be a quantum channel, and let $\varepsilon \in [0,1]$ be a parameter such that
    \begin{equation}\label{eq:ch_J_eps_cond}
        J^{\varepsilon}_{\min}\!\left(\mathcal{N}_{A\to B}\right) > \varepsilon.
    \end{equation}
    Using the equality relating the one-shot, one-way distillable key of a channel and the one-shot, one-way distillable key of a state from~\eqref{eq:ch_1shot_1w_key_wrt_st}, we arrive at the following:
    \begin{align}
        K^{\varepsilon,\to}_D\!\left(\mathcal{N}_{A\to B}\right)
        &= \sup_{\Theta \in \operatorname{1WL}, \psi_{A'A''\hat{A}}\in \mathcal{S}(A'A''\hat{A})} K^{\varepsilon,\to}_D\!\left(\Theta\!\left(\mathcal{N}_{A\to B}\right)\!\left(\psi_{A'A''\hat{A}}\right)\right)\\
        &\le \sup_{\substack{\psi_{A'A''\hat{A}}\in \mathcal{S}(A'A''\hat{A}),\\\Theta \in \operatorname{1WL}} } \left\{
        \begin{array}{c}
            \frac{1}{2}\log_2\!\left[\left(\frac{\sqrt{J^{\varepsilon}_{\min}(1-J^{\varepsilon}_{\min})}+\sqrt{\varepsilon(1-\varepsilon)}}{J^{\varepsilon}_{\min}-\varepsilon}\right)^2+1\right]:\\
              J^{\varepsilon}_{\min} = J^{\varepsilon}_{\min}\!\left(\left(\Theta\!\left(\mathcal{N}\right)\right)\!\left(\psi_{A'A''\hat{A}}\right)\right)
        \end{array}\right\} ,\label{eq:dist_key_ch_st_bnd_1}
    \end{align}
    where $J^{\varepsilon}_{\min}(\cdot)$ for states is defined in~\eqref{eq:J_hypo_test_defn}. The inequality in~\eqref{eq:dist_key_ch_st_bnd_1} follows from Theorem~\ref{theo:distill_key_st_ub_hypo_test}. Note that the above inequality holds only if $J^{\varepsilon}_{\min} > \varepsilon$, but since we have assumed that $J^{\varepsilon}_{\min}(\mathcal{N})>\varepsilon$, the quantity $J^{\varepsilon}_{\min}$ is guaranteed to be strictly greater than $\varepsilon$ due to~\eqref{eq:J_eps_ch_le_st}.

    Let us define
    \begin{equation}
        J^{\varepsilon,s}_{\min}\!\left(\mathcal{N}_{A\to B}\right) \coloneqq \sup_{\substack{\psi_{A'A''\hat{A}}\in \mathcal{S}(A'A''\hat{A}),\\\Theta \in \operatorname{1WL}} } J^{\varepsilon}_{\min}\!\left(\left(\Theta\!\left(\mathcal{N}\right)\right)\!\left(\psi_{A'A''\hat{A}}\right)\right).
    \end{equation}
    Then the inequality in~\eqref{eq:dist_key_ch_st_bnd_1} can be written as follows:
    \begin{equation}
        K^{\varepsilon,\to}_D\!\left(\mathcal{N}_{A\to B}\right) \le \frac{1}{2}\log_2\!\left[\left(\frac{\sqrt{J^{\varepsilon,s}_{\min}\!\left(\mathcal{N}\right)\left(1-J^{\varepsilon,s}_{\min}\!\left(\mathcal{N}\right)\right)}+\sqrt{\varepsilon(1-\varepsilon)}}{J^{\varepsilon,s}_{\min}\!\left(\mathcal{N}\right)-\varepsilon}\right)^2+1\right].
    \end{equation}
    The inequality in~\eqref{eq:J_eps_ch_le_st} states that $J^{\varepsilon}_{\min}(\mathcal{N}) \le J^{\varepsilon,s}_{\min}(\mathcal{N})$. Therefore, by applying Lemma~\ref{lem:monotoncity_obj_func}, we arrive at the following inequality:
    \begin{equation}
        K^{\varepsilon,\to}_D\!\left(\mathcal{N}_{A\to B}\right) \le \frac{1}{2}\log_2\!\left[\left(\frac{\sqrt{J^{\varepsilon}_{\min}\!\left(\mathcal{N}\right)\left(1-J^{\varepsilon}_{\min}\!\left(\mathcal{N}\right)\right)}+\sqrt{\varepsilon(1-\varepsilon)}}{J^{\varepsilon}_{\min}\!\left(\mathcal{N}\right)-\varepsilon}\right)^2+1\right].
    \end{equation}
    Finally, using the fact that the one-shot forward-assisted private capacity of a channel is equal to the one-shot one-way distillable key of the channel, we conclude the statement of the theorem.
\end{proof}

\medskip

In Figure~\ref{fig_eras_dist_key}, we plot the upper bounds on the one-shot forward-assisted private capacity of a two-dimensional and a three-dimensional erasure channel for different erasure probabilities and different values of $\varepsilon$. The erasure channel is mathematically defined as follows:
\begin{equation}\label{eq:eras_ch_param}
    \mathcal{E}^p_{A\to B}\!\left(\rho_{RA}\right) = (1-p)\rho_{RB} + p\operatorname{Tr}_A\!\left[\rho_{RA}\right]\otimes |e\rangle\!\langle e|_B,
\end{equation}
where $|e\rangle_B$ is the erasure symbol, which is orthogonal to every state in the span of $\left\{|i\rangle\!\langle j|\right\}_{i,j=0}^{d-1}$, and $d$ is the dimension of the system $A$. The parameter $p \in [0,1]$ is the erasure probability of the channel.

\begin{figure}
    \centering
    \begin{subfigure}{0.45\linewidth}
        \includegraphics[width=\linewidth]{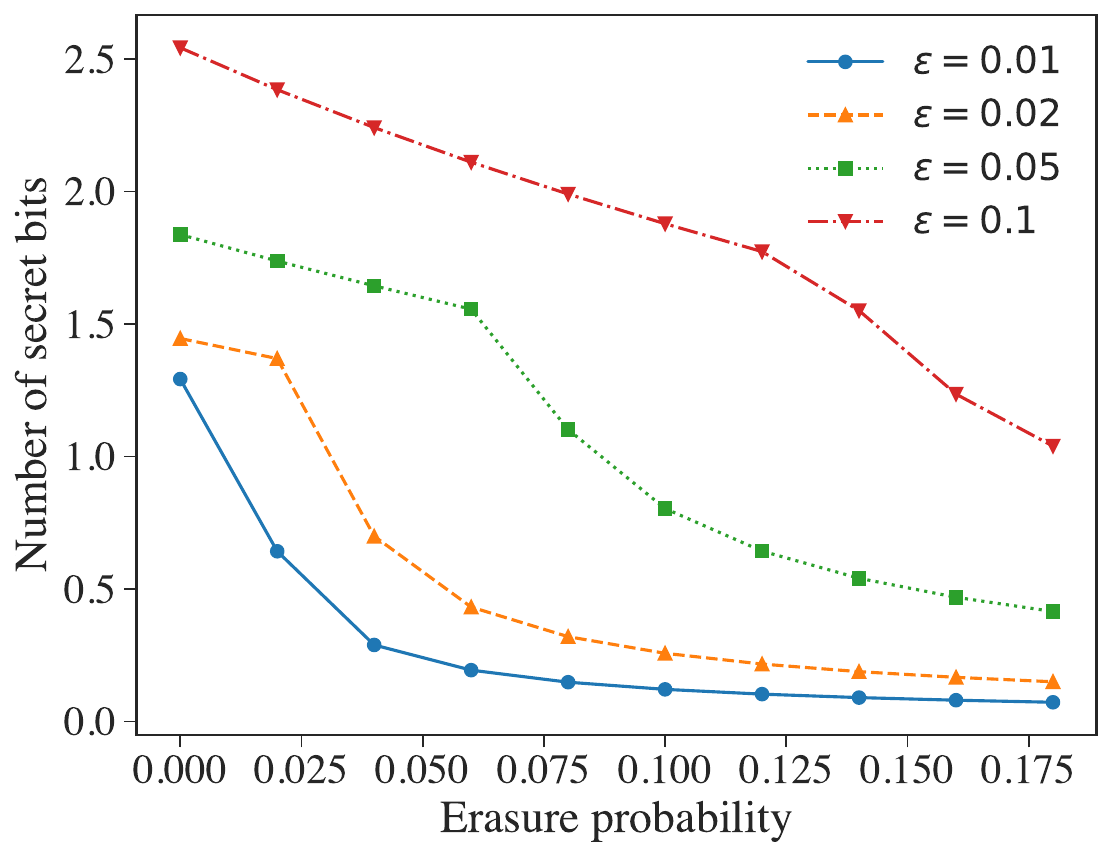}
        \caption{\centering Two-dimensional erasure channel}
        \label{fig:eras_dist_key_2d}
    \end{subfigure}
    \begin{subfigure}{0.45\linewidth}
        \includegraphics[width=\linewidth]{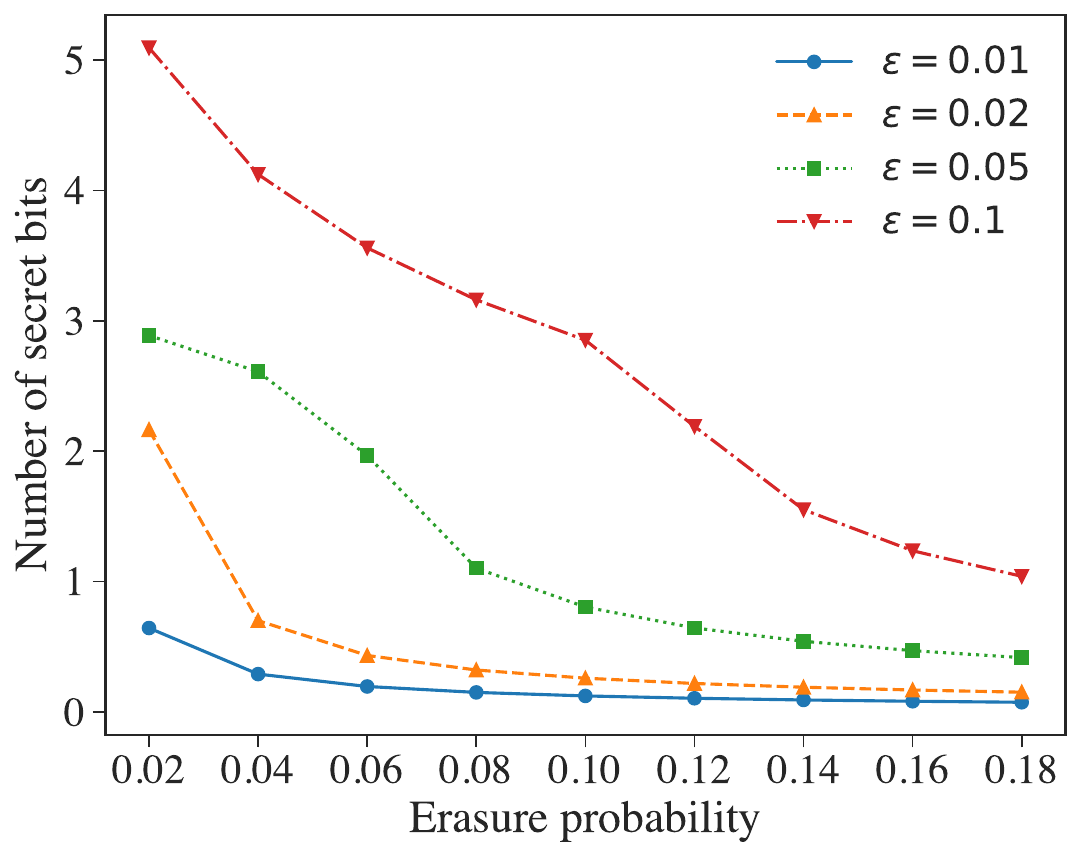}
        \caption{\centering Three-dimensional erasure channel}
        \label{fig:eras_dist_key_3d}
    \end{subfigure}
    \caption{Upper bound on the number of private bits that can be transmitted over a single use of an erasure channel assisted by local operations and forward-classical communication. The upper bound  given in Theorem~\ref{theo:dist_key_ch_hypo_test_bnd} is plotted against the erasure probability of an erasure channel for different values of $\varepsilon$.}
    \label{fig_eras_dist_key}
    
\end{figure}

The simplified upper bounds obtained in Corollaries~\ref{cor:dist_key_st_alg_ub} and~\ref{cor:dist_key_st_td_ub} can be used to obtain simplified upper bounds on the one-shot forward-assisted private capacity of a channel. The inequality in~\eqref{eq:J_eps_ch_le_st}, the equality in~\eqref{eq:ch_1shot_1w_key_wrt_st}, and the application of Lemma~\ref{lem:monotoncity_obj_func} together lead to simplified upper bounds on the one-shot forward-assisted private capacity of a channel, stated in the corollary below.
\begin{corollary}\label{cor:smooth_min_simplified_bnd_channels}
    Fix $\varepsilon \in [0,1]$. Let $\mathcal{N}_{A\to B}$ be a channel such that the following inequality holds:
    \begin{equation}
        J^{\varepsilon}_{\min}\!\left(\mathcal{N}_{A\to B}\right) > \varepsilon,
    \end{equation}
    where $J^{\varepsilon}_{\min}(\mathcal{N})$ is defined in~\eqref{eq:J_eps_ch_defn}. Then the one-shot forward-assisted private capacity of a channel is bounded from above as follows:
    \begin{equation}
        P^{\varepsilon,\to}\!\left(\mathcal{N}_{A\to B}\right) \le -\log_2\!\left(\sqrt{J^{\varepsilon}_{\min}\!\left(\mathcal{N}_{A\to B}\right)} - \sqrt{\varepsilon}\right).
    \end{equation}
    If $J^{\varepsilon}_{\min}(\mathcal{N}) > \sqrt{\varepsilon}$ then the following inequality also holds:
    \begin{equation}
        P^{\varepsilon,\to}\!\left(\mathcal{N}_{A\to B}\right) \le -\frac{1}{2}\log_2\!\left(J^{\varepsilon}_{\min}\!\left(\mathcal{N}_{A\to B}\right) - \sqrt{\varepsilon}\right).
    \end{equation}
    
\end{corollary}

\subsubsection{\texorpdfstring{$\alpha$}{alpha}-Geometric unextendible entanglement upper bound}

The subadditivity of the $\alpha$-geometric unextendible entanglement of channels, as given in~\eqref{eq:alpha-geo-subadd}, can be used to obtain an upper bound on the $n$-shot, forward-assisted private capacity of a channel.

Consider an arbitrary quantum channel $\mathcal{N}_{A\to B}$. Recall the definition of $J^{\varepsilon}_{\min}(\mathcal{N}_{A\to B})$ from~\eqref{eq:J_eps_ch_defn}. The inequality in~\eqref{eq:smooth_min_vs_geo_unext_ent_ineq} implies that the following inequality holds for all $\alpha \in (1,2]$ and $\varepsilon \in [0,1)$:
\begin{align}
    J^{\varepsilon}_{\min}\!\left(\mathcal{N}_{A\to B}\right) &\ge 2^{-2\widehat{E}^u_{\alpha}(\mathcal{N}) - \frac{\alpha}{\alpha - 1}\log_2\left(\frac{1}{1-\varepsilon}\right)}\\
    &= (1-\varepsilon)^{\frac{\alpha}{\alpha - 1}}2^{-2\widehat{E}^u_{\alpha}(\mathcal{N})}.
\end{align}
Now consider the following quantity:
\begin{align}
    J^{\varepsilon}_{\min}\!\left(\mathcal{N}^{\otimes n}_{A\to B}\right) &\ge (1-\varepsilon)^{\frac{\alpha}{\alpha - 1}}2^{-2\widehat{E}^u_{\alpha}(\mathcal{N}^{\otimes n})}\\
    &\ge (1-\varepsilon)^{\frac{\alpha}{\alpha - 1}}2^{-2n\widehat{E}^u_{\alpha}(\mathcal{N})},\label{eq:geo_ch_J_pre_defn}
\end{align}
where the second inequality follows from the subadditivity of the $\alpha$-geometric unextendible entanglement of channels (see~\eqref{eq:alpha-geo-subadd}). 

Let us define the following quantity:
\begin{equation}\label{eq:geo_ch_J_defn}
    \widehat{J}^{\varepsilon,n}_{\alpha}\!\left(\mathcal{N}_{A\to B}\right) \coloneqq (1-\varepsilon)^{\frac{\alpha}{\alpha - 1}}2^{-2n\widehat{E}^u_{\alpha}(\mathcal{N})},
\end{equation}
which, according to~\eqref{eq:geo_ch_J_pre_defn}, is a lower bound on $J^{\varepsilon}_{\min}(\mathcal{N}_{A\to B})$. If $\alpha$ is a rational number in the interval $(1,2]$, then $\widehat{J}^{\varepsilon,n}_{\alpha}\!\left(\mathcal{N}_{A\to B}\right)$ can be computed using a semidefinite program, by employing the algorithms given in~\cite{FS17} (see Appendix~\ref{app:semidefinite_programs} for a special case). We present the semidefinite program The application of Lemma~\ref{lem:monotoncity_obj_func} to Theorem~\ref{theo:dist_key_ch_hypo_test_bnd}, along with the inequality in~\eqref{eq:geo_ch_J_pre_defn}, directly leads to a single-letter, semidefinite computable upper bound on the $n$-shot forward-assisted private capacity of a channel, which we state formally in Corollary~\ref{cor:dist_key_ch_geo_bnd} below.
\begin{corollary}\label{cor:dist_key_ch_geo_bnd}
    Fix $\varepsilon \in (0,1)$. Let $\mathcal{N}_{A\to B}$ be a quantum channel  such that the following inequality holds for some $\alpha \in (1,2]$:
    \begin{equation}
        \widehat{J}^{\varepsilon,n}_{\alpha}\!\left(\mathcal{N}_{A\to B}\right) > \varepsilon
    \end{equation}
   where $\widehat{J}^{\varepsilon,n}_{\alpha}\!\left(\mathcal{N}_{A\to B}\right)$ is defined in~\eqref{eq:geo_ch_J_defn}. Then the $n$-shot forward-assisted private capacity of a channel $\mathcal{N}_{A\to B}$ is bounded from above by the following quantity:
	\begin{equation}\label{eq:dist_key_geo_unext_ent_ch_ub}
		P^{\varepsilon,\to}\!\left(\mathcal{N}^{\otimes n}_{A\to B}\right) \le \frac{1}{2}\log_2\!\left[\left(\frac{\sqrt{\widehat{J}^{\varepsilon,n}_{\alpha}\!\left(\mathcal{N}_{A\to B}\right)\!\left(1-\widehat{J}^{\varepsilon,n}_{\alpha}\!\left(\mathcal{N}_{A\to B}\right)\right)}+\sqrt{\varepsilon(1-\varepsilon)}}{\widehat{J}^{\varepsilon,n}_{\alpha}\!\left(\mathcal{N}_{A\to B}\right)-\varepsilon}\right)^2+1\right].
	\end{equation}
\end{corollary}

We turn to the erasure channel once again to demonstrate our results stated in Corollary~\ref{cor:dist_key_ch_geo_bnd}. An erasure channel with erasure probability greater than or equal to $\frac{1}{2}$ (see~\eqref{eq:eras_ch_param}) is a two-extendible channel, and hence, its $\alpha$-geometric unextendible entanglement is equal to zero for all $\alpha \in (0,1)\cup(1,2]$. If the erasure probability is less than $\frac{1}{2}$, then the explicit form of the $\alpha$-geometric unextendible entanglement can be derived, which we state in Proposition~\ref{prop:geo_unext_eras_alpha} below.

\begin{proposition}
\label{prop:geo_unext_eras_alpha}
    For all $\alpha \in (0,1)\cup(1,2]$, the $\alpha$-geometric unextendible entanglement of a $d$-dimensional erasure channel, with erasure probability $p$, evaluates to the following:
    \begin{equation}
        \widehat{E}^u_{\alpha}\!\left(\mathcal{E}^p_{A\to B}\right) = \frac{1}{2}\cdot\frac{1}{\alpha - 1}\log_2\!\left(\left(p + \frac{b_{\operatorname{opt}}}{d^2}\right)^{1-\alpha}(1-p)^{\alpha} + (1-p-b_{\operatorname{opt}})^{1-\alpha}p^{\alpha}\right) 
    \end{equation}
    for all $p\in \left(0,\frac{1}{d^{1/\alpha}+1}\right]$, where 
    \begin{equation}
        b_{\operatorname{opt}} \coloneqq \frac{d^2\!\left((1-p)^2 - p^2d^{2/\alpha}\right)}{pd^{2/\alpha} + (1-p)d^2}.
    \end{equation} For all $p\in \left(\frac{1}{d^{1/\alpha}+1},\frac{1}{2}\right]$,
    \begin{equation}
        \widehat{E}^u_{\alpha}\!\left(\mathcal{E}^p_{A\to B}\right) = \frac{1}{2}\cdot\frac{1}{\alpha - 1}\log_2\!\left(p^{1-\alpha}(1-p)^{\alpha} + (1-p)^{1-\alpha}p^{\alpha}\right).
    \end{equation}
\end{proposition}
\begin{proof}
    See Appendix~\ref{app:geo_unext_ent_eras}.
\end{proof}

\medskip

\begin{figure}
    \centering
    \begin{subfigure}{0.45\linewidth}
        \includegraphics[width=\linewidth]{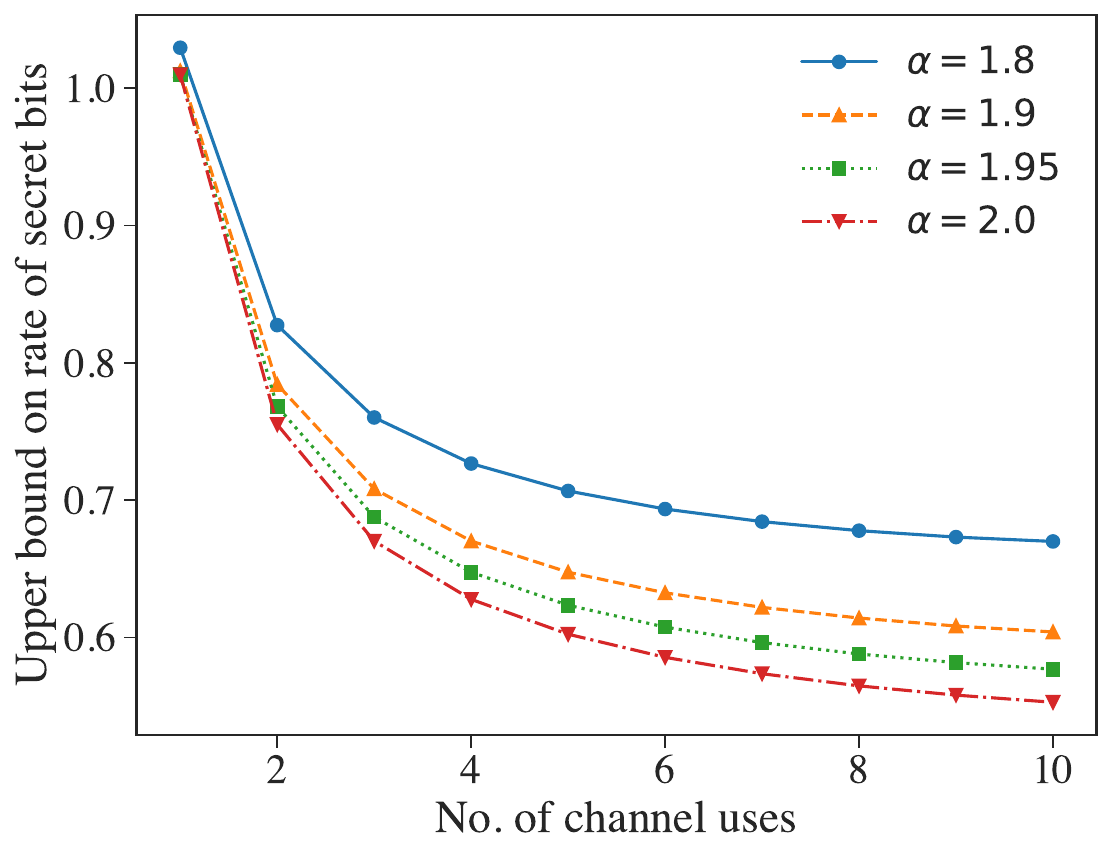}
        \caption{\centering Erasure probability = 0.2}
        \label{fig:eras_nshot_dist_key_p_02}
    \end{subfigure}
    \begin{subfigure}{0.45\linewidth}
        \includegraphics[width=\linewidth]{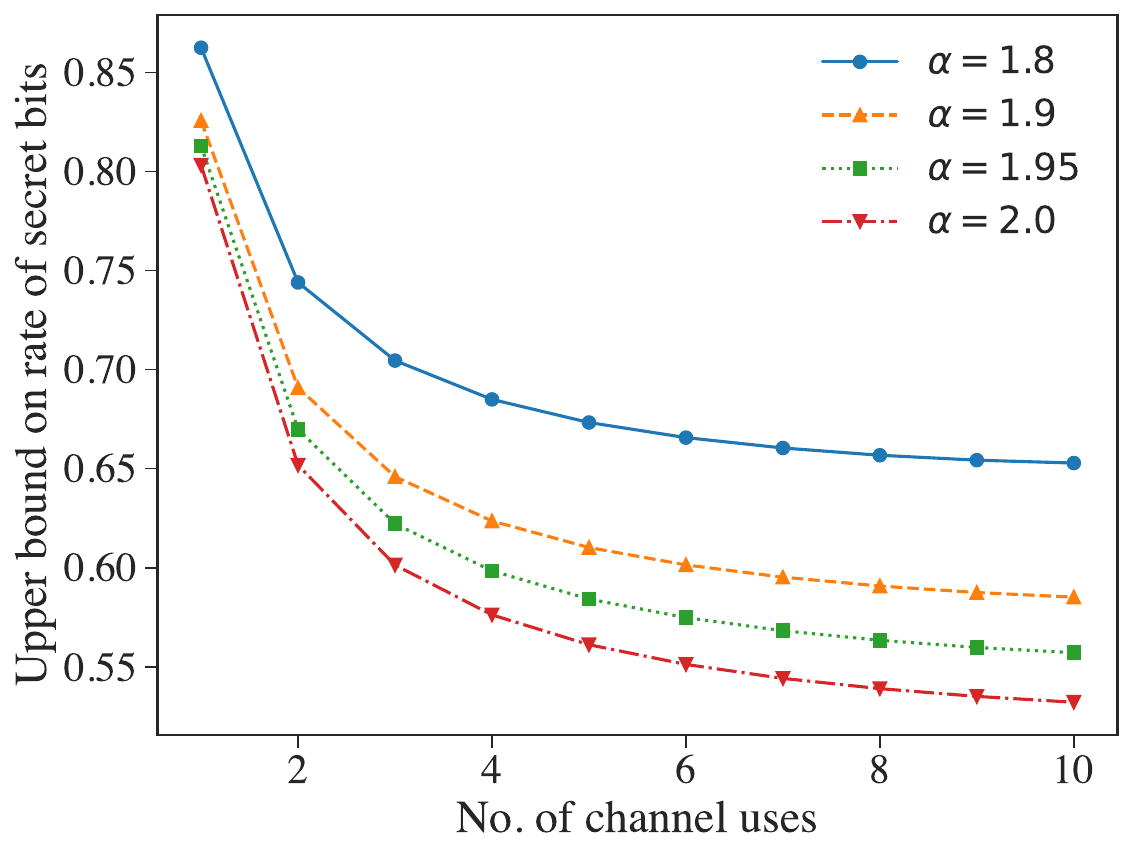}
        \caption{\centering Erasure probability = 0.3}
        \label{fig:eras_nshot_dist_key_p_03}
    \end{subfigure}
    \caption{Upper bound on the $n$-shot private capacity of an erasure channel with the error parameter $\varepsilon = 10^{-7}$. The bounds are computed for different values of $\alpha$ using Corollary~\ref{cor:dist_key_ch_geo_bnd}, where the  $\alpha$-geometric unextendible entanglement of the erasure channel is computed using Proposition~\ref{prop:geo_unext_eras_alpha}.}
    \label{fig_eras_nshot_dist_key}
    
\end{figure}

In Figure~\ref{fig_eras_nshot_dist_key} we plot the upper bound on the $n$-shot, forward-assisted private capacity of erasure channels computed using Corollary~\ref{cor:dist_key_ch_geo_bnd}. The $\alpha$-geometric unextendible entanglement of the erasure channel is computed using Proposition~\ref{prop:geo_unext_eras_alpha}. The $n$-shot forward-assisted private capacity of a channel is expected to increase with the number of channel uses. In Figure~\ref{fig_eras_nshot_dist_key} however, the computed value of the upper bound on the $n$-shot forward-assisted private capacity decreases with the number of channel uses, which indicates that the upper bound improves with an increasing number of channel uses.

\subsubsection{Private communication over a channel in the asymptotic setting}

In this section we study private communication over quantum channels using one-way LOCC superchannels in the asymptotic setting.

First, let us consider the task of secret-key distillation from quantum channels in a setting similar to the one discussed in Section~\ref{sec:asymp_key_dist_states}. In the asymptotic setting, a one-way LOCC protocol to distill secret keys from a channel $\mathcal{N}_{A\to B}$ is described by a sequence of positive integers $\{k_n\}_{n\in \mathbb{N}}$, a sequence of bipartite private states $\left\{\gamma^{k_n}_{A'A''B'B''}\right\}_{n\in \mathbb{N}}$, a sequence of one-way LOCC superchannels $\left\{\Theta^{n,\operatorname{1WL}}_{(A^n\to B^n)\to (\hat{A}\to B'B'')}\right\}_{n\in \mathbb{N}}$, a sequence of states $\left\{\rho^n_{A'A''\hat{A}}\right\}_{n\in \mathbb{N}}$, and a sequence of error parameters $\left\{\varepsilon_n\right\}_{n\in \mathbb{N}}$. A sequence of tuples, $\left\{\left(k_n, \gamma^{k_n}_{A'A''B'B''}, \Theta^{n,\operatorname{1WL}}_{(A^n\to B^n)\to (\hat{A}\to B'B'')}, \rho^n_{A'A''\hat{A}}, \varepsilon_n\right)\right\}_{n\in \mathbb{N}}$, describes a  one-way LOCC secret-key distillation protocol for a channel $\mathcal{N}_{A\to B}$ if the following inequality holds for all $n\in \mathbb{N}$:
\begin{equation}
    F\!\left(\gamma^{k_n},\left(\Theta^{n,\operatorname{1WL}}(\mathcal{N}^{\otimes n})\right)(\rho^n) \right) \ge 1-\varepsilon_n,
\end{equation}
and $\varepsilon_n \to 0$ as $n\to \infty$. We are interested in the maximum rate at which secret bits can be distilled from a channel $\mathcal{N}_{A\to B}$ using any such one-way LOCC secret-key distillation protocol.

Similar to the discussion in Section~\ref{sec:asymp_key_dist_states}, we impose an additional constraint on the sequence of error parameters $\varepsilon_n$, that $\varepsilon_n \le 2^{-an}$ for some fixed error exponent $a > 0$. We define the quantity ``$a$-exponential one-way distillable key of a channel'' as the maximum rate at which secret bits can be distilled from a channel using a one-way LOCC secret-key distillation protocol, with $a$ being the error exponent. 
\begin{definition}
    Fix $a > 0$. The $a$-exponential one-way distillable key of a channel $\mathcal{N}_{A\to B}$ is defined as follows:
    \begin{multline}
        K^{\to}_{D,a}\!\left(\mathcal{N}_{A\to B}\right) \coloneqq\\ \sup_{\substack{\left\{k_n\right\}_{n\in \mathbb{N}},\left\{\gamma^{k_n}_{A'B'A''B''}\right\}_{n\in \mathbb{N}}\\\left\{\Theta^{n,\operatorname{1WL}}\right\}_{n\in \mathbb{N}}, \left\{\rho^n_{A'A''\hat{A}}\right\}_{n\in \mathbb{N}} }} \liminf_{n\to \infty}\left\{\frac{\log_2 k_n}{n}: F\!\left(\gamma^{k_n},\Theta^{n,\operatorname{1WL}}\!\left(\mathcal{N}^{\otimes n}\right)(\rho^n)\right)\ge 1-2^{-an}\right\},
    \end{multline}
    where the supremum is over all sequences of integers $\{k_n\}_{n\in \mathbb{N}}$, all sequences of bipartite states $\left\{\gamma^{k_n}_{A'B'A''B''}\right\}_{n\in \mathbb{N}}$, all sequences of one-way LOCC superchannels $\left\{\Theta^{n,\operatorname{1WL}}_{(A^n\to B^n)\to A'B'A''B''}\right\}_{n\in \mathbb{N}}$, and all sequences of states $\left\{\rho_{A'A''\hat{A}}\right\}_{n\in \mathbb{N}}$. The bipartite state $\gamma^{k_n}_{A'B'A''B''}$ holds $\log_2 k_n$ secret bits.

    We define the converse of $a$-exponential one-way distillable key of the channel $\mathcal{N}_{A\to B}$ as follows:
    \begin{multline}
        \widetilde{K}^{\to}_{D,a}\!\left(\mathcal{N}_{A\to B}\right) \coloneqq \\ \sup_{\substack{\left\{k_n\right\}_{n\in \mathbb{N}},\left\{\gamma^{k_n}_{A'B'A''B''}\right\}_{n\in \mathbb{N}}\\\left\{\Theta^{n,\operatorname{1WL}}\right\}_{n\in \mathbb{N}}, \left\{\rho^n_{A'A''\hat{A}}\right\}_{n\in \mathbb{N}} }} \limsup_{n\to \infty}\left\{\frac{\log_2 k_n}{n}: F\!\left(\gamma^{k_n},\Theta^{n,\operatorname{1WL}}\!\left(\mathcal{N}^{\otimes n}\right)(\rho^n)\right)\ge 1-2^{-an}\right\}.
    \end{multline}
\end{definition}

We can also consider the task of private communication over channels in this setting.
We define the $a$-exponential forward-assisted private capacity of a channel as follows:
\begin{definition}
        Fix $a > 0$. The $a$-exponential forward-assisted private capacity of a channel $\mathcal{N}_{A\to B}$ is defined as follows:
    \begin{equation}
        P^{\to}_{a}\!\left(\mathcal{N}_{A\to B}\right) \coloneqq\sup_{\substack{\left\{\mathcal{X}_n\right\}_{n\in \mathbb{N}},\left\{\Theta^{n,\operatorname{1WL}}\right\}_{n\in \mathbb{N}} }} \liminf_{n\to \infty}\left\{\frac{\log_2 |\mathcal{X}_n|}{n}: p_{\operatorname{err}}\!\left(\mathcal{N}^{\otimes n},\Theta^{n,\operatorname{1WL}}\right)\le 2^{-an}\right\},
    \end{equation}
    where the supremum is over all sequences of message sets $\{\mathcal{X}_n\}_{n\in \mathbb{N}}$ and all sequences of one-way LOCC superchannels $\left\{\Theta^{n,\operatorname{1WL}}_{(A^n\to B^n)\to (X\to \hat{X})}\right\}_{n\in \mathbb{N}}$.

    We define the converse of $a$-exponential forward-assisted private capacity of the channel $\mathcal{N}_{A\to B}$ as follows:
    \begin{equation}
        \widetilde{P}^{\to}_{a}\!\left(\mathcal{N}_{A\to B}\right) \coloneqq\sup_{\substack{\left\{\mathcal{X}_n\right\}_{n\in \mathbb{N}},\left\{\Theta^{n,\operatorname{1WL}}\right\}_{n\in \mathbb{N}} }} \limsup_{n\to \infty}\left\{\frac{\log_2 |\mathcal{X}_n|}{n}: p_{\operatorname{err}}\!\left(\mathcal{N}^{\otimes n},\Theta^{n,\operatorname{1WL}}\right)\le 2^{-an}\right\}.
    \end{equation}
\end{definition}

Using the arguments mentioned before~\eqref{eq:dist_key_eq_priv_cap}, we can see that the following equalities hold:
\begin{align}
    K^{\to}_{D,a}\!\left(\mathcal{N}_{A\to B}\right) &= P^{\to}_{a}\!\left(\mathcal{N}_{A\to B}\right)\\
    \widetilde{K}^{\to}_{D,a}\!\left(\mathcal{N}_{A\to B}\right) &= \widetilde{P}^{\to}_{a}\!\left(\mathcal{N}_{A\to B}\right).\label{eq:a_exp_dist_ke_eq_priv_cap}
\end{align}
Therefore, any bounds obtained on the $a$-exponential distillable key of a channel hold for the $a$-exponential private capacity of the channel as well.

\begin{theorem}\label{theo:dist_key_asymptotic_bnd_channels}
    Consider an arbitrary quantum channel $\mathcal{N}_{A\to B}$. Let $d \coloneqq \min\{\operatorname{dim}(A),\operatorname{dim}(B)\}$ with $\operatorname{dim}(A)$ and $\operatorname{dim}(B)$ being the dimensions of systems $A$ and $B$, respectively. Fix $a\in (2\log_2d , \infty)$. Then the following bound holds:
    \begin{equation}
        \widetilde{P}^{\to}_{a}\!\left(\mathcal{N}_{A\to B}\right) = \widetilde{K}^{\to}_{D,a}\!\left(\mathcal{N}_{A\to B}\right) \le \widehat{E}^u\!\left(\mathcal{N}_{A\to B}\right),
    \end{equation}
    where $\widehat{E}^u(\mathcal{N}_{A\to B})$ is the unextendible entanglement of the channel $\mathcal{N}_{A\to B}$ induced by the Belavkin--Staszewski relative entropy defined in~\eqref{eq:Bel_Stas_unext_ent}.
\end{theorem}

\begin{proof}
    The proof is similar to the proof of Theorem~\ref{theo:dist_key_asymptotic_bnd}. We sketch out the main arguments here.
    
    Let $\mathcal{N}_{A\to B}$ be an arbitrary quantum channel, with input and output dimensions $\operatorname{dim}(A)$ and $\operatorname{dim}(B)$ respectively, from which we wish to distill secret keys. Let $\left\{\Theta^{n,\operatorname{1WL}}_{(A^n\to B^n)\to (\hat{A}\to B'B'')}\right\}_{n\in \mathbb{N}}$ be a sequence of one-way LOCC superchannels, let $\left\{\gamma^{k_n}_{A'B'A''B''}\right\}_{n\in \mathbb{N}}$ be a sequence of bipartite private states, and let $\left\{\rho^n_{A'A''\hat{A}}\right\}_{n\in \mathbb{N}}$ be a sequence of quantum states such that the following condition holds for all $a>2\log_2 d$ and $n\in \mathbb{N}$:
    \begin{equation}\label{eq:priv_st_fid_exponential_ch}
        F\!\left(\gamma^{k_n}_{A'B'A''B''},\left(\Theta^{n,\operatorname{1WL}}\!\left(\mathcal{N}^{\otimes n}_{A\to B}\right)\right)\!\left(\rho^n_{A'A''\hat{A}}\right)\right) \ge 1-2^{-an},
    \end{equation}
    where $d \coloneqq \min\{\operatorname{dim}(A),\operatorname{dim}(B)\}$.

    Let us set $\varepsilon_n\coloneqq 2^{-an}$ for convenience. Corollary~\ref{cor:smooth_min_simplified_bnd_channels} implies that the following inequality holds for all one-way LOCC secret-key distillation protocols such that~\eqref{eq:priv_st_fid_exponential_ch} is satisfied:
    \begin{align}
        \log_2k_n &\le -\log_2\!\left(\sqrt{J^{\varepsilon_n}_{\min}\!\left(\mathcal{N}^{\otimes n}\right)} - \sqrt{\varepsilon_n}\right)\\
        &= -\frac{1}{2}\log_2\!\left(J^{\varepsilon_n}_{\min}\!\left(\mathcal{N}^{\otimes n}\right)\right) - \log_2\!\left(1-\sqrt{\frac{\varepsilon_n}{J^{\varepsilon_n}_{\min}\!\left(\mathcal{N}^{\otimes n}\right)}}\right).\label{eq:secret_bits_ub_asymp_setup}
    \end{align}
    The quantity $J^{\varepsilon_n}_{\min}(\mathcal{N}^{\otimes n})$ is bounded from below by the following quantity:
    \begin{equation}\label{eq:J_eps_lb_ch}
        J^{\varepsilon_n}_{\min}\!\left(\mathcal{N}^{\otimes n}_{A\to B}\right) \ge \frac{1-\varepsilon_n}{d^{2n}},
    \end{equation}
    which is evident from Proposition~\ref{prop:smooth_min_ch_range}. The inequalities in~\eqref{eq:secret_bits_ub_asymp_setup} and~\eqref{eq:J_eps_lb_ch} allow us to use the mathematical arguments presented in~\eqref{eq:eps_J_ratio_ub_step_1}--\eqref{eq:eps_J_ratio_log_final} in order to conclude the following:
    \begin{equation}\label{eq:eps_J_ratio_log_ch_0}
        \limsup_{n\to \infty}-\frac{1}{n}\log_2\!\left(1-\sqrt{\frac{\varepsilon_n}{J^{\varepsilon_n}_{\min}\!\left(\mathcal{N}^{\otimes n}\right)}}\right) = 0.
    \end{equation}
    Therefore, taking $\limsup_{n\to \infty}$ in~\eqref{eq:secret_bits_ub_asymp_setup} leads to the following inequality:
    \begin{align}
        \limsup_{n\to \infty}\frac{\log_2 k_n}{n} &\le \limsup_{n\to \infty}\left\{-\frac{1}{2n}\log_2\!\left(J^{\varepsilon_n}_{\min}\!\left(\mathcal{N}^{\otimes n}\right)\right)-\frac{1}{n}\log_2\!\left(1-\sqrt{\frac{\varepsilon_n}{J^{\varepsilon_n}_{\min}\!\left(\mathcal{N}^{\otimes n}\right)}}\right)\right\}\\
        &= \limsup_{n\to \infty}-\frac{1}{2n}\log_2\!\left(J^{\varepsilon_n}_{\min}\!\left(\mathcal{N}^{\otimes n}\right)\right) + \limsup_{n\to \infty}-\frac{1}{n}\log_2\!\left(1-\sqrt{\frac{\varepsilon_n}{J^{\varepsilon_n}_{\min}\!\left(\mathcal{N}^{\otimes n}\right)}}\right)\\
        &= \limsup_{n\to \infty}-\frac{1}{2n}\log_2\!\left(J^{\varepsilon_n}_{\min}\!\left(\mathcal{N}^{\otimes n}\right)\right)\\
        &= \limsup_{n\to \infty} \frac{1}{n}E^{u,\varepsilon_n}_{\min}\!\left(\mathcal{N}^{\otimes n}\right),\label{eq:rate_ub_smooth_min}
    \end{align}
    where the second equality follows from~\eqref{eq:eps_J_ratio_log_ch_0}. 

    Recall the relation between the smooth-min unextendible entanglement of a channel and the $\alpha$-geometric unextendible entanglement of the channel from~\eqref{eq:smooth_min_vs_geo_unext_ent_ineq}. The inequality in~\eqref{eq:smooth_min_vs_geo_unext_ent_ineq} combined with~\eqref{eq:rate_ub_smooth_min} leads to the following inequality:
    \begin{align}
        \limsup_{n\to \infty}\frac{\log_2 k_n}{n} &\le \limsup_{n\to \infty} \frac{1}{n}\left(\widehat{E}^u_{\alpha}\!\left(\mathcal{N}^{\otimes n}\right) -\frac{1}{2}\cdot\frac{\alpha}{\alpha - 1}\log_2\!\left(1-\varepsilon_n\right)\right)\\
        &=\limsup_{n\to \infty} \frac{1}{n}\widehat{E}^u_{\alpha}\!\left(\mathcal{N}^{\otimes n}\right) + \limsup_{n\to \infty}-\frac{1}{2n}\cdot\frac{\alpha}{\alpha - 1}\log_2\!\left(1-\varepsilon_n\right),
    \end{align}
    which holds for all $\alpha \in (1,2]$. Since $\varepsilon_n \to 0$ as $n\to \infty$,
    \begin{equation}
        \limsup_{n\to \infty}-\frac{1}{2n}\cdot\frac{\alpha}{\alpha - 1}\log_2\!\left(1-\varepsilon_n\right) = 0.
    \end{equation}
    Now using the subadditivity of the $\alpha$-geometric unextendible entanglement of channels, we arrive at the following inequality:
    \begin{equation}
        \limsup_{n\to \infty}\frac{\log_2 k_n}{n} \le \limsup_{n\to \infty} \frac{1}{n}\widehat{E}^u_{\alpha}\!\left(\mathcal{N}^{\otimes n}\right) \le \widehat{E}^u_{\alpha}\!\left(\mathcal{N}\right).
    \end{equation}
    The above inequality holds for every sequence $\left\{k_n\right\}_{n\in \mathbb{N}}$ for which there exists a sequence of private states $\left\{\gamma^{k_n}_{A'B'A''B''}\right\}_{n\in \mathbb{N}}$, a sequence of one-way LOCC superchannels $\left\{\Theta^{n,\operatorname{1WL}}_{(A^n\to B^n)\to (\hat{A}\to B'B'')}\right\}_{n\in \mathbb{N}}$, and a sequence of states $\left\{\rho^n_{A'A''\hat{A}}\right\}_{n\in \mathbb{N}}$ such that~\eqref{eq:priv_st_fid_exponential_ch} holds for all $a>2\log_2 d$ and $n\in \mathbb{N}$. Therefore,
    \begin{equation}
        \widetilde{K}^{\to}_{D,a}\!\left(\mathcal{N}_{A\to B}\right) \le \widehat{E}^u_{\alpha}\!\left(\mathcal{N}_{A\to B}\right) \qquad \forall \alpha \in (1,2],
    \end{equation}
    which follows from the definition of $\widetilde{K}^{\to}_{D,a}\!\left(\mathcal{N}_{A\to B}\right)$. Since the $\alpha$-geometric unextendible entanglement of a channel increases monotonically with $\alpha$, we can take $\lim_{\alpha \to 1^+}$ to obtain the tightest upper bound, which is the unextendible entanglement of the channel $\mathcal{N}_{A\to B}$ induced by the Belavkin--Staszewski relative entropy. Finally, using~\eqref{eq:a_exp_dist_ke_eq_priv_cap} leads to the statement of the theorem.
\end{proof}

\section{Conclusion}

In this paper we studied the task of secret-key distillation from bipartite states and point-to-point quantum channels using local operations and one-way classical communication. Using the resource theory of unextendible entanglement, which is a semidefinite relaxation of the resource theory of entanglement, we obtained efficiently computable upper bounds on several quantities of interest in the theory of private communication over a quantum network.

We derived efficiently computable upper bounds on the one-shot, one-way distillable key of a bipartite state using the resource theory of unextendible entanglement. We also derived upper bounds on the one-shot forward-assisted private capacity of a channel that can be computed using a semidefinite program. In both cases, these are the first instances of efficiently computable upper bounds on these quantities, to the best of our knowledge. 

We extended our results to the i.i.d.~setting and obtained single-letter efficiently computable upper bounds on the $n$-shot one-way distillable key of bipartite states and $n$-shot forward-assisted private capacity of point-to-point channels. Finally, we obtained efficiently computable upper bounds on the rate at which secret keys can be distilled from a bipartite state or a quantum channel using one-way LOCC when the error is required to decay exponentially with an error exponent larger than a fixed threshold.  

The majority of bounds obtained in this work can be computed using semidefinite programs. We numerically computed the upper bounds on the one-shot, one-way distillable key and $n$-shot one-way distillable key for isotropic states to demonstrate our results. We also found analytical expressions for the upper bounds on the $n$-shot forward-assisted private capacity of erasure channels. 

We obtained a family of upper bounds on the $n$-shot, one-way distillable key of a bipartite state in this work using the $\alpha$-sandwiched R\'enyi relative entropy. However, a semidefinite representation of the $\alpha$-sandwiched R\'enyi relative entropy is only known when $\alpha\to \infty$. As such, only one member from the family of upper bounds on the $n$-shot, one-way distillable key of a state is known to be efficiently computable. 

Going forward from here, there are some open problems left for future investigation. The bounds obtained in this work are based on the resource theory of unextendible entanglement. It may be possible to obtain stronger bounds by studying entanglement measures that combine the concepts of unextendibility and the positive partial transpose (PPT) criterion. Furthermore, it can give insights into the asymptotic setting of private communication where there are no assumptions on the rate at which error decays. As another open problem of interest,  finding semidefinite representations of the $\alpha$-sandwiched R\'enyi relative entropies would improve our numerical findings here, as they can lead to tighter efficiently computable bounds on the $n$-shot, one-way distillable key of a state.  

\bigskip

\noindent \textbf{Acknowledgements}: VS and MMW acknowledge support from the National Science Foundation under Grant No.~2329662. MMW acknowledges Nilanjana Datta for many insightful discussions about this project, starting from a visit to her group at University of Cambridge in October 2022, and VS and MMW are also grateful to her for feedback on the manuscript. We also thank Kaiyuan Ji, Theshani Nuradha, Dhrumil Patel, and Aby Philip for helpful discussions. 

\bigskip

\noindent \textbf{Author Contributions}:
The following describes the different contributions of the authors of this work, using roles defined by the CRediT
(Contributor Roles Taxonomy) project~\cite{NISO}:

\medskip 
\noindent \textbf{VS}:
Formal Analysis, Investigation, Methodology, Software, Writing - Original draft, Validation, Writing - Review \& Editing.

\medskip 
\noindent \textbf{MMW}: Conceptualization, Formal Analysis, Funding acquisition,  Investigation, Methodology, Supervision, Validation,  Writing - Review \& Editing.

\bibliographystyle{alphaurl}
\bibliography{Ref}

\newcommand{\etalchar}[1]{$^{#1}$}
\begin{thebibliography}{MLDS{\etalchar{+}}13}

\bibitem[BB84]{BB84}
Charles~H. Bennett and Gilles Brassard.
\newblock Quantum cryptography: Public key distribution and coin tossing.
\newblock In {\em Proceedings of IEEE International Conference on Computers, Systems, and Signal Processing}, page 175, India, 1984.
\newblock \href {https://doi.org/10.1016/j.tcs.2014.05.025} {\path{doi:10.1016/j.tcs.2014.05.025}}.

\bibitem[BD10]{BD10}
Francesco Buscemi and Nilanjana Datta.
\newblock The quantum capacity of channels with arbitrarily correlated noise.
\newblock {\em IEEE Transactions on Information Theory}, 56(3):1447--1460, 2010.
\newblock \href {https://arxiv.org/abs/0902.0158} {\path{arXiv:0902.0158}}, \href {https://doi.org/10.1109/TIT.2009.2039166} {\path{doi:10.1109/TIT.2009.2039166}}.

\bibitem[BD11]{BD11}
Fernando G. S.~L. Brandao and Nilanjana Datta.
\newblock One-shot rates for entanglement manipulation under non-entangling maps.
\newblock {\em IEEE Transactions on Information Theory}, 57(3):1754--1760, 2011.
\newblock \href {https://arxiv.org/abs/0905.2673} {\path{arXiv:0905.2673}}, \href {https://doi.org/10.1109/TIT.2011.2104531} {\path{doi:10.1109/TIT.2011.2104531}}.

\bibitem[BS82]{BS82}
V.~P. Belavkin and P.~Staszewski.
\newblock C*-algebraic generalization of relative entropy and entropy.
\newblock {\em Annales de l'I.H.P. Physique théorique}, 37(1):51--58, 1982.
\newblock URL: \url{http://eudml.org/doc/76163}.

\bibitem[CDP08]{Chiribella_2008}
G.~Chiribella, G.~M. D{\textquotesingle}Ariano, and P.~Perinotti.
\newblock Transforming quantum operations: Quantum supermaps.
\newblock {\em Europhysics Letters}, 83(3):30004, July 2008.
\newblock \href {https://arxiv.org/abs/0804.0180} {\path{arXiv:0804.0180}}, \href {https://doi.org/10.1209/0295-5075/83/30004} {\path{doi:10.1209/0295-5075/83/30004}}.

\bibitem[CEH{\etalchar{+}}07]{CEHHOR07}
Matthias Christandl, Artur Ekert, Michal Horodecki, Pawel Horodecki, Jonathan Oppenheim, and Renato Renner.
\newblock Unifying classical and quantum key distillation.
\newblock In Salil~P. Vadhan, editor, {\em Theory of Cryptography}, pages 456--478, Berlin, Heidelberg, 2007. Springer Berlin Heidelberg.
\newblock \href {https://arxiv.org/abs/quant-ph/0608199} {\path{arXiv:quant-ph/0608199}}, \href {https://doi.org/10.1007/978-3-540-70936-7_25} {\path{doi:10.1007/978-3-540-70936-7_25}}.

\bibitem[Chr06]{Christandl06}
Matthias Christandl.
\newblock The structure of bipartite quantum states - insights from group theory and cryptography, 2006.
\newblock URL: \url{https://arxiv.org/abs/quant-ph/0604183}, \href {https://arxiv.org/abs/quant-ph/0604183} {\path{arXiv:quant-ph/0604183}}.

\bibitem[CMW16]{CMW16}
Tom Cooney, Mil{\'a}n Mosonyi, and Mark~M. Wilde.
\newblock Strong converse exponents for a quantum channel discrimination problem and quantum-feedback-assisted communication.
\newblock {\em Communications in Mathematical Physics}, 344(3):797--829, 2016.
\newblock \href {https://arxiv.org/abs/1408.3373} {\path{arXiv:1408.3373}}, \href {https://doi.org/10.1007/s00220-016-2645-4} {\path{doi:10.1007/s00220-016-2645-4}}.

\bibitem[CSW12]{CSW12}
Matthias Christandl, Norbert Schuch, and Andreas Winter.
\newblock Entanglement of the antisymmetric state.
\newblock {\em Communications in Mathematical Physics}, 311(2):397--422, 2012.
\newblock \href {https://arxiv.org/abs/0910.4151} {\path{arXiv:0910.4151}}, \href {https://doi.org/10.1007/s00220-012-1446-7} {\path{doi:10.1007/s00220-012-1446-7}}.

\bibitem[CW04]{CW04}
Matthias Christandl and Andreas Winter.
\newblock ``{S}quashed entanglement'': an additive entanglement measure.
\newblock {\em Journal of Mathematical Physics}, 45(3):829--840, 9/2/2024 2004.
\newblock \href {https://arxiv.org/abs/quant-ph/0308088} {\path{arXiv:quant-ph/0308088}}, \href {https://doi.org/10.1063/1.1643788} {\path{doi:10.1063/1.1643788}}.

\bibitem[CWY04]{CWY04}
N.~Cai, A.~Winter, and R.~W. Yeung.
\newblock Quantum privacy and quantum wiretap channels.
\newblock {\em Problems of Information Transmission}, 40(4):318--336, 2004.
\newblock \href {https://doi.org/10.1007/s11122-005-0002-x} {\path{doi:10.1007/s11122-005-0002-x}}.

\bibitem[Dat09]{Dat09}
Nilanjana Datta.
\newblock Min- and max-relative entropies and a new entanglement monotone.
\newblock {\em IEEE Transactions on Information Theory}, 55(6):2816--2826, 2009.
\newblock arXiv:0803.2770.
\newblock \href {https://doi.org/10.1109/TIT.2009.2018325} {\path{doi:10.1109/TIT.2009.2018325}}.

\bibitem[DFW{\etalchar{+}}18]{DFW18}
Mar{\'{i}}a~Garc{\'{i}}a D{\'{i}}az, Kun Fang, Xin Wang, Matteo Rosati, Michalis Skotiniotis, John Calsamiglia, and Andreas Winter.
\newblock Using and reusing coherence to realize quantum processes.
\newblock {\em {Quantum}}, 2:100, October 2018.
\newblock \href {https://arxiv.org/abs/1805.04045} {\path{arXiv:1805.04045}}, \href {https://doi.org/10.22331/q-2018-10-19-100} {\path{doi:10.22331/q-2018-10-19-100}}.

\bibitem[DKQ{\etalchar{+}}23]{DKQSWW23}
Dawei Ding, Sumeet Khatri, Yihui Quek, Peter~W. Shor, Xin Wang, and Mark~M. Wilde.
\newblock Bounding the forward classical capacity of bipartite quantum channels.
\newblock {\em IEEE Transactions on Information Theory}, 69(5):3034--3061, 2023.
\newblock arXiv:2010.01058.
\newblock \href {https://doi.org/10.1109/TIT.2022.3233924} {\path{doi:10.1109/TIT.2022.3233924}}.

\bibitem[DLL03]{DLL03}
Fu-Guo Deng, Gui~Lu Long, and Xiao-Shu Liu.
\newblock Two-step quantum direct communication protocol using the {Einstein--Podolsky--Rosen} pair block.
\newblock {\em Physical Review A}, 68(4):042317, October 2003.
\newblock arXiv:quant-ph/0308173.
\newblock \href {https://doi.org/10.1103/PhysRevA.68.042317} {\path{doi:10.1103/PhysRevA.68.042317}}.

\bibitem[DPS04]{DPS04}
Andrew~C. Doherty, Pablo~A. Parrilo, and Federico~M. Spedalieri.
\newblock Complete family of separability criteria.
\newblock {\em Physical Review A}, 69(2):022308, February 2004.
\newblock \href {https://arxiv.org/abs/quant-ph/0308032} {\path{arXiv:quant-ph/0308032}}, \href {https://doi.org/10.1103/PhysRevA.69.022308} {\path{doi:10.1103/PhysRevA.69.022308}}.

\bibitem[DW05]{DW05}
Igor Devetak and Andreas Winter.
\newblock Distillation of secret key and entanglement from quantum states.
\newblock {\em Proceedings of the Royal Society A: Mathematical, Physical and Engineering Sciences}, 461(2053):207–235, January 2005.
\newblock \href {https://doi.org/10.1098/rspa.2004.1372} {\path{doi:10.1098/rspa.2004.1372}}.

\bibitem[Eke91]{Ekert91}
Artur~K. Ekert.
\newblock Quantum cryptography based on {B}ell's theorem.
\newblock {\em Physical Review Letters}, 67(6):661--663, August 1991.
\newblock \href {https://doi.org/10.1103/PhysRevLett.67.661} {\path{doi:10.1103/PhysRevLett.67.661}}.

\bibitem[FF21]{FF21}
Kun Fang and Hamza Fawzi.
\newblock Geometric {R\'{e}}nyi divergence and its applications in quantum channel capacities.
\newblock {\em Communications in Mathematical Physics}, 384(3):1615--1677, May 2021.
\newblock \href {https://arxiv.org/abs/1909.05758} {\path{arXiv:1909.05758}}, \href {https://doi.org/10.1007/s00220-021-04064-4} {\path{doi:10.1007/s00220-021-04064-4}}.

\bibitem[FS17]{FS17}
Hamza Fawzi and James Saunderson.
\newblock Lieb's concavity theorem, matrix geometric means, and semidefinite optimization.
\newblock {\em Linear Algebra and its Applications}, 513:240--263, 2017.
\newblock arXiv:1512.03401.
\newblock URL: \url{https://www.sciencedirect.com/science/article/pii/S0024379516304852}, \href {https://doi.org/10.1016/j.laa.2016.10.012} {\path{doi:10.1016/j.laa.2016.10.012}}.

\bibitem[FvdG99]{FvdG98}
Christopher~A. Fuchs and Jeroen van~de Graaf.
\newblock Cryptographic distinguishability measures for quantum-mechanical states.
\newblock {\em IEEE Transactions on Information Theory}, 45(4):1216--1227, 1999.
\newblock \href {https://arxiv.org/abs/quant-ph/9712042} {\path{arXiv:quant-ph/9712042}}, \href {https://doi.org/10.1109/18.761271} {\path{doi:10.1109/18.761271}}.

\bibitem[Gha10]{G10}
Sevag Gharibian.
\newblock Strong {NP}-hardness of the quantum separability problem.
\newblock {\em Quantum Information and Computation}, 10(3):343--360, March 2010.
\newblock arXiv:0810.4507.
\newblock \href {https://doi.org/10.26421/qic10.3-4-11} {\path{doi:10.26421/qic10.3-4-11}}.

\bibitem[Gou19]{Gour_2019}
Gilad Gour.
\newblock Comparison of quantum channels by superchannels.
\newblock {\em {IEEE} Transactions on Information Theory}, 65(9):5880--5904, September 2019.
\newblock \href {https://arxiv.org/abs/1808.02607} {\path{arXiv:1808.02607}}, \href {https://doi.org/10.1109/tit.2019.2907989} {\path{doi:10.1109/tit.2019.2907989}}.

\bibitem[Gur03]{Gur03}
Leonid Gurvits.
\newblock Classical deterministic complexity of {E}dmonds' problem and quantum entanglement.
\newblock In {\em Proceedings of the Thirty-Fifth Annual ACM Symposium on Theory of Computing}, STOC '03, page 10–19, New York, NY, USA, 2003. Association for Computing Machinery.
\newblock arXiv:quant-ph/0303055.
\newblock \href {https://doi.org/10.1145/780542.780545} {\path{doi:10.1145/780542.780545}}.

\bibitem[HH99]{HH99}
Michal Horodecki and Pawel Horodecki.
\newblock Reduction criterion of separability and limits for a class of distillation protocols.
\newblock {\em Physical Review A}, 59(6):4206--4216, June 1999.
\newblock \href {https://arxiv.org/abs/quant-ph/9708015} {\path{arXiv:quant-ph/9708015}}, \href {https://doi.org/10.1103/PhysRevA.59.4206} {\path{doi:10.1103/PhysRevA.59.4206}}.

\bibitem[HHH{\etalchar{+}}08a]{HHHLO08}
Karol Horodecki, Michal Horodecki, Pawel Horodecki, Debbie Leung, and Jonathan Oppenheim.
\newblock Quantum key distribution based on private states: Unconditional security over untrusted channels with zero quantum capacity.
\newblock {\em IEEE Transactions on Information Theory}, 54(6):2604--2620, 2008.
\newblock \href {https://arxiv.org/abs/quant-ph/0608195} {\path{arXiv:quant-ph/0608195}}, \href {https://doi.org/10.1109/TIT.2008.921870} {\path{doi:10.1109/TIT.2008.921870}}.

\bibitem[HHH{\etalchar{+}}08b]{HHHLO08_QP}
Karol Horodecki, Michal Horodecki, Pawel Horodecki, Debbie Leung, and Jonathan Oppenheim.
\newblock Unconditional privacy over channels which cannot convey quantum information.
\newblock {\em Physical Review Letters}, 100:110502, March 2008.
\newblock \href {https://arxiv.org/abs/quant-ph/0702077} {\path{arXiv:quant-ph/0702077}}, \href {https://doi.org/10.1103/PhysRevLett.100.110502} {\path{doi:10.1103/PhysRevLett.100.110502}}.

\bibitem[HHHO05]{HHHO05}
Karol Horodecki, Michal Horodecki, Pawel Horodecki, and Jonathan Oppenheim.
\newblock Secure key from bound entanglement.
\newblock {\em Physical Review Letters}, 94(16):160502, April 2005.
\newblock \href {https://doi.org/10.1103/PhysRevLett.94.160502} {\path{doi:10.1103/PhysRevLett.94.160502}}.

\bibitem[HHHO09]{HHHO09}
Karol Horodecki, Michal Horodecki, Pawel Horodecki, and Jonathan Oppenheim.
\newblock General paradigm for distilling classical key from quantum states.
\newblock {\em {IEEE} Transactions on Information Theory}, 55(4):1898--1929, April 2009.
\newblock \href {https://arxiv.org/abs/quant-ph/0506189} {\path{arXiv:quant-ph/0506189}}, \href {https://doi.org/10.1109/tit.2008.2009798} {\path{doi:10.1109/tit.2008.2009798}}.

\bibitem[HMZ16]{HMZ16}
Teiko Heinosaari, Takayuki Miyadera, and Mário Ziman.
\newblock An invitation to quantum incompatibility.
\newblock {\em Journal of Physics A: Mathematical and Theoretical}, 49(12):123001, February 2016.
\newblock arXiv:1511.07548.
\newblock URL: \url{https://dx.doi.org/10.1088/1751-8113/49/12/123001}, \href {https://doi.org/10.1088/1751-8113/49/12/123001} {\path{doi:10.1088/1751-8113/49/12/123001}}.

\bibitem[HSW23]{HSW23}
Tharon Holdsworth, Vishal Singh, and Mark~M. Wilde.
\newblock Quantifying the performance of approximate teleportation and quantum error correction via symmetric 2-{PPT}-extendible channels.
\newblock {\em Physical Review A}, 107(1):012428, Jan 2023.
\newblock \href {https://arxiv.org/abs/2207.06931} {\path{arXiv:2207.06931}}, \href {https://doi.org/10.1103/PhysRevA.107.012428} {\path{doi:10.1103/PhysRevA.107.012428}}.

\bibitem[KDWW19]{KDWW19}
Eneet Kaur, Siddhartha Das, Mark~M. Wilde, and Andreas Winter.
\newblock Extendibility limits the performance of quantum processors.
\newblock {\em Physical Review Letters}, 123(7):070502, August 2019.
\newblock \href {https://arxiv.org/abs/2108.03137} {\path{arXiv:2108.03137}}, \href {https://doi.org/10.1103/physrevlett.123.070502} {\path{doi:10.1103/physrevlett.123.070502}}.

\bibitem[KDWW21]{KDWW21}
Eneet Kaur, Siddhartha Das, Mark~M. Wilde, and Andreas Winter.
\newblock Resource theory of unextendibility and nonasymptotic quantum capacity.
\newblock {\em Physical Review A}, 104(2):022401, August 2021.
\newblock \href {https://arxiv.org/abs/1803.10710} {\path{arXiv:1803.10710}}, \href {https://doi.org/10.1103/physreva.104.022401} {\path{doi:10.1103/physreva.104.022401}}.

\bibitem[KKGW21]{KKGW21}
Sumeet Khatri, Eneet Kaur, Saikat Guha, and Mark~M. Wilde.
\newblock Second-order coding rates for key distillation in quantum key distribution, 2021.
\newblock \href {https://arxiv.org/abs/1910.03883} {\path{arXiv:1910.03883}}.

\bibitem[KS24]{KS24}
Gereon Koßmann and René Schwonnek.
\newblock Optimising the relative entropy under semi definite constraints -- a new tool for estimating key rates in {QKD}, 2024.
\newblock \href {https://arxiv.org/abs/2404.17016} {\path{arXiv:2404.17016}}.

\bibitem[KW24]{KW24}
Sumeet Khatri and Mark~M. Wilde.
\newblock Principles of quantum communication theory: A modern approach, 2024.
\newblock \href {https://arxiv.org/abs/2011.04672v2} {\path{arXiv:2011.04672v2}}.

\bibitem[LDS18]{LDS18}
Felix Leditzky, Nilanjana Datta, and Graeme Smith.
\newblock Useful states and entanglement distillation.
\newblock {\em IEEE Transactions on Information Theory}, 64(7):4689--4708, 2018.
\newblock \href {https://arxiv.org/abs/1701.03081} {\path{arXiv:1701.03081}}, \href {https://doi.org/10.1109/TIT.2017.2776907} {\path{doi:10.1109/TIT.2017.2776907}}.

\bibitem[LKDW18]{LKDW18}
Felix Leditzky, Eneet Kaur, Nilanjana Datta, and Mark~M. Wilde.
\newblock Approaches for approximate additivity of the {H}olevo information of quantum channels.
\newblock {\em Physical Review A}, 97:012332, January 2018.
\newblock \href {https://doi.org/10.1103/PhysRevA.97.012332} {\path{doi:10.1103/PhysRevA.97.012332}}.

\bibitem[LM15]{LM15}
Debbie Leung and William Matthews.
\newblock On the power of {PPT}-preserving and non-signalling codes.
\newblock {\em IEEE Transactions on Information Theory}, 61(8):4486--4499, 2015.
\newblock arXiv:1406.7142.
\newblock \href {https://doi.org/10.1109/TIT.2015.2439953} {\path{doi:10.1109/TIT.2015.2439953}}.

\bibitem[Mat13]{Mat13}
Keiji Matsumoto.
\newblock A new quantum version of $f$-divergence, 2013.
\newblock \href {https://arxiv.org/abs/1311.4722} {\path{arXiv:1311.4722}}, \href {https://doi.org/10.48550/ARXIV.1311.4722} {\path{doi:10.48550/ARXIV.1311.4722}}.

\bibitem[MH11]{MH11}
Milán Mosonyi and Fumio Hiai.
\newblock On the quantum {R}ényi relative entropies and related capacity formulas.
\newblock {\em IEEE Transactions on Information Theory}, 57(4):2474--2487, 2011.
\newblock arXiv:0912.1286.
\newblock \href {https://doi.org/10.1109/TIT.2011.2110050} {\path{doi:10.1109/TIT.2011.2110050}}.

\bibitem[MLDS{\etalchar{+}}13]{MDSST13}
Martin Müller-Lennert, Fr{\'e}d{\'e}ric Dupuis, Oleg Szehr, Serge Fehr, and Marco Tomamichel.
\newblock On quantum {R\'{e}}nyi entropies: A new generalization and some properties.
\newblock {\em Journal of Mathematical Physics}, 54(12):122203, December 2013.
\newblock \href {https://arxiv.org/abs/1306.3142} {\path{arXiv:1306.3142}}, \href {https://doi.org/10.1063/1.4838856} {\path{doi:10.1063/1.4838856}}.

\bibitem[MO21]{MO21}
Milán Mosonyi and Tomohiro Ogawa.
\newblock Divergence radii and the strong converse exponent of classical-quantum channel coding with constant compositions.
\newblock {\em IEEE Transactions on Information Theory}, 67(3):1668--1698, 2021.
\newblock arXiv:1811.10599.
\newblock \href {https://doi.org/10.1109/TIT.2020.3041205} {\path{doi:10.1109/TIT.2020.3041205}}.

\bibitem[NIS]{NISO}
NISO.
\newblock Credit – contributor roles taxonomy.
\newblock \url{https://credit.niso.org/}, Accessed 2024-10-14.

\bibitem[NO00]{nagaoka2000StrongConverseSteins}
H.~Nagaoka and T.~Ogawa.
\newblock Strong converse and {{Stein}}'s lemma in quantum hypothesis testing.
\newblock {\em IEEE Transactions on Information Theory}, 46:2428--2433, November 2000.
\newblock URL: \url{http://ieeexplore.ieee.org/document/887855/}, \href {https://doi.org/10.1109/18.887855} {\path{doi:10.1109/18.887855}}.

\bibitem[Pet86]{Petz86}
Dénes Petz.
\newblock Quasi-entropies for finite quantum systems.
\newblock {\em Reports on Mathematical Physics}, 23(1):57--65, 1986.
\newblock URL: \url{https://www.sciencedirect.com/science/article/pii/0034487786900674}, \href {https://arxiv.org/abs/1009.2679} {\path{arXiv:1009.2679}}, \href {https://doi.org/10.1016/0034-4877(86)90067-4} {\path{doi:10.1016/0034-4877(86)90067-4}}.

\bibitem[PV10]{PV10}
Yury Polyanskiy and Sergio Verdú.
\newblock Arimoto channel coding converse and {R}ényi divergence.
\newblock In {\em 2010 48th Annual Allerton Conference on Communication, Control, and Computing (Allerton)}, pages 1327--1333, 2010.
\newblock \href {https://doi.org/10.1109/ALLERTON.2010.5707067} {\path{doi:10.1109/ALLERTON.2010.5707067}}.

\bibitem[QSW18]{QSW18}
Haoyu Qi, Kunal Sharma, and Mark~M. Wilde.
\newblock Entanglement-assisted private communication over quantum broadcast channels.
\newblock {\em Journal of Physics A: Mathematical and Theoretical}, 51(37):374001, August 2018.
\newblock arXiv:1803.03976.
\newblock URL: \url{https://dx.doi.org/10.1088/1751-8121/aad5f3}, \href {https://doi.org/10.1088/1751-8121/aad5f3} {\path{doi:10.1088/1751-8121/aad5f3}}.

\bibitem[RBL18]{RBL18}
Denis Rosset, Francesco Buscemi, and Yeong-Cherng Liang.
\newblock Resource theory of quantum memories and their faithful verification with minimal assumptions.
\newblock {\em Physical Review X}, 8(2):021033, May 2018.
\newblock arXiv:1710.04710.
\newblock URL: \url{https://link.aps.org/doi/10.1103/PhysRevX.8.021033}, \href {https://doi.org/10.1103/PhysRevX.8.021033} {\path{doi:10.1103/PhysRevX.8.021033}}.

\bibitem[RR11]{RR11}
Joseph~M. Renes and Renato Renner.
\newblock Noisy channel coding via privacy amplification and information reconciliation.
\newblock {\em IEEE Transactions on Information Theory}, 57(11):7377--7385, 2011.
\newblock arXiv:1012.4814.
\newblock \href {https://doi.org/10.1109/TIT.2011.2162226} {\path{doi:10.1109/TIT.2011.2162226}}.

\bibitem[RR12]{RR12}
Joseph~M. Renes and Renato Renner.
\newblock One-shot classical data compression with quantum side information and the distillation of common randomness or secret keys.
\newblock {\em IEEE Transactions on Information Theory}, 58(3):1985--1991, 2012.
\newblock arXiv:1008.0452.
\newblock \href {https://doi.org/10.1109/TIT.2011.2177589} {\path{doi:10.1109/TIT.2011.2177589}}.

\bibitem[RSW17]{RSW17}
Jaikumar Radhakrishnan, Pranab Sen, and Naqueeb~Ahmad Warsi.
\newblock One-shot private classical capacity of quantum wiretap channel: Based on one-shot quantum covering lemma, 2017.
\newblock URL: \url{https://arxiv.org/abs/1703.01932}, \href {https://arxiv.org/abs/1703.01932} {\path{arXiv:1703.01932}}.

\bibitem[Sio58]{Sion58}
Maurice Sion.
\newblock On general minimax theorems.
\newblock {\em Pacific Journal of Mathematics}, 8(1):171--176, March 1958.
\newblock \href {https://doi.org/10.2140/pjm.1958.8.171} {\path{doi:10.2140/pjm.1958.8.171}}.

\bibitem[SW24a]{SW24}
Vishal Singh and Mark~M. Wilde.
\newblock No-go theorem for probabilistic one-way secret-key distillation, 2024.
\newblock \href {https://arxiv.org/abs/2404.01392} {\path{arXiv:2404.01392}}.

\bibitem[SW24b]{SW24_channels}
Vishal Singh and Mark~M. Wilde.
\newblock Unextendible entanglement of quantum channels, 2024.
\newblock \href {https://arxiv.org/abs/2407.15944} {\path{arXiv:2407.15944}}.

\bibitem[TGW14]{TGW14}
Masahiro Takeoka, Saikat Guha, and Mark~M. Wilde.
\newblock The squashed entanglement of a quantum channel.
\newblock {\em IEEE Transactions on Information Theory}, 60(8):4987--4998, 2014.
\newblock arXiv:1310.0129.
\newblock \href {https://doi.org/10.1109/TIT.2014.2330313} {\path{doi:10.1109/TIT.2014.2330313}}.

\bibitem[Tom15]{Tomamichel15}
Marco Tomamichel.
\newblock {\em Quantum Information Processing with Finite Resources}.
\newblock Springer Cham, 2015.
\newblock arXiv:1504.00233.
\newblock \href {https://doi.org/10.1007/978-3-319-21891-5} {\path{doi:10.1007/978-3-319-21891-5}}.

\bibitem[Ume62]{Ume62}
Hisaharu Umegaki.
\newblock {Conditional expectation in an operator algebra. IV. Entropy and information}.
\newblock {\em Kodai Mathematical Seminar Reports}, 14(2):59 -- 85, 1962.
\newblock \href {https://doi.org/10.2996/kmj/1138844604} {\path{doi:10.2996/kmj/1138844604}}.

\bibitem[Wat18]{Watrous2018}
John Watrous.
\newblock {\em The Theory of Quantum Information}.
\newblock Cambridge University Press, 2018.
\newblock \href {https://doi.org/10.1017/9781316848142} {\path{doi:10.1017/9781316848142}}.

\bibitem[WBHK20]{WBHK20}
Mark~M. Wilde, Mario Berta, Christoph Hirche, and Eneet Kaur.
\newblock Amortized channel divergence for asymptotic quantum channel discrimination.
\newblock {\em Letters in Mathematical Physics}, 110(8):2277--2336, 2020.
\newblock \href {https://arxiv.org/abs/1808.01498} {\path{arXiv:1808.01498}}, \href {https://doi.org/10.1007/s11005-020-01297-7} {\path{doi:10.1007/s11005-020-01297-7}}.

\bibitem[Wer89]{Wer89}
Reinhard~F. Werner.
\newblock An application of {B}ell's inequalities to a quantum state extension problem.
\newblock {\em Letters in Mathematical Physics}, 17(4):359--363, 1989.
\newblock \href {https://doi.org/10.1007/BF00399761} {\path{doi:10.1007/BF00399761}}.

\bibitem[Wil16]{Wilde16}
Mark~M. Wilde.
\newblock Squashed entanglement and approximate private states.
\newblock {\em Quantum Information Processing}, 15(11):4563--4580, November 2016.
\newblock arXiv:1606.08028.
\newblock \href {https://doi.org/10.1007/s11128-016-1432-7} {\path{doi:10.1007/s11128-016-1432-7}}.

\bibitem[Wil17]{Wilde17}
Mark~M. Wilde.
\newblock Position-based coding and convex splitting for private communication over quantum channels.
\newblock {\em Quantum Information Processing}, 16(10):264, September 2017.
\newblock arXiv:1703.01733.
\newblock \href {https://doi.org/10.1007/s11128-017-1718-4} {\path{doi:10.1007/s11128-017-1718-4}}.

\bibitem[WR12]{WR12}
Ligong Wang and Renato Renner.
\newblock One-shot classical-quantum capacity and hypothesis testing.
\newblock {\em Physical Review Letters}, 108(20):200501, May 2012.
\newblock \href {https://arxiv.org/abs/1007.5456} {\path{arXiv:1007.5456}}, \href {https://doi.org/10.1103/PhysRevLett.108.200501} {\path{doi:10.1103/PhysRevLett.108.200501}}.

\bibitem[WTB17]{WTB17}
Mark~M. Wilde, Marco Tomamichel, and Mario Berta.
\newblock Converse bounds for private communication over quantum channels.
\newblock {\em IEEE Transactions on Information Theory}, 63(3):1792--1817, 2017.
\newblock arXiv:1602.08898.
\newblock \href {https://doi.org/10.1109/TIT.2017.2648825} {\path{doi:10.1109/TIT.2017.2648825}}.

\bibitem[WW19]{WW19}
Xin Wang and Mark~M. Wilde.
\newblock Resource theory of asymmetric distinguishability for quantum channels.
\newblock {\em Physical Review Research}, 1(3):033169, December 2019.
\newblock \href {https://arxiv.org/abs/1907.06306} {\path{arXiv:1907.06306}}, \href {https://doi.org/10.1103/PhysRevResearch.1.033169} {\path{doi:10.1103/PhysRevResearch.1.033169}}.

\bibitem[WWW24]{WWW24}
Kun Wang, Xin Wang, and Mark~M. Wilde.
\newblock Quantifying the unextendibility of entanglement.
\newblock {\em New Journal of Physics}, 26(3):033013, March 2024.
\newblock \href {https://arxiv.org/abs/1911.07433} {\path{arXiv:1911.07433}}, \href {https://doi.org/10.1088/1367-2630/ad264e} {\path{doi:10.1088/1367-2630/ad264e}}.

\bibitem[WWY14]{WWY14}
Mark~M. Wilde, Andreas Winter, and Dong Yang.
\newblock Strong converse for the classical capacity of entanglement-breaking and {H}adamard channels via a sandwiched {R\'{e}}nyi relative entropy.
\newblock {\em Communications in Mathematical Physics}, 331(2):593--622, July 2014.
\newblock \href {https://arxiv.org/abs/1306.1586} {\path{arXiv:1306.1586}}, \href {https://doi.org/10.1007/s00220-014-2122-x} {\path{doi:10.1007/s00220-014-2122-x}}.

\bibitem[Yan06]{Yang06}
Dong Yang.
\newblock A simple proof of monogamy of entanglement.
\newblock {\em Physics Letters A}, 360(2):249--250, 2006.
\newblock URL: \url{https://www.sciencedirect.com/science/article/pii/S0375960106012801}, \href {https://arxiv.org/abs/quant-ph/0604168} {\path{arXiv:quant-ph/0604168}}, \href {https://doi.org/10.1016/j.physleta.2006.08.027} {\path{doi:10.1016/j.physleta.2006.08.027}}.

\end{thebibliography}

\appendix

\numberwithin{equation}{section}

\section{Proof of Proposition~\ref{prop:max_ent_unext_ent_hypo_test}}\label{app:max_ent_unext_ent_hypo_test}

    In this section, we calculate the smooth-min unextendible entanglement of the maximally entangled state $\Phi^d_{AB}$.
    
    We first note that all extensions of the state $\Phi^d_{AB}$ are of the form $\Phi^d_{AB}\otimes\tau_{E}$ since $\Phi^d_{AB}$ is a pure state. Therefore, all states in the set $\mathcal{F}(\Phi^d_{AB})$ are of the form $\pi_A\otimes \tau_E$, where $\pi_A$ is the maximally mixed state and $E\cong B$. The unextendible entanglement of $\Phi^d_{AB}$ induced by the hypothesis testing relative entropy can be calculated as follows:
    \begin{align}
        E^{u,\varepsilon}_{\min}\!\left(\Phi^d_{AB}\right) &= \inf_{\tau_B \in \mathcal{S}(B)}\frac{1}{2}D^{\varepsilon}_{\min}\!\left(\Phi^d_{AB}\Vert\pi_A\otimes \tau_B\right)\\
        &= \inf_{\tau_B}-\frac{1}{2}\log_2\inf_{0\le \Lambda \le I}\left\{\operatorname{Tr}\!\left[\Lambda_{AB}(\pi_A\otimes \tau_B)\right]: \operatorname{Tr}\!\left[\Lambda_{AB}\Phi^d_{AB}\right]\ge 1-\varepsilon\right\}.
    \end{align}
    Choosing $\Lambda_{AB} = (1-\varepsilon)\Phi^d_{AB}$, we find that
    \begin{align}
        E^{u,\varepsilon}_{\min}\!\left(\Phi^d_{AB}\right) &\ge \inf_{\tau_B} -\frac{1}{2}\log_2\left((1-\varepsilon)\operatorname{Tr}\!\left[\Phi^d_{AB}(\pi_A\otimes \tau_B)\right]\right)\\
        &= \inf_{\tau_B} -\frac{1}{2}\log_2\left(\frac{1-\varepsilon}{d}\operatorname{Tr}\!\left[\Phi^d_{AB}(I_A\otimes \tau_B)\right]\right)\\
        &= \inf_{\tau_B} -\frac{1}{2}\log_2\left(\frac{1-\varepsilon}{d}\operatorname{Tr}\!\left[\pi_{B}\tau_B\right]\right)\\
        &= \inf_{\tau_B} -\frac{1}{2}\log_2\left(\frac{1-\varepsilon}{d^2}\operatorname{Tr}\!\left[\tau_B\right]\right)\\
        &= -\frac{1}{2}\log_2\left(\frac{1-\varepsilon}{d^2}\right).\label{eq:max_ent_hypo_test_unext_ent_lb}
    \end{align}

    The hypothesis testing relative entropy can also be computed using the following SDP:
    \begin{equation}
        D^{\varepsilon}_{\min}\!\left(\rho_{AB}\right) = -\log_2 \sup_{\mu \ge 0, Z\ge 0}\left\{\mu(1-\varepsilon) - \operatorname{Tr}\!\left[Z\right]: \mu\rho \le \sigma + Z\right\}.
    \end{equation}
    The unextendible entanglement of the maximally entangled state induced by the hypothesis testing relative entropy can then be computed as follows:
    \begin{equation}
        E^{u,\varepsilon}_{\min}\!\left(\Phi^d_{AB}\right) = \inf_{\tau_B \in \mathcal{S}(B)} -\frac{1}{2}\log_2 \sup_{\mu \ge 0, Z\ge 0}\left\{\mu(1-\varepsilon) - \operatorname{Tr}\!\left[Z\right]: \mu\Phi^d_{AB} \le \pi_A\otimes \tau_B + Z_{AB}\right\}
    \end{equation}
    Choosing $\tau_B$ to be the maximally mixed state, we arrive at the following inequality:
    \begin{align}
        E^{u,\varepsilon}_{\min}\!\left(\Phi^d_{AB}\right) &\le -\frac{1}{2}\log_2 \sup_{\mu \ge 0, Z\ge 0}\left\{\mu(1-\varepsilon) - \operatorname{Tr}\!\left[Z\right]: \mu\Phi^d_{AB} \le \pi_{AB} + Z_{AB}\right\}\\
        &= \inf_{\mu \ge 0, Z\ge 0}-\frac{1}{2}\log_2 \left\{\mu(1-\varepsilon) - \operatorname{Tr}\!\left[Z\right]: \mu\Phi^d_{AB} \le \pi_{AB} + Z_{AB}\right\}.
    \end{align}
    Note that the pair $\left(\mu = 1/d^2,Z = 0\right)$ lies in the feasible set of the aforementioned SDP. Therefore, setting $\mu = 1/d^2$ and $Z = 0$ leads to the following inequality:
    \begin{equation}\label{eq:max_ent_hypo_test_unext_ent_ub}
        E^{u,\varepsilon}_{\min}\!\left(\Phi^d_{AB}\right) \le -\frac{1}{2}\log_2\!\left(\frac{1-\varepsilon}{d^2}\right) = \log_2 d - \frac{1}{2}\log_2(1-\varepsilon).
    \end{equation}
    Combining~\eqref{eq:max_ent_hypo_test_unext_ent_lb} and~\eqref{eq:max_ent_hypo_test_unext_ent_ub} concludes the proof.

\section{Proof of Proposition~\ref{prop:J_eps_range}}\label{app:J_eps_range}

In this section we find the range of values that the smooth-min unextendible entanglement of a state can take.

The smooth-min relative entropy between two states is never smaller than $-\log_2(1-\varepsilon)$, which can be seen from the data-processing inequality of the smooth-min relative entropy as follows:
    \begin{align}
        D^{\varepsilon}_{\min}\!\left(\rho\Vert\sigma\right) &\ge D^{\varepsilon}_{\min}\!\left(\mathcal{R}^{\pi}\!\left(\rho\right)\Vert\mathcal{R}^{\pi}\!\left(\sigma\right)\right)\\
        &= D^{\varepsilon}_{\min}\!\left(\pi\Vert\pi\right)\\
        &= -\log_2 \inf_{0\le \Lambda \le I}\left\{\operatorname{Tr}\!\left[\Lambda\pi\right]: \operatorname{Tr}\!\left[\Lambda\pi\right] \ge 1-\varepsilon\right\}\\
        &= -\log_2(1-\varepsilon),
    \end{align}
    where $\mathcal{R}^{\pi}$ is a channel that traces out the state it acts on and replaces it with the maximally mixed state $\pi$. This leads to the following bound on the smooth-min unextendible entanglement of a bipartite state $\rho_{AB}$:
    \begin{equation}
        E^{u,\varepsilon}_{\min}\!\left(\rho_{AB}\right) = \inf_{\sigma \in \mathcal{F}(\rho)}\frac{1}{2}D^{\varepsilon}_{\min}\!\left(\rho_{AB}\Vert\sigma_{AB}\right) \ge -\frac{1}{2}\log_2(1-\varepsilon).
    \end{equation}

    To find an upper bound on $E^{u,\varepsilon}_{\min}\!\left(\rho_{AB}\right)$, we invoke~\eqref{eq:unext_ent_le_max_ent}. Since the hypothesis testing relative entropy is an example of generalized divergence,~\eqref{eq:unext_ent_le_max_ent} implies the following inequality:
    \begin{equation}\label{eq:smooth_min_rho_ub}
        E^{u,\varepsilon}_{\min}\!\left(\rho_{AB}\right) \le E^{u,\varepsilon}_{\min}\!\left(\Phi^d_{A_0B_0}\right) = \log_2d - \frac{1}{2}\log_2(1-\varepsilon),
    \end{equation}
    where $d \coloneqq \min\{\operatorname{dim}(A),\operatorname{dim}(B)\}$. The equality in~\eqref{eq:smooth_min_rho_ub} follows from Proposition~\ref{prop:max_ent_unext_ent_hypo_test}.

\section{Proof of Equation~\texorpdfstring{\eqref{eq:q_final_expression}}{79}}\label{app:priv_marg_phi_fid_bnd}

In this appendix, we show that~\eqref{eq:q_eq_main} implies~\eqref{eq:q_final_expression}.

Let us analyze the expression on the right-hand side of~\eqref{eq:q_eq_main} in the interval $q\in [0,1]$ and for $k\in \mathbb{N}$. In what follows, we find the inflection points of the expression by setting the derivative equal to zero:
\begin{align}
    0 &=\frac{\partial}{\partial q}\left(\sqrt{\frac{q}{k^2}} + \sqrt{\left(1-q\right)\left(1-\frac{1}{k^2}\right)}\right)^2\\
    &= 2\left(\sqrt{\frac{q}{k^2}} + \sqrt{\left(1-q\right)\left(1-\frac{1}{k^2}\right)}\right)\left(\frac{1}{2k\sqrt{q}} - \frac{1}{2\sqrt{1-q}}\sqrt{1-\frac{1}{k^2}}\right). \label{eq:deriv-calc-q-k}
\end{align}
Solving the above equation for $q$, we find that
\begin{align}
    \frac{1}{2k\sqrt{q}} &= \frac{1}{2\sqrt{1-q}}\sqrt{1-\frac{1}{k^2}}\\
    \implies \sqrt{1-q} &= \sqrt{q(k^2-1)}\\
    \implies 1-q &= q(k^2-1)\\
    \implies q &= \frac{1}{k^2}.
\end{align}
It is easy to verify that the function of $q$ given on the right-hand side of~\eqref{eq:q_eq_main} achieves its maximum value at this inflection point. Therefore, the function is monotonically increasing for $q\in [0,\frac{1}{k^2}]$, and it is monotonically decreasing for $q \in [\frac{1}{k^2}, 1]$. Equivalently, the derivative in~\eqref{eq:deriv-calc-q-k} is non-negative for $q\in [0,\frac{1}{k^2}]$, and it is non-positive for $q \in [\frac{1}{k^2}, 1]$.

Now let us find the values of $q$ that satisfy~\eqref{eq:q_eq_main}. We can rewrite the inequality in~\eqref{eq:q_eq_main} as follows:
\begin{equation}
    1-\varepsilon \le \frac{q}{k^2} + (1-q)\left(1-\frac{1}{k^2}\right) + 2\sqrt{q(1-q)}\sqrt{\frac{1}{k^2}\left(1-\frac{1}{k^2}\right)}.
\end{equation}
Rearranging the terms, we arrive at the following inequality:
\begin{equation}\label{eq:solving_q_step_1}
    q\left(1-\frac{2}{k^2}\right) - \varepsilon + \frac{1}{k^2} \le 2\sqrt{q(1-q)}\sqrt{\frac{1}{k^2}\left(1-\frac{1}{k^2}\right)}.
\end{equation}
The right-hand side of the above equation is always non-negative for all $q\in [0,1]$ and $k\in \mathbb{N}$. If the left-hand side of the above inequality is negative, then the above inequality is satisfied. As such, the above inequality is satisfied if the following condition holds:
\begin{equation}
    q\left(1-\frac{2}{k^2}\right) - \varepsilon + \frac{1}{k^2} \le 0 \implies q \le \frac{\varepsilon - \frac{1}{k^2}}{1-\frac{2}{k^2}} \quad \forall \, k\ge 2.
\end{equation}
If $\varepsilon-\frac{1}{k^2} \ge 0$ then the inequality in~\eqref{eq:solving_q_step_1} is satisfied for all $q\in \left[0,\frac{\varepsilon-k^{-2}}{1-2k^{-2}}\right]$ and $k\ge 2$.

Now let us consider the case where the left-hand side of~\eqref{eq:solving_q_step_1} is non-negative, which is true for $k\geq 2$ if $\varepsilon \le \frac{1}{k^2}$ or $q \ge \left(\varepsilon-\frac{1}{k^2}\right)/\left(1-\frac{2}{k^2}\right)$. We can square both sides to get the following inequality:
\begin{align}
    \left(q\left(1-\frac{2}{k^2}\right) - \varepsilon + \frac{1}{k^2}\right)^2 &\le \left(2\sqrt{q(1-q)}\sqrt{\frac{1}{k^2}\left(1-\frac{1}{k^2}\right)}\right)^2\\
    &= \frac{4}{k^2}\left(1-\frac{1}{k^2}\right)q(1-q)\label{eq:solving_q_step_2}.
\end{align}
Setting
\begin{align}
    d &\coloneqq 1-\frac{2}{k^2},\\
    e &\coloneqq \varepsilon - \frac{1}{k^2},\\
    f &\coloneqq \frac{4}{k^2}\left(1-\frac{1}{k^2}\right),
\end{align}
we can rewrite the above inequality as follows:
\begin{align}
    (d\cdot q - e)^2 &\le f\cdot q(1-q)\\
    \implies d^2 q^2 + e^2 -2deq &\le fq - fq^2\\
    \implies (d^2+f)q^2 - (2de+f)q + e^2 &\le 0.
\end{align}
The above inequality is in the standard form of a quadratic inequality. Let us first find each of the coefficients. The coefficient of $q^2$ evaluates to the following:
\begin{align}
    d^2+f &= \left(1-\frac{2}{k^2}\right)^2 + \frac{4}{k^2}\left(1-\frac{1}{k^2}\right)\\
    &= 1 + \frac{4}{k^4} - \frac{4}{k^2} + \frac{4}{k^2} - \frac{4}{k^4}\\
    &= 1.
\end{align}
The coefficient of $q$ evaluates to the following:
\begin{align}
    -2de -f &= -2\left(1-\frac{2}{k^2}\right)\left(\varepsilon-\frac{1}{k^2}\right) - \frac{4}{k^2}\left(1-\frac{1}{k^2}\right)\\
    &= -2\left(\varepsilon - \frac{1}{k^2} - \frac{2\varepsilon}{k^2} + \frac{2}{k^4}\right) - \frac{4}{k^2} + \frac{4}{k^4}\\
    &= -2\left(\varepsilon + \frac{1-2\varepsilon}{k^2}\right).
\end{align}
Finally, the term independent of $q$ is equal to the following:
\begin{align}
    e^2 &= \left(\varepsilon - \frac{1}{k^2}\right)^2\\
    &= \varepsilon^2 + \frac{1}{k^4} - \frac{2\varepsilon}{k^2}.
\end{align}
The quadratic inequality in~\eqref{eq:solving_q_step_2} can now be written as follows:
\begin{equation}\label{eq:solving_q_final_expr}
    q^2 - 2q\left(\varepsilon + \frac{1-2\varepsilon}{k^2}\right) + \varepsilon^2 +\frac{1}{k^4} - \frac{2\varepsilon}{k^2} \le 0.
\end{equation}
The discriminant of the above quadratic expression can be evaluated as follows:
\begin{align}
    &(-2de-f)^2 - 4(d^2+f)e^2\notag \\
    &= 4\left(\varepsilon + \frac{1-2\varepsilon}{k^2}\right)^2 - 4\left(\varepsilon^2 + \frac{1}{k^4} - \frac{2\varepsilon}{k^2}\right)\\
    &= 4\left(\varepsilon^2 + \left(\frac{1-2\varepsilon}{k^2}\right)^2 + 2\varepsilon\left(\frac{1-2\varepsilon}{k^2}\right) - \varepsilon^2 - \frac{1}{k^4} + \frac{2\varepsilon}{k^2}\right)\\
    &= 4\left(\frac{1+4\varepsilon^2 - 4\varepsilon}{k^4} - \frac{1}{k^4}+\frac{2\varepsilon}{k^2}\left(2-2\varepsilon\right)\right)\\
    &= 4\left(\frac{4\varepsilon(\varepsilon-1)}{k^4} + \frac{4\varepsilon(1-\varepsilon)}{k^2}\right)\\
    &= 16\frac{(k^2-1)\varepsilon(1-\varepsilon)}{k^4}.
\end{align}
We can now factor the quadratic expression in~\eqref{eq:solving_q_final_expr} as follows:
\begin{equation}\label{eq:solving_q_factorized}
    \left(q-\alpha_{q+}\right)\!\left(q-\alpha_{q-}\right) \le 0,
\end{equation}
where
\begin{align}
    \alpha_{q\pm} &= \frac{1}{2}\left(2\left(\varepsilon + \frac{1- 2\varepsilon}{k^2}\right) \pm \frac{\sqrt{16(k^2-1)\varepsilon(1-\varepsilon)} }{k^2}\right)\\
    &= \varepsilon + \frac{1-2\varepsilon}{k^2} \pm \frac{2\sqrt{(k^2-1)\varepsilon(1-\varepsilon)}}{k^2}.\label{eq:alpha_q_pm_defn}
\end{align}

Now we are in a position to identify the range of values that $q$ can take for all $\varepsilon \in [0,1]$ and $k\ge 2$.
\begin{itemize}
    \item If $\varepsilon \in \left[0,\frac{1}{k^2}\right]$, then
    \begin{equation}
        \alpha_{q-} \le q \le \alpha_{q+},
    \end{equation}
    where $\alpha_{q-}$ and $\alpha_{q+}$ are defined in~\eqref{eq:alpha_q_pm_defn}.
    \item If $\varepsilon \in \left[\frac{1}{k^2},1\right]$, then
    \begin{equation}\label{eq:q_range_eps_ge_1/k^2}
        q\in \left[0,\min\left\{\frac{\varepsilon - \frac{1}{k^2}}{1-\frac{2}{k^2}},1\right\}\right]\cup\left[\alpha_{q-},\alpha_{q+}\right].
    \end{equation}
\end{itemize}

We can identify the values of $\varepsilon$ such that the two intervals in~\eqref{eq:q_range_eps_ge_1/k^2} overlap. Let us first find the values of $\varepsilon$ that satisfy the following inequality:
\begin{equation}
    \frac{\varepsilon-\frac{1}{k^2}}{1-\frac{2}{k^2}} \le \alpha_{q+} = \varepsilon + \frac{1-2\varepsilon}{k^2} + \frac{2\sqrt{(k^2-1)\varepsilon(1-\varepsilon)}}{k^2}.\label{eq:eps_range_inequality}
\end{equation}
We can rearrange the terms of the above inequality to get the following inequality:
\begin{align}
    \frac{k^2\varepsilon-1}{k^2-2} - \frac{(k^2-2)\varepsilon+ 1}{k^2} &\le \frac{2\sqrt{(k^2-1)\varepsilon(1-\varepsilon)}}{k^2}\\
    \implies \frac{k^4\varepsilon - k^2 - (k^2-2)^2\varepsilon -k^2 + 2}{k^2(k^2-2)} &\le \frac{2\sqrt{(k^2-1)\varepsilon(1-\varepsilon)}}{k^2}\\
    \implies \frac{4(k^2-1)\varepsilon + 2(1-k^2)}{k^2-2} &\le 2\sqrt{(k^2-1)\varepsilon(1-\varepsilon)}\\
    \implies \frac{(k^2-1)(2\varepsilon - 1)}{k^2-2} &\le \sqrt{(k^2-1)\varepsilon(1-\varepsilon)}.\label{eq:eps_range_inequality_common_step}
\end{align}
The above inequality is satisfied for all $\varepsilon\in [0,1/2]$ and $k\ge 2$. Now assuming that $\varepsilon \ge 1/2$ and $k\ge 2$, we can square both sides of the above inequality to get the following inequality:
\begin{equation}\label{eq:eps_range_ineq_common_step_red}
    \frac{(k^2-1)(2\varepsilon-1)^2}{(k^2-2)^2} \le \varepsilon(1-\varepsilon).
\end{equation}
Setting $a \coloneqq (k^2-1)/(k^2-2)^2$, we can rewrite the above inequality as follows:
\begin{align}
    a(2\varepsilon-1)^2 &\le \varepsilon(1-\varepsilon)\\
    \implies (4a+1)\varepsilon^2 - (4a+1)\varepsilon + a &\le 0.
\end{align}
Note that $a$ is a positive number, which implies that $4a+1$ is also a positive number. Therefore, the above quadratic inequality can be factored as follows:
\begin{equation}\label{eq:eps_factored_ineq}
    \left(\varepsilon - \frac{1}{2}\left(1+\frac{1}{\sqrt{4a+1}}\right)\right)\left(\varepsilon - \frac{1}{2}\left(1-\frac{1}{\sqrt{4a+1}}\right)\right) \le 0.
\end{equation}
Substituting the value of $a$, the quantity $\sqrt{4a+1}$ evaluates to the following:
\begin{equation}
    \sqrt{4a+1} = \frac{k^2}{k^2-2}.
\end{equation}
Therefore, the inequality in~\eqref{eq:eps_factored_ineq} can be written as follows:
\begin{align}
    \left(\varepsilon - \frac{1}{2}\left(1+\frac{k^2-2}{k^2}\right)\right)\left(\varepsilon - \frac{1}{2}\left(1-\frac{k^2-2}{k^2}\right)\right) &\le 0\\
    \implies \left(\varepsilon - \left(1-\frac{1}{k^2}\right)\right)\left(\varepsilon - \frac{1}{k^2}\right) &\le 0.
\end{align}
The above inequality is satisfied only when $\varepsilon \in \left[1/k^2,1-1/k^2\right]$. To get the above inequality we assumed that $\alpha \ge 1/2$. Since $1-\frac{1}{k^2} \ge \frac{1}{2}$ for all $k\ge 2$ and every $\varepsilon \in [0,1/2]$ satisfies the inequality in~\eqref{eq:eps_range_inequality}, we conclude that the inequality in~\eqref{eq:eps_range_inequality} is satisfied if and only if $\varepsilon \in \left[0,1-\frac{1}{k^2}\right]$. 

Now let us find the values of $\varepsilon$ for which the following inequality is satisfied:
\begin{equation}
    \frac{\varepsilon-\frac{1}{k^2}}{1-\frac{2}{k^2}} \le \alpha_{q-} = \varepsilon + \frac{1-2\varepsilon}{k^2} - \frac{2\sqrt{(k^2-1)\varepsilon(1-\varepsilon)}}{k^2}.\label{eq:eps_range_inequality_2}
\end{equation}
Following the same steps as above, we arrive at the following inequality:
\begin{equation}\label{eq:eps_range_inequality_common_step_2}
    \frac{(k^2-1)(2\varepsilon - 1)}{k^2-2} \le -\sqrt{(k^2-1)\varepsilon(1-\varepsilon)},
\end{equation}
which is similar to~\eqref{eq:eps_range_inequality_common_step}. This inequality is not satisfied by any value of $\varepsilon \ge 1/2$. Under the assumption that $\varepsilon \le 1/2$, we can square both sides of the above inequality to get the following inequality:
\begin{equation}\label{eq:eps_range_ineq_common_step_red_2}
    \frac{(k^2-1)(2\varepsilon-1)^2}{(k^2-2)^2} \ge \varepsilon(1-\varepsilon).
\end{equation}
From the solution of~\eqref{eq:eps_range_ineq_common_step_red}, we know that the opposite of this quadratic inequality is satisfied when $\varepsilon \in \left[\frac{1}{k^2},1-\frac{1}{k^2}\right]$. Therefore, the inequality in~\eqref{eq:eps_range_ineq_common_step_red_2} is satisfied when $\varepsilon \in \left[0,\frac{1}{k^2}\right]\cup \left[1-\frac{1}{k^2},1\right]$. Recall that we assumed $\varepsilon\in [0,1/2]$ to arrive at~\eqref{eq:eps_range_ineq_common_step_red_2} from~\eqref{eq:eps_range_inequality_common_step_2}. Therefore, for every $k\ge 2$, we conclude that~\eqref{eq:eps_range_inequality_2} is satisfied for all $\varepsilon \in \left[0,\frac{1}{k^2}\right]$. 

Now we have a clearer picture of the range of values $q$ that satisfy~\eqref{eq:q_eq_main} for some fixed value of $\varepsilon \in [0,1]$ and integer $k\ge 2$, which is given as follows:
\begin{equation}
    q \in \begin{cases}
        \left[\alpha_{q-},\alpha_{q+}\right] &\text{if } \varepsilon \in \left[0,\frac{1}{k^2}\right]\\
        \left[0,\alpha_{q+}\right] &\text{if }\varepsilon \in \left[\frac{1}{k^2}, 1-\frac{1}{k^2}\right]\\
        \left[0,\min\left\{\frac{\varepsilon - \frac{1}{k^2}}{1-\frac{2}{k^2}},1\right\}\right] &\text{if }\varepsilon \in \left[1-\frac{1}{k^2},1\right]
    \end{cases},
\end{equation}
where $\alpha_{q+}$ and $\alpha_{q-}$ are defined in~\eqref{eq:alpha_q_pm_defn}. Observe that $\frac{\varepsilon - \frac{1}{k^2}}{1-\frac{2}{k^2}} \ge 1$ for all $\varepsilon \in \left[1-\frac{1}{k^2},1\right]$. Therefore, we can rewrite the above condition as follows:
\begin{equation}
    q \in \begin{cases}
        \left[\alpha_{q-},\alpha_{q+}\right] &\text{if } \varepsilon \in \left[0,\frac{1}{k^2}\right]\\
        \left[0,\alpha_{q+}\right] &\text{if }\varepsilon \in \left[\frac{1}{k^2}, 1-\frac{1}{k^2}\right]\\
        \left[0,1\right] &\text{if }\varepsilon \in \left[1-\frac{1}{k^2},1\right]
    \end{cases}.
\end{equation}
As such, if $\varepsilon \le 1-\frac{1}{k^2}$, then $q$ is bounded from above by the following quantity:
\begin{equation}
    q \le \varepsilon + \frac{1-2\varepsilon}{k^2} +\frac{2\sqrt{(k^2-1)\varepsilon(1-\varepsilon)} }{k^2},
\end{equation}
and otherwise, when $\varepsilon > 1-\frac{1}{k^2}$, we only have the trivial bound $q \leq 1$.

\section{Proof of Theorem~\ref{theo:distill_key_st_ub_hypo_test}}\label{app:distill_key_st_ub_hypo_test}

First let us note that the condition $\varepsilon < J^{\varepsilon}_{\min}\!\left(\rho_{AB}\right)$ is only satisfied if $\varepsilon < \frac{1}{2}$. This is because $J^{\varepsilon}_{\min}\!\left(\rho_{AB}\right) \le 1-\varepsilon$, as stated in Proposition~\ref{prop:J_eps_range}. So we restrict to $\varepsilon < \frac{1}{2}$ for the remainder of the proof.

Let $\mathcal{L}^{\to}_{AB\to A'B'A''B''} $ be a one-way LOCC channel that acts on $\rho_{AB}$ to give a state $\sigma_{A'B'A''B''}$ such that $F(\sigma_{A'B'A''B''},\gamma^k_{A'B'A''B''})\ge 1-\varepsilon$ for some bipartite private state $\gamma^k_{A'B'A''B''}$. That is,
	\begin{equation}
		\mathcal{L}^{\to}_{AB\to A'B'A''B''}\!\left(\rho_{AB}\right) = \sigma_{A'B'A''B''}.
	\end{equation}
	The smooth-min unextendible entanglement of a bipartite state does not increase under the action of a one-way LOCC channel~\cite[Theorem~2]{WWW24}. Therefore,
	\begin{align}
		E^{u,\varepsilon}_{\min}\!\left(\rho_{AB}\right) &\ge E^{u,\varepsilon}_{\min}\!\left(\mathcal{L}^{\to}\!\left(\rho_{AB}\right)\right)\\
		 &= E^{u,\varepsilon}_{\min}\!\left(\sigma_{A'B'A''B''}\right).
\end{align}
Consequently,
\begin{equation}
    J^{\varepsilon}_{\min}\!\left(\rho_{AB}\right) \le J^{\varepsilon}_{\min}\!\left(\sigma_{A'B'A''B''}\right).
\end{equation}
Since $\varepsilon < \frac{1}{2}$, we can use Remark~\ref{rem:nogo_key_distill} to state that one cannot distill any secret bits from a state $\rho_{AB}$ with an error tolerance of $\varepsilon$ if $J^{\varepsilon}_{\min}\!\left(\rho_{AB}\right) > \varsigma\!\left(\varepsilon,2\right)$, where $\varsigma$ is defined in~\eqref{eq:varsig-func}. 

Let us now consider the case when one-shot, one-way secret-key distillation is possible. Proposition~\ref{prop:unext_ent_hyp_test_lb} implies that the following inequality holds for all $\varepsilon \in \left[0,\frac{1}{2}\right)$:
\begin{equation}\label{eq:J_eps_ineq_hypo_test}
    J^{\varepsilon}_{\min}\!\left(\rho_{AB}\right) \le \varsigma\!\left(\varepsilon,k\right) = \varepsilon + \frac{1-2\varepsilon}{k^2} + \frac{2\sqrt{(k^2-1)\varepsilon(1-\varepsilon)}}{k^2}
\end{equation}
if $F\!\left(\mathcal{L}^{\to}_{AB\to A'B'A''B''}\!\left(\rho_{AB}\right),\gamma^k_{A'B'A''B''}\right)\ge 1-\varepsilon$ for some one-way LOCC channel $\mathcal{L}^{\to}_{AB\to A'B'A''B''}$ and some private state $\gamma^k_{A'B'A''B''}$ with $k\ge 2$. We will use $J^{\varepsilon}_{\min}$ as a shorthand for $J^{\varepsilon}_{\min}\!\left(\rho_{AB}\right)$ in the remainder of the proof for convenience. Rearranging the terms in~\eqref{eq:J_eps_ineq_hypo_test}, we arrive at the following inequality:
\begin{align}
    k^2(J^{\varepsilon}_{\min}-\varepsilon)&\le 1-2\varepsilon + 2\sqrt{(k^2-1)\varepsilon(1-\varepsilon)}\\
    \implies (k^2-1)(J^{\varepsilon}_{\min}-\varepsilon) + J^{\varepsilon}_{\min}-\varepsilon &\le 1-2\varepsilon + 2\sqrt{(k^2-1)\varepsilon(1-\varepsilon)}\\
    \implies (k^2-1)(J^{\varepsilon}_{\min}-\varepsilon) + J^{\varepsilon}_{\min}+\varepsilon -1 &\le 2\sqrt{(k^2-1)\varepsilon(1-\varepsilon)}.\label{eq:k_inequality_main_hypo_test}
\end{align}
The above inequality is always satisfied if the left-hand side is non-positive, that is, if
\begin{equation}\label{eq:k_ineq_triv_step}
    (k^2-1)(J^{\varepsilon}_{\min}-\varepsilon) + J^{\varepsilon}_{\min}+\varepsilon -1 \le 0 .
\end{equation}
Note that $J^{\varepsilon}_{\min}+\varepsilon-1$ is always non-positive due to~\eqref{eq:J_eps_range}. Therefore, if $J^{\varepsilon}_{\min} \le \varepsilon$, then the above inequality holds for all $k\in \mathbb{N}$. If $J^{\varepsilon}_{\min} > \varepsilon$, then~\eqref{eq:k_ineq_triv_step} is satisfied when the following condition holds:
\begin{equation}\label{eq:k_ineq_triv_soln}
    k^2-1 \le \frac{1-\varepsilon-J^{\varepsilon}_{\min}}{J^{\varepsilon}_{\min}-\varepsilon}.
\end{equation}
Consequently, if the inequality in~\eqref{eq:k_ineq_triv_soln} holds, then the inequality in~\eqref{eq:k_inequality_main_hypo_test} also holds for all $J^{\varepsilon}_{\min} \in [0,1-\varepsilon]$.

Now let us assume that $(k^2-1)(J^{\varepsilon}_{\min}-\varepsilon) + J^{\varepsilon}_{\min}+\varepsilon -1 > 0$. Then we can square both sides of~\eqref{eq:k_inequality_main_hypo_test} to get the following inequality:
\begin{equation}\label{eq:k_ineq_hypo_test_squared}
    \left((k^2-1)(J^{\varepsilon}_{\min}-\varepsilon) + J^{\varepsilon}_{\min}+\varepsilon -1\right)^2 \le 4(k^2-1)\varepsilon(1-\varepsilon).
\end{equation}
We write the above inequality as follows:
\begin{equation}\label{eq:k_ineq_hypo_test_comp}
    \left(d\cdot x + e\right)^2 \le f\cdot x,
\end{equation}
where we have made the following assignments:
\begin{align}
    x &\coloneqq k^2 - 1,\\
    d &\coloneqq J^{\varepsilon}_{\min} - \varepsilon,\\
    e &\coloneqq J^{\varepsilon}_{\min}+\varepsilon-1,\\
    f &\coloneqq 4\varepsilon(1-\varepsilon).
\end{align}
We can write the inequality in~\eqref{eq:k_ineq_hypo_test_comp} as the following quadratic inequality in~$x$:
\begin{align}
    d^2x^2 + e^2 + 2de\cdot x &\le f\cdot x\\
    \implies d^2x^2 + (2de-f)x + e^2 &\le 0.
\end{align}
The coefficient of $x^2$ evaluates to the following:
\begin{equation}
    d^2 = (J^{\varepsilon}_{\min}-\varepsilon)^2.
\end{equation}
The coefficient of $x$ evaluates to the following:
\begin{align}
    2de-f &= 2(J^{\varepsilon}_{\min}-\varepsilon)(J^{\varepsilon}_{\min}+\varepsilon-1) - 4\varepsilon(1-\varepsilon)\\
    &= -2J^{\varepsilon}_{\min}(1-J^{\varepsilon}_{\min}) - 2\varepsilon(1-\varepsilon).
\end{align}
In what follows, we compute the discriminant of the above quadratic expression:
\begin{align}
    (2de-f)^2 - 4d^2e^2 &= 4d^2e^2+f^2-4def-4d^2e^2\\
    &= f^2-4def\\
    &= 16\varepsilon^2(1-\varepsilon)^2 - 16(J^{\varepsilon}_{\min}-\varepsilon)(J^{\varepsilon}_{\min}+\varepsilon-1)\varepsilon(1-\varepsilon)\\
    &= 16\varepsilon(1-\varepsilon)\left[\varepsilon(1-\varepsilon)-\left(J^{\varepsilon}_{\min}\right)^2+J^{\varepsilon}_{\min}-\varepsilon(1-\varepsilon)\right]\\
    &= 16\varepsilon(1-\varepsilon)J^{\varepsilon}_{\min}(1-J^{\varepsilon}_{\min}).
\end{align}
Note that the coefficient of $x^2$, that is $(J^{\varepsilon}_{\min}-\varepsilon)^2$, is always positive since we have assumed that $J^{\varepsilon}_{\min}>\varepsilon$, which allows us to factor the quadratic expression in~\eqref{eq:k_ineq_hypo_test_squared} as follows:
\begin{equation}\label{eq:x_ineq_factorized}
    \left(x-\beta_{x-}\right)\left(x-\beta_{x+}\right) \le 0,
\end{equation}
where 
\begin{align}
    \beta_{x\pm} &\coloneqq \frac{2J^{\varepsilon}_{\min}(1-J^{\varepsilon}_{\min})+2\varepsilon(1-\varepsilon) \pm 4\sqrt{\varepsilon(1-\varepsilon)J^{\varepsilon}_{\min}(1-J^{\varepsilon}_{\min})}}{2(J^{\varepsilon}_{\min}-\varepsilon)^2}\\
    &= \left(\frac{\sqrt{J^{\varepsilon}_{\min}(1-J^{\varepsilon}_{\min})} \pm \sqrt{\varepsilon(1-\varepsilon)}}{J^{\varepsilon}_{\min}-\varepsilon}\right)^2.\label{eq:beta_x_pm_defn}
\end{align}
The inequality in~\eqref{eq:x_ineq_factorized} is satisfied if and only if $\beta_{x-}\le x \le \beta_{x+}$. Combining~\eqref{eq:k_ineq_triv_soln} and~\eqref{eq:x_ineq_factorized}, we conclude that~\eqref{eq:k_inequality_main_hypo_test} is satisfied if and only if
\begin{equation}\label{eq:k_soln_naive}
    k^2-1 \in \left[0,\frac{1-\varepsilon-J^{\varepsilon}_{\min}}{J^{\varepsilon}_{\min}-\varepsilon}\right] \cup \left[\beta_{x-},\beta_{x+}\right].
\end{equation}

Recall that we had assumed $k^2-1 > \frac{1-\varepsilon-J^{\varepsilon}_{\min}}{J^{\varepsilon}_{\min}-\varepsilon}$ to arrive at~\eqref{eq:k_ineq_hypo_test_squared}, and the inequality in~\eqref{eq:k_inequality_main_hypo_test} holds for all $0\le k^2-1\le  \frac{1-\varepsilon-J^{\varepsilon}_{\min}}{J^{\varepsilon}_{\min}-\varepsilon}$. In what follows, we shall show that $\beta_{x-}\le \frac{1-\varepsilon-J^{\varepsilon}_{\min}}{J^{\varepsilon}_{\min}-\varepsilon}\le \beta_{x+}$ for all $\varepsilon < J^{\varepsilon}_{\min} \le 1-\varepsilon$. 

First, let us consider the following inequality:
\begin{align}
    \sqrt{J^{\varepsilon}_{\min}(1-J^{\varepsilon}_{\min})} &\le \sqrt{\varepsilon(1-\varepsilon)}\label{eq:J_eps_sqrt_rev_ineq}\\
    \Longleftrightarrow J^{\varepsilon}_{\min}(1-J^{\varepsilon}_{\min}) &\le \varepsilon(1-\varepsilon)\\
    \Longleftrightarrow (J^{\varepsilon}_{\min})^2-J^{\varepsilon}_{\min}+\varepsilon(1-\varepsilon) &\ge 0.
\end{align}
The above inequality can be factored as follows:
\begin{equation}
    \left(J^{\varepsilon}_{\min}-\varepsilon\right)\left(J^{\varepsilon}_{\min}-(1-\varepsilon)\right)\ge 0.
\end{equation}
Therefore, the inequality in~\eqref{eq:J_eps_sqrt_rev_ineq} is satisfied if and only if $J^{\varepsilon}_{\min} \le \varepsilon$ or $J^{\varepsilon}_{\min} \ge 1-\varepsilon$, with the inequality being saturated if $J^{\varepsilon}_{\min} = \varepsilon$ or $J^{\varepsilon}_{\min} = 1-\varepsilon$. Thus, we conclude that the following inequality holds for all $J^{\varepsilon}_{\min} \in (\varepsilon,1-\varepsilon]$:
\begin{equation}\label{eq:J_eps_sqrt_ineq}
    \sqrt{J^{\varepsilon}_{\min}(1-J^{\varepsilon}_{\min})} \ge \sqrt{\varepsilon(1-\varepsilon)}.
\end{equation}

Now we prove that $\beta_{x-}\leq  \frac{1-\varepsilon-J^{\varepsilon}_{\min}}{J^{\varepsilon}_{\min}-\varepsilon}$, provided that $J^{\varepsilon}_{\min} \in (\varepsilon, 1-\varepsilon]$. Consider that
\begin{align}
    \beta_{x-} &\leq  \frac{1-\varepsilon - J^{\varepsilon}_{\min}}{J^{\varepsilon}_{\min} - \varepsilon}\\
    \Longleftrightarrow\left(\frac{\sqrt{J^{\varepsilon}_{\min}(1-J^{\varepsilon}_{\min})} - \sqrt{\varepsilon(1-\varepsilon)}}{J^{\varepsilon}_{\min}-\varepsilon}\right)^2 &\leq  \frac{1-\varepsilon - J^{\varepsilon}_{\min}}{J^{\varepsilon}_{\min} - \varepsilon}\\
    \Longleftrightarrow\left(\sqrt{J^{\varepsilon}_{\min}(1-J^{\varepsilon}_{\min})} - \sqrt{\varepsilon(1-\varepsilon)}\right)^2 &\leq  \left(J^{\varepsilon}_{\min} - \varepsilon\right)\left(1-\varepsilon - J^{\varepsilon}_{\min}\right)\\
    \Longleftrightarrow\left(\sqrt{J^{\varepsilon}_{\min}(1-J^{\varepsilon}_{\min})} - \sqrt{\varepsilon(1-\varepsilon)}\right)^2 &\leq  J^{\varepsilon}_{\min} - (J^{\varepsilon}_{\min})^2 - \varepsilon + \varepsilon^2\\
    \Longleftrightarrow\left(\sqrt{J^{\varepsilon}_{\min}(1-J^{\varepsilon}_{\min})} - \sqrt{\varepsilon(1-\varepsilon)}\right)^2 &\leq  J^{\varepsilon}_{\min}(1-J^{\varepsilon}_{\min}) - \varepsilon(1-\varepsilon)\\
    \Longleftrightarrow\sqrt{J^{\varepsilon}_{\min}(1-J^{\varepsilon}_{\min})} - \sqrt{\varepsilon(1-\varepsilon)} &\leq  \sqrt{J^{\varepsilon}_{\min}(1-J^{\varepsilon}_{\min})} +\sqrt{\varepsilon(1-\varepsilon)}.\label{eq:beta_le_cond_mid_steps}\\
    \Longleftrightarrow  0 &\leq   2\sqrt{\varepsilon(1-\varepsilon)}.
\end{align}
To arrive at the penultimate inequality, we have used the fact that
\begin{multline}
    J^{\varepsilon}_{\min}(1-J^{\varepsilon}_{\min}) - \varepsilon(1-\varepsilon) \\= \left(\sqrt{J^{\varepsilon}_{\min}(1-J^{\varepsilon}_{\min})} +\sqrt{\varepsilon(1-\varepsilon)}\right)\left(\sqrt{J^{\varepsilon}_{\min}(1-J^{\varepsilon}_{\min})} -\sqrt{\varepsilon(1-\varepsilon)}\right).
\end{multline}
Since the last inequality $0\leq 2\sqrt{\varepsilon(1-\varepsilon)}$ holds trivially for $\varepsilon \in [0,1]$, we conclude that $\beta_{x-} \le \frac{1-\varepsilon-J^{\varepsilon}_{\min}}{J^{\varepsilon}_{\min}-\varepsilon}$ for all $J^{\varepsilon}_{\min} \in (\varepsilon,1-\varepsilon]$.

To show that $\beta_{x+} \ge \frac{1-\varepsilon-J^{\varepsilon}_{\min}}{J^{\varepsilon}_{\min}-\varepsilon}$, we consider the following inequality:
\begin{equation}
    \beta_{x+} = \left(\frac{\sqrt{J^{\varepsilon}_{\min}(1-J^{\varepsilon}_{\min})} + \sqrt{\varepsilon(1-\varepsilon)}}{J^{\varepsilon}_{\min}-\varepsilon}\right)^2 \geq  \frac{1-\varepsilon-J^{\varepsilon}_{\min}}{J^{\varepsilon}_{\min}-\varepsilon}.
\end{equation}
Following the same steps as before, the above inequality can be transformed into the following inequality:
\begin{align}
    \left(\sqrt{J^{\varepsilon}_{\min}(1-J^{\varepsilon}_{\min})} + \sqrt{\varepsilon(1-\varepsilon)}\right)^2 &\geq  J^{\varepsilon}_{\min}(1-J^{\varepsilon}_{\min}) - \varepsilon(1-\varepsilon)\\
    \Longleftrightarrow\sqrt{J^{\varepsilon}_{\min}(1-J^{\varepsilon}_{\min})} + \sqrt{\varepsilon(1-\varepsilon)} &\geq  \sqrt{J^{\varepsilon}_{\min}(1-J^{\varepsilon}_{\min})} -\sqrt{\varepsilon(1-\varepsilon)},
\end{align}
which holds for all $J^{\varepsilon}_{\min}\in (\varepsilon,1-\varepsilon]$ and $\varepsilon \in [0,1]$. Hence, we conclude the following:
\begin{align}
    \left(\frac{\sqrt{J^{\varepsilon}_{\min}(1-J^{\varepsilon}_{\min})} - \sqrt{\varepsilon(1-\varepsilon)}}{J^{\varepsilon}_{\min}-\varepsilon}\right)^2 &\le \frac{1-\varepsilon-J^{\varepsilon}_{\min}}{J^{\varepsilon}_{\min} - \varepsilon}\label{eq:beta_x-_ub} \\ 
    &\le \left(\frac{\sqrt{J^{\varepsilon}_{\min}(1-J^{\varepsilon}_{\min})} + \sqrt{\varepsilon(1-\varepsilon)}}{J^{\varepsilon}_{\min}-\varepsilon}\right)^2\label{eq:beta_x+_lb}
\end{align}

Recall from~\eqref{eq:k_soln_naive} that~\eqref{eq:k_inequality_main_hypo_test} is satisfied for all $J^{\varepsilon}_{\min}\in (\varepsilon,1-\varepsilon]$ if and only if $k^2-1\in\left[0,\frac{1-\varepsilon-J^{\varepsilon}_{\min}}{J^{\varepsilon}_{\min}-\varepsilon}\right]\cup [\beta_{x-},\beta_{x+}]$. The inequalities in~\eqref{eq:beta_x-_ub} and~\eqref{eq:beta_x+_lb} further reveal that the inequality in~\eqref{eq:k_inequality_main_hypo_test} is satisfied for all $J^{\varepsilon}_{\min} \in (\varepsilon,1-\varepsilon]$ if and only if the following condition holds:
\begin{equation}\label{eq:k^2-1_range_final}
    0\le k^2-1 \le \left(\frac{\sqrt{J^{\varepsilon}_{\min}(1-J^{\varepsilon}_{\min})} + \sqrt{\varepsilon(1-\varepsilon)}}{J^{\varepsilon}_{\min}-\varepsilon
    }\right)^2.
\end{equation}

 The quantity $\log_2 k$ is interpreted as the number of secret bits that can be distilled from the state $\rho_{AB}$ using the one-way LOCC channel $\mathcal{L}^{\to}_{AB\to A'B'A''B''}$. Therefore, we rewrite the condition in~\eqref{eq:k^2-1_range_final} as the following upper bound on $\log_2 k$:
\begin{equation}\label{eq:secret_key_ub_hypo_test}
	\log_2 k \le \frac{1}{2}\log_2\!\left[\left(\frac{\sqrt{J^{\varepsilon}_{\min}(1-J^{\varepsilon}_{\min})}+\sqrt{\varepsilon(1-\varepsilon)}}{J^{\varepsilon}_{\min}-\varepsilon}\right)^2+1\right] 
 \end{equation}
Since the inequality in~\eqref{eq:secret_key_ub_hypo_test} holds for every integer $k \ge 2$, every private state $\gamma^k_{A'B'A''B''}$, and every one-way LOCC channels $\mathcal{L}^{\to}_{AB\to A'B'A''B''}$, we conclude~\eqref{eq:dist_key_hypo_test_ub}.

\section{Proof of Lemma~\ref{lem:monotoncity_obj_func}}\label{app:monotonicity_obj_func}

In this section we show that the following function:
\begin{equation}
    f(J,\varepsilon) \coloneqq \frac{1}{2}\log_2\!\left[\left(\frac{\sqrt{J(1-J)}+\sqrt{\varepsilon(1-\varepsilon)}}{J-\varepsilon}\right)^2+1\right]
\end{equation}
decreases monotonically with increasing $J$ and increases monotonically with increasing $\varepsilon$ for all $\varepsilon \in [0,1]$ and $J\in (\varepsilon,1-\varepsilon]$.

First, let us analyze the monotonicity of $f(J,\varepsilon)$ in $J$. The logarithm function is monotonic in its argument. Therefore, we only need to check the monotonicity of the following function:
    \begin{equation}
        g(J,\varepsilon) \coloneqq \left(\frac{\sqrt{J(1-J)}+\sqrt{\varepsilon(1-\varepsilon)}}{J-\varepsilon}\right)^2.
    \end{equation}
    Note that the quantity $\frac{\sqrt{J(1-J)}+\sqrt{\varepsilon(1-\varepsilon)}}{J-\varepsilon}$  is non-negative for all $\varepsilon < J \le 1$ and $0\le \varepsilon \le 1$. Therefore, $g(J,\varepsilon)$ is monotonic in $J$ if $\sqrt{g(J,\varepsilon)}$ is monotonic in $J$ in the given domain.

    Let us compute the derivative of $\sqrt{g(J,\varepsilon)}$ with respect to $J$.
    \begin{equation}\label{eq:sqrt_arg_differential}
        \frac{d}{dJ}\sqrt{g(J,\varepsilon)}
        = \frac{1}{(J-\varepsilon)^2}\left((J-\varepsilon)\!\left(\frac{1}{2}\cdot\frac{1-2J}{\sqrt{J(1-J)}}\right) - \left(\sqrt{J(1-J)}+\sqrt{\varepsilon(1-\varepsilon)}\right)\right).
    \end{equation}
    Set $h \coloneqq J- \varepsilon$ so that $h\in (0,1-2\varepsilon]$. The first term of the above expression can then be expressed as follows:
    \begin{align}
        (J-\varepsilon)\!\left(\frac{1}{2}\cdot\frac{1-2J}{\sqrt{J(1-J)}}\right)
        &= \frac{h(1-2h-2\varepsilon)}{2\sqrt{(h+\varepsilon)(1-h-\varepsilon)}}\\
        &= \frac{h(2-2h-2\varepsilon)}{2\sqrt{(h+\varepsilon)(1-h-\varepsilon)}} - \frac{h}{2\sqrt{(h+\varepsilon)(1-h-\varepsilon)}}\\
        &= \frac{h\sqrt{1-h-\varepsilon}}{\sqrt{h+\varepsilon}} - \frac{h}{2\sqrt{(h+\varepsilon)(1-h-\varepsilon)}}\\
        &\le \sqrt{(h+\varepsilon)(1-h-\varepsilon)} - \frac{h}{2\sqrt{(h+\varepsilon)(1-h-\varepsilon)}}\\
        &= \sqrt{J(1-J)} - \frac{J-\varepsilon}{2\sqrt{J(1-J)}},
    \end{align}
    where we have used the fact that $h \le h+\varepsilon$ to arrive at the inequality. Substituting the above inequality in~\eqref{eq:sqrt_arg_differential},
    \begin{align}
        \frac{d}{dJ}\sqrt{g(J,\varepsilon)}
        &\le \frac{1}{(J-\varepsilon)^2}\left(\sqrt{J(1-J)} - \frac{J-\varepsilon}{2\sqrt{J(1-J)}} - \sqrt{J(1-J)} - \sqrt{\varepsilon(1-\varepsilon)}\right)\\
        &= -\frac{1}{(J-\varepsilon)^2}\left(\frac{J-\varepsilon}{2\sqrt{J(1-J)}} + \sqrt{\varepsilon(1-\varepsilon)}\right),
    \end{align}
    which is negative for all $J\in (\varepsilon,1-\varepsilon]$ and $\varepsilon\in [0,1]$. Therefore, $\sqrt{g(J,\varepsilon)}$ decreases monotonically with $J$ in the given domain, and consequently, $f(J,\varepsilon)$ decreases monotonically with $J$ in the same domain.

    Now let us analyze the monotonicity of $f(J,\varepsilon)$ in $\varepsilon$. Once again, we only need to determine the monotonicity of $\sqrt{g(J,\varepsilon)}$ in $\varepsilon$ to determine the monotonicity of $f(J,\varepsilon)$ in $\varepsilon$. Taking the derivative of $\sqrt{g(J,\varepsilon)}$, we arrive at the following equality:
    \begin{align}
        \frac{d}{d\varepsilon}\sqrt{g(J,\varepsilon)} &= \frac{d}{d\varepsilon} \left(\frac{\sqrt{J(1-J)}+\sqrt{\varepsilon(1-\varepsilon)}}{J-\varepsilon}\right)\\
        &= \frac{1}{(J-\varepsilon)^2}\left((J-\varepsilon)\left(\frac{1}{2}\cdot \frac{1-2\varepsilon}{\sqrt{\varepsilon(1-\varepsilon)}}\right)+\sqrt{J(1-J)} +\sqrt{\varepsilon(1-\varepsilon)}\right)\\
        &= \frac{1}{(J-\varepsilon)^2}\left(\frac{J-\varepsilon}{2\sqrt{\varepsilon(1-\varepsilon)}} - \frac{\varepsilon(J-\varepsilon)}{\sqrt{\varepsilon(1-\varepsilon)}} + \sqrt{J(1-J)} + \sqrt{\varepsilon(1-\varepsilon)}\right)\\
        &=\frac{1}{(J-\varepsilon)^2}\left(\frac{J-\varepsilon}{2\sqrt{\varepsilon(1-\varepsilon)}} + \frac{\varepsilon(1-\varepsilon)-\varepsilon(J-\varepsilon)}{\sqrt{\varepsilon(1-\varepsilon)}} + \sqrt{J(1-J)}\right)\\
        &= \frac{1}{(J-\varepsilon)^2}\left(\frac{J-\varepsilon}{2\sqrt{\varepsilon(1-\varepsilon)}} + \frac{\varepsilon(1-J)}{\sqrt{\varepsilon(1-\varepsilon)}} + \sqrt{J(1-J)}\right).
    \end{align}
    The above expression is strictly positive for all $\varepsilon \in [0,1]$ and all $J \in (\varepsilon,1-\varepsilon]$. Therefore, $\sqrt{g(J,\varepsilon)}$ is monotonically increasing in $\varepsilon$ in the given domain, and consequently, $f(J,\varepsilon)$ is monotonically increasing in $\varepsilon$ in the same domain.

\section{Proof of Proposition~\ref{prop:smooth_min_ch_range}}\label{app:smooth_min_ch_range}

In this section, we show that the smooth-min unextendible entanglement of a channel lies in the following range:
\begin{equation}
    -\frac{1}{2}(1-\varepsilon)\le E^{u,\varepsilon}_{\min}\!\left(\mathcal{N}_{A\to B}\right) \le \log_2d -\frac{1}{2}(1-\varepsilon).
\end{equation}

The lower bound on the smooth-min unextendible entanglement can be obtained from Proposition~\ref{prop:J_eps_range} and Lemma~\ref{lem:unext_ent_state_le_unext_ent_ch_gen}. For every quantum state $\rho_{RA}$, 
\begin{align}
    E^{u,\varepsilon}_{\min}\!\left(\mathcal{N}_{A\to B}\right)&\ge E^{u,\varepsilon}_{\min}\!\left(\mathcal{N}_{A\to B}\!\left(\rho_{RA}\right)\right)\\
    &\ge -\frac{1}{2}\log_2\!\left(1-\varepsilon\right),
\end{align}
where the first inequality follows from Lemma~\ref{lem:unext_ent_state_le_unext_ent_ch_gen} and the second inequality follows from Proposition~\ref{prop:J_eps_range}.

To obtain an upper bound on the smooth-min unextendible entanglement of a channel, consider the smooth-min unextendible entanglement of the identity channel. Recall the inequality in~\eqref{eq:smooth_min_vs_sandwich_ineq}. Setting $\rho \to \operatorname{id}_{A\to B}(\rho_{RA})$, $\sigma \to \mathcal{M}_{A\to E}(\rho_{RA})$, and $\alpha \to \infty$, we arrive at the following inequality:
\begin{equation}
    D^{\varepsilon}_{\min}\!\left(\operatorname{id}_{A\to B}\!\left(\rho_{RA}\right)\middle\Vert\mathcal{M}_{A\to E}\!\left(\rho_{RA}\right)\right) \le D_{\max}\!\left(\operatorname{id}_{A\to B}\!\left(\rho_{RA}\right)\middle\Vert\mathcal{M}_{A\to E}\!\left(\rho_{RA}\right)\right) - \log_2(1-\varepsilon),
\end{equation}
where $\mathcal{M}_{A\to E}$ is an arbitrary quantum channel and systems $B$ and $E$ are isomorphic. Since the above inequality holds for every state $\rho_{RA}$, we can take a supremum over all states and arrive at the following inequality between the smooth-min relative entropy of channels and the max-relative entropy of channels:
\begin{equation}\label{eq:smooth_min_vs_dmax_channels}
    D^{\varepsilon}_{\min}\!\left(\operatorname{id}_{A\to B}\middle\Vert\mathcal{M}_{A\to E}\right) \le D_{\max}\!\left(\operatorname{id}_{A\to B}\middle\Vert\mathcal{M}_{A\to E}\right) - \log_2(1-\varepsilon).
\end{equation}
In~\cite{DFW18, WBHK20}, it was shown that the following equality holds for the max-relative entropy of channels:
\begin{equation}\label{eq:Dmax_saturate_max_ent}
    D_{\max}\!\left(\mathcal{N}_{A\to B}\Vert\mathcal{M}_{A\to B}\right) = D_{\max}\!\left(\mathcal{N}_{A\to B}\!\left(\Phi^d_{RA}\right)\middle\Vert \mathcal{M}_{A\to B}\!\left(\Phi^d_{RA}\right)\right),
\end{equation}
where $d$ is the dimension of the input system of the channel and $\Phi^d_{RA}$ is the maximally entangled state with Schmidt rank $d$. Therefore, we can rewrite~\eqref{eq:smooth_min_vs_dmax_channels} as follows: 
\begin{equation}\label{eq:smooth_min_vs_dmax_Choi}
    D^{\varepsilon}_{\min}\!\left(\operatorname{id}_{A\to B}\middle\Vert\mathcal{M}_{A\to E}\right) \le D_{\max}\!\left(\Phi^d_{RB}\middle\Vert\mathcal{M}_{A\to E}\!\left(\Phi^d_{RA}\right)\right) - \log_2(1-\varepsilon).
\end{equation}

Any channel that lies in the set $\mathcal{F}\!\left(\operatorname{id}_{A\to B}\right)$ is a trace and replace channel; that is, it is of the following form:
\begin{equation}
    \mathcal{M}_{A\to E}\!\left(\cdot\right) = \operatorname{Tr}\!\left[\cdot\right]\sigma_E,
\end{equation}
where $\sigma_E$ is a quantum state. Therefore, the inequality in~\eqref{eq:smooth_min_vs_dmax_Choi} leads to the following inequality:
\begin{align}
    \inf_{\mathcal{M}\in \mathcal{F}(\mathcal{N})}D^{\varepsilon}_{\min}\!\left(\operatorname{id}_{A\to B}\middle\Vert\mathcal{M}_{A\to E}\right) &\le \inf_{\mathcal{M}\in \mathcal{F}(\mathcal{N})}D_{\max}\!\left(\Phi^d_{RB}\middle\Vert\mathcal{M}_{A\to E}\!\left(\Phi^d_{RA}\right)\right) - \log_2(1-\varepsilon)\\
    &= \inf_{\sigma_E\in \mathcal{S}(E)}D_{\max}\!\left(\Phi^d_{RB}\middle\Vert\pi_R\otimes \sigma_E\right) - \log_2(1-\varepsilon),\label{eq:smooth_min_vs_dmax_max_ent}
\end{align}
where $\pi_R$ is the maximally mixed state. Moreover, an arbitrary extension of the maximally entangled state is of the following form $\Phi^d_{AB}\otimes \tau_E$ because the maximally entangled state is a pure state. As such, every state in the set $\mathcal{F}\!\left(\Phi^d_{RB}\right)$ can be written as $\pi_R\otimes \tau_E$ for some $\tau_E \in \mathcal{S}(E)$. Therefore, the max-unextendible entanglement of the maximally entangled state can be written as follows:
\begin{align}
    E^u_{\max}\!\left(\Phi^d_{AB}\right) &= \inf_{\tau_E\in \mathcal{S}(E)}\frac{1}{2}D_{\max}\!\left(\Phi^d_{RB}\middle\Vert\pi_R\otimes \tau_E\right)\\
    &\ge \frac{1}{2}D^{\varepsilon}_{\min}\!\left(\operatorname{id}_{A\to B}\middle\Vert\mathcal{M}_{A\to E}\right) + \frac{1}{2}\log_2(1-\varepsilon)\\
    &= E^{u,\varepsilon}_{\min}\!\left(\operatorname{id}_{A\to B}\right) + \frac{1}{2}\log_2(1-\varepsilon),
\end{align}
where the first inequality follows from~\eqref{eq:smooth_min_vs_dmax_max_ent} and the last equality follows from the definition of smooth-min unextendible entanglement of channels.

The max-unextendible entanglement of a maximally entangled state with Schmidt rank $d$ is equal to $\log_2 d$ as shown in~\cite[Proposition~11]{WWW24}. Therefore, 
\begin{equation}
    E^{u,\varepsilon}_{\min}\!\left(\operatorname{id}_{A\to B}\right) \le \log_2 d - \frac{1}{2}\log_2(1-\varepsilon).
\end{equation}
The converse of the above inequality is also true as Lemma~\ref{lem:unext_ent_state_le_unext_ent_ch_gen} implies the following inequality:
\begin{align}
    E^{\varepsilon,u}_{\min}\!\left(\operatorname{id}_{A\to B}\right) &\ge E^{\varepsilon,u}_{\min}\!\left(\operatorname{id}_{A\to B}\!\left(\Phi^d_{AB}\right)\right)\\
    &= E^{\varepsilon,u}_{\min}\!\left(\Phi^d_{AB}\right)\\
    &= \log_2 d - \frac{1}{2}\log_2(1-\varepsilon),
\end{align}
where the final equality follows from Proposition~\ref{prop:max_ent_unext_ent_hypo_test}. Therefore,
\begin{equation}
    E^{u,\varepsilon}_{\min}\!\left(\operatorname{id}_{A\to B}\right) = \log_2 d - \frac{1}{2}\log_2(1-\varepsilon).
\end{equation}
Since
\begin{equation}
    E^{u,\varepsilon}_{\min}\!\left(\mathcal{N}_{A\to B}\right) \le \min\!\left\{E^{u,\varepsilon}_{\min}\!\left(\operatorname{id}_{A\to C}\right),E^{u,\varepsilon}_{\min}\!\left(\operatorname{id}_{B\to D}\right)\right\},
\end{equation}
where $\operatorname{dim}(A) = \operatorname{dim}(C)$ and $\operatorname{dim}(B) = \operatorname{dim}(D)$, we conclude the second inequality in the statement of the proposition.

\section{Proof of Proposition~\ref{prop:geo_unext_eras_alpha}}\label{app:geo_unext_ent_eras}

In this section we derive an expression for the $\alpha$-geometric unextendible entanglement of the $d$-dimensional erasure channel, for $\alpha \in (0,1)\cup(1,2]$. An upper bound on the $\alpha$-geometric unextendible entanglement of the $d$-dimensional erasure channel was obtained in~\cite[Appendix J]{SW24_channels}. Here we show that the inequality stated in~\cite{SW24_channels} is, in fact, an equality.

Lemma~\ref{lem:unext_ent_state_le_unext_ent_ch_gen} implies the following inequality:
\begin{equation}\label{eq:eras_st_gen_div_le_eras_ch}
    \mathbf{E}^u\!\left(\mathcal{E}^p_{A\to B}\!\left(\Phi^d_{RA}\right)\right) = \mathbf{E}^u\!\left((1-p)~\Phi^d_{AB} + p\frac{I_A}{d}\otimes |e\rangle\!\langle e|_B\right)\le \mathbf{E}^u\!\left(\mathcal{E}^p_{A\to B}\right),
\end{equation}
where $\Phi^d_{RA}$ is the maximally entangled state of Schmidt rank $d \in \mathbb{N}$, $|e\rangle$ is the erasure symbol, and $\mathcal{E}^p_{A\to B}$ is a $d$-dimensional erasure channel with erasure probability $p\in [0,1]$. The bipartite state obtained after sending one share of a maximally entangled state through an erasure channel is called an erased state, which we denote as follows:
\begin{equation}\label{eq:eras_st_defn}
    \eta^p_{AB} \coloneqq (1-p)~\Phi^d_{AB} + p\frac{I_A}{d}\otimes |e\rangle\!\langle e|_B.
\end{equation}
In what follows, we will obtain an analytical expression for the $\alpha$-geometric unextendible entanglement of a $d$-dimensional erased state $\eta^p_{AB}$ and show that it matches the upper bound on the $\alpha$-geometric unextendible entanglement of the erasure channel $\mathcal{E}^p_{A\to B}$ obtained in~\cite[Appendix J]{SW24_channels}. The inequality in~\eqref{eq:eras_st_gen_div_le_eras_ch} then simply implies that the $\alpha$-geometric unextendible entanglement of the erasure channel $\mathcal{E}^p_{A\to B}$ is equal to the $\alpha$-geometric unextendible entanglement of the erased state $\eta^p_{AB}$, thus establishing the equality stated in Proposition~\ref{prop:geo_unext_eras_alpha}.

Let us analyze the generalized unextendible entanglement of an erased state. We will restrict our discussion to $p<\frac{1}{2}$ since $\widehat{E}^u_{\alpha}\!\left(\eta^p_{AB}\right) = 0$ for all $p\ge \frac{1}{2}$.
\begin{lemma}\label{lem:eras_st_opt_marg_gen_div}
    For $p\in [0,1/2)$, let $\eta^p_{AB}$ be the erased state defined in~\eqref{eq:eras_st_defn}. The generalized unextendible entanglement of the erased state is equal to the following:
    \begin{equation}
        \mathbf{E}^u\!\left(\eta^p_{AB}\right) = \inf_{b\in [0,1-p]}\frac{1}{2}\mathbf{D}\!\left(\eta^p_{AB}\middle\Vert\omega^{p,b}_{AE'}\right),
    \end{equation}
    where
    \begin{equation}\label{eq:omega_pb_defn}
        \omega^{p,b}_{AB} = p~\Phi^d_{AB} + b~\frac{I_{A}\otimes \Pi_{B}}{d^2} + (1-p-b)\frac{I_A}{d}\otimes |e\rangle\!\langle e|_{B}
    \end{equation}
    and 
    \begin{equation}
        \Pi \coloneqq |0\rangle\!\langle 0| + \cdots + |d-1\rangle\!\langle d-1|.
    \end{equation}
\end{lemma}
\begin{proof}
    Let $\eta^p_{AB} \coloneqq (1-p)~\Phi^d_{AB} + p\frac{I_A}{d}\otimes |e\rangle\!\langle e|_B$ be an erased state. Consider the following purification of $\eta^p_{AB}$:
    \begin{equation}
        |\psi^{\eta}\rangle_{ABE} \coloneqq \sqrt{1-p}|\Phi^d\rangle_{AB}\otimes|e\rangle_E + \sqrt{p}|\Phi^d\rangle_{AE}\otimes|e\rangle_B,
    \end{equation}
    where systems $B$ and $E$ are isomorphic to each other. For clarity, we define the following projector:
    \begin{equation}
        \Pi\coloneqq |0\rangle\!\langle 0| + \cdots |d-1\rangle\!\langle d-1|,
    \end{equation}
    which is the projector onto the subspace orthogonal to $|e\rangle\!\langle e|$. The maximally mixed state of a $d$-dimensional system is defined as $\pi \coloneqq \frac{\Pi}{d}$. Since system $A$ does not have any component of the erasure symbol, $I_A = \Pi_A$. 
    
    Using the correspondence between a purification and an extension of a state established in~\cite{CW04}, we can write an arbitrary extension of the erased state as $\mathcal{N}_{E\to E'}\!\left(\psi^{\eta}_{ABE}\right)$. Therefore, any state $\sigma_{AE'} \in \mathcal{F}(\eta^p_{AB})$ can be written as follows:
    \begin{align}
        \sigma_{AE'} &= \operatorname{Tr}_{B}\!\left[\mathcal{N}_{E\to E'}\!\left(\psi^{\eta}_{ABE}\right)\right]\\
        &= \mathcal{N}_{E\to E'}\!\left(\operatorname{Tr}_B\!\left[\psi^{\eta}_{ABE}\right]\right)\\
        &= \mathcal{N}_{E\to E'}\!\left((1-p)\pi_A\otimes |e\rangle\!\langle e|_{E} + p~\Phi^d_{AE}\right)\\
        &= (1-p)\pi_A\otimes \mathcal{N}_{E\to E'}\!\left(|e\rangle\!\langle e|_{E}\right) + p~\mathcal{N}_{E\to E'}\!\left(\Phi^d_{AE}\right),
    \end{align}
    where systems $B$, $E$, and $E'$ are all isomorphic to each other.

    Let us consider the following partially dephasing channel:
    \begin{equation}
        \Delta_{E'}\!\left(\cdot\right) = \Pi_{E'}\!\left(\cdot\right)\Pi_{E'} + |e\rangle\!\langle e|_{E'}\!\left(\cdot\right)|e\rangle\!\langle e|_{E'}.
    \end{equation}
     Applying this dephasing channel on $\sigma_{AE'}$ leads to a state of the following form:
    \begin{equation}\label{eq:dephased_state}
        \Delta_{E'}\!\left(\sigma_{AE'}\right) = (1-x)\rho_{AE'} + x\pi_A\otimes |e\rangle\!\langle e|_{E'},
    \end{equation}
    where $x \in [0,1]$ and $\rho_{AE'}$ is a quantum state such that $\rho_{AE'}|e\rangle_{E'} = 0$. 

    Let $U_{A}$ be an arbitrary unitary operator acting on the Hilbert space of system $A$. The corresponding operator acting on the Hilbert space of system $E'$ has the following property:
    \begin{equation}
        U^{\dagger}_{E'}U_{E'} = U_{E'}U^{\dagger}_{E'} = \Pi_{E'}.
    \end{equation}
    We can promote $U_{E'}$ to a unitary operator on the Hilbert space of system $E'$ as follows:
    \begin{equation}
        V^U_{E'} = U_{E'} + |e\rangle\!\langle e|_{E'}.
    \end{equation}
    Note that 
    \begin{equation}\label{eq:V_eq_U}
        V^U_{E'}\rho_{AE'}\left(V^U_{E'}\right)^{\dagger} = U_{E'}\rho_{AE'}U^{\dagger}_{E'},
    \end{equation}
    and
    \begin{equation}
        V^U_{E'}|e\rangle\!\langle e|_{E'}\left(V^U_{E'}\right)^{\dagger} = |e\rangle\!\langle e|_{E'}.
    \end{equation}
    Now consider the following twirling channel:
    \begin{equation}
        \mathcal{T}_{AE'} \coloneqq \int dU\left(\overline{U}_A\otimes V^U_{E'}\right)(\cdot)\left(\overline{U}_A\otimes V^U_{E'}\right)^{\dagger},
    \end{equation}
    which can be implemented by local operations and common randomness (LOCR). The action of this twirling channel on the dephased state in~\eqref{eq:dephased_state} results in the following state:
    \begin{align}
        \mathcal{T}_{AE'}\circ\Delta_{E'}\!\left(\sigma_{AE'}\right) &= (1-x)\mathcal{T}_{AE'}\!\left(\rho_{AE'}\right) + x\int dU U_A\pi_AU^{\dagger}_A\otimes |e\rangle\!\langle e|_{E'}\\
        &= (1-x)\int dU \left(\overline{\mathcal{U}}_A\otimes \mathcal{U}_{E'}\right)(\rho_{AE'}) + x\pi_A\otimes |e\rangle\!\langle e|_{E'}\\
        &= q\Phi^d_{AE'} + (1-q-x)\frac{\Pi_{A}\otimes \Pi_{E'}-\Phi^d_{AE'}}{d^2-1} + x\pi_A\otimes|e\rangle\!\langle e|_{E'}\label{eq:twirled_dephased_st_q_form},
    \end{align}
    where $q\coloneqq (1-x)\operatorname{Tr}\!\left[\rho_{AE'}\Phi^d_{AE'}\right]$. In the above, the second equality follows from~\eqref{eq:V_eq_U}, and the final equality is a consequence of the following equality~\cite{HH99}:
    \begin{equation}
        \int dU \left(\overline{\mathcal{U}}_A\otimes \mathcal{U}_{E'}\right)(\tau_{AE'}) = \operatorname{Tr}\!\left[\tau_{AE'}\Phi^d_{AE'}\right]\Phi^d_{AE'} + \operatorname{Tr}\!\left[\tau_{AE'}\!\left(\Pi_{A}\otimes \Pi_{E'}-\Phi^d_{AE'}\right)\right]\frac{\Pi_{A}\otimes \Pi_{E'}-\Phi^d_{AE'}}{d^2-1},
    \end{equation}
    which holds for every quantum state $\tau_{AE'}$. We can rewrite~\eqref{eq:twirled_dephased_st_q_form} as follows:
    \begin{equation}\label{eq:omega_ab_defn}
        \omega^{a,b}_{AE'}\coloneqq \mathcal{T}_{AE'}\circ\Delta_{E'}\!\left(\sigma_{AE'}\right) = a~\Phi^d_{AE'} + b~\pi_{AE'} + (1-a-b)\pi_A\otimes |e\rangle\!\langle e|_{E'},
    \end{equation}
    where $a = q-\frac{1-q-x}{d^2-1}$ and $b=\frac{(1-q-x)d^2}{d^2-1}$.

    The generalized unextendible entanglement of the erased state can now be written as follows:
    \begin{align}
        \mathbf{E}^u\!\left(\eta^p_{AB}\right) &= \inf_{\sigma\in\mathcal{F}(\eta^p)}\frac{1}{2}\mathbf{D}\!\left(\eta^p_{AB}\middle\Vert\sigma_{AE'}\right)\label{eq:gen_unext_ent_eras_defn}\\
        &\ge \inf_{\sigma\in\mathcal{F}(\eta^p)}\frac{1}{2}\mathbf{D}\!\left(\mathcal{T}_{AB}\circ\Delta_{B}\!\left(\eta^p_{AB}\right)\middle\Vert\mathcal{T}_{AE'}\circ\Delta_{E'}\!\left(\sigma_{AE'}\right)\right) \\
        &= \inf_{\sigma\in\mathcal{F}(\eta^p)}\frac{1}{2}\mathbf{D}\!\left(\eta^p_{AB}\middle\Vert\mathcal{T}_{AE'}\circ\Delta_{E'}\!\left(\sigma_{AE'}\right)\right),\label{eq:gen_div_min_twirl_dephas}
    \end{align}
    where the inequality follows from the data-processing inequality of the generalized divergence and the final equality follows from the fact that the erased state is invariant under the action of the dephasing channel $\Delta_{B}$ as well as the twirling channel $\mathcal{T}_{AB}$. The inequality in~\eqref{eq:gen_div_min_twirl_dephas} implies that for every state $\sigma_{AE'}\in \mathcal{F}\!\left(\eta^p_{AB}\right)$ there exists a state $\omega^{a,b}_{AE'}$, defined in~\eqref{eq:omega_ab_defn}, such that
    \begin{equation}
        \frac{1}{2}\mathbf{D}\!\left(\eta^p_{AB}\middle\Vert\sigma_{AE'}\right) \ge \frac{1}{2}\mathbf{D}\!\left(\eta^p_{AB}\middle\Vert\omega^{a,b}_{AE'}\right).
    \end{equation}
    Therefore, it suffices to restrict the infimum in~\eqref{eq:gen_unext_ent_eras_defn} to the following optimization:
    \begin{equation}\label{eq:gen_unext_ent_eras_eta_omega}
        \mathbf{E}^u\!\left(\eta^p_{AB}\right) = \inf_{a\in\mathcal{A},b\in \mathcal{B}}\frac{1}{2}\mathbf{D}\!\left(\eta^p_{AB}\middle\Vert\omega^{a,b}_{AE'}\right),
    \end{equation}
    where sets $\mathcal{A}$ and $\mathcal{B}$ correspond to the sets of parameters $a$ and $b$ such that $\omega^{a,b}_{AE'}\in \mathcal{F}\!\left(\eta^p_{AB}\right)$.

    Now let us find the range of values that $a$ and $b$ can take such that $\omega^{a,b}_{AE'}$, defined in~\eqref{eq:omega_ab_defn}, lies in the set $\mathcal{F}\!\left(\eta^p_{AB}\right)$. Note that
    \begin{equation}
        \omega^{a,b}_{AE'} = \mathcal{T}_{AE'}\circ\Delta_{E'}\circ\mathcal{N}_{E\to E'}\!\left(p~\Phi^d_{AE} + (1-p)\pi_A\otimes |e\rangle\!\langle e|_E\right).
    \end{equation}
    For every $\mathcal{N}_{E\to E'}$, the channel $\mathcal{T}_{AE'}\circ\Delta_{E'}\circ\mathcal{N}_{E\to E'}$ acts on $\pi_{A}\otimes|e\rangle\!\langle e|_E$ and $\Phi^d_{AE}$ as follows:
    \begin{align}
        \mathcal{T}_{AE'}\circ\Delta_{E'}\circ\mathcal{N}_{E\to E'}\!\left(\pi_A\otimes |e\rangle\!\langle e|_E\right) &= y\pi_{AE'} + (1-y)\pi_A\otimes |e\rangle\!\langle e|_{E'},\\
        \mathcal{T}_{AE'}\circ\Delta_{E'}\circ\mathcal{N}_{E\to E'}\!\left(\Phi^d_{AE}\right) &= y'\Phi^d_{AE'} + y''\pi_{AE'} + (1-y'-y'')\pi_A\otimes|e\rangle\!\langle e|_{E'},
    \end{align}
    where $y\in [0,1]$, $y' \in [0,1]$ and $y'' \in [0,1-y']$.    
    The state $\omega^{a,b}_{AE'}$ can hence be written as follows:
    \begin{multline}
        \omega^{a,b}_{AE'} = py'~\Phi^d_{AE} + (y(1-p)+py'')\pi_{AE'} \\+ ((1-y)(1-p)+p(1-y'-y''))\pi_{A}\otimes |e\rangle\!\langle e|_{E'},
    \end{multline}
    for some $y\in [0,1]$, $y'\in [0,1]$ and $y''\in [0,1-y']$. Comparing with~\eqref{eq:omega_ab_defn}, it is clear that $a \le p$. Therefore, if $\omega^{a,b}_{AE'}\in \mathcal{F}\!\left(\eta^p_{AB}\right)$, then $a$ must be less than or equal to $p$.

    Now we will show that $\omega^{a,b}_{AE'}\in \mathcal{F}\!\left(\eta^p_{AB}\right)$ for all $a\in [0,p]$ and $b\in [0,1-a]$. Consider the following extension of the state $\omega^{a,b}_{AE'}$:
    \begin{equation}
        \omega^{a,b,c,g}_{ABE'} = \left(a~\Phi^d_{AE'} + c~\pi_{AE'}\right)\otimes|e\rangle\!\langle e|_B + \left(g~\Phi^d_{AB}+ f~\pi_{AB}\right)\otimes|e\rangle\!\langle e|_{E'} + (b-c)\Phi^d_{AB}\otimes\pi_{E'},
    \end{equation}
    where $g+f = 1-a-b$ and $a,b \ge 0$ follow from~\eqref{eq:omega_ab_defn} and $c,g,f\geq 0$ and $c\leq b$ ensure that $\omega^{a,b,c,g}_{ABE'}$ is positive semi-definite. It can be easily verified that $\operatorname{Tr}_{B}\!\left[\omega^{a,b,c,g}_{ABE'}\right] = \omega^{a,b}_{AE'}$. Moreover, if $a+c = p$ and $f=0$, then $\operatorname{Tr}_{E'}\!\left[\omega^{a,b,c,g}_{ABE'}\right] = \eta^p_{AB}$ (and also $g+b-c = 1-p$ as a consequence). Therefore, for all $a \in [0,p]$, there exists $b\in [0,1-a]$ such that $\omega^{a,b}_{AE'}\in \mathcal{F}\!\left(\eta^p_{AB}\right)$. As such, $\omega^{a,b}_{AE'}\in \mathcal{F}\!\left(\eta^p_{AB}\right)$ if and only if $a \in [0,p]$ and $b\in [0,1-a]$. Invoking~\eqref{eq:gen_unext_ent_eras_eta_omega}, we can write the generalized unextendible entanglement of an erased state $\eta^p_{AB}$ as follows:
    \begin{equation}\label{eq:gen_unext_ent_ab_inf}
        \mathbf{E}^u\!\left(\eta^p_{AB}\right) = \inf_{\substack{a\in[0,p],\\b\in [0,1-a]}}\frac{1}{2}\mathbf{D}\!\left(\eta^p_{AB}\middle\Vert\omega^{a,b}_{AE'}\right).
    \end{equation}

Consider the following channel:
\begin{equation}
    \mathcal{R}^s_{AE'}\!\left(\cdot\right) = |e\rangle\!\langle e|_{E'}\!\left(\cdot\right)|e\rangle\!\langle e|_{E'} + (1-s)\operatorname{Tr}\!\left[\Pi_{E'}\!\left(\cdot\right)\Pi_{E'}\right]\Phi^d_{AE'} + s~\operatorname{id}_{AE'}\!\left(\Pi_{E'}\!\left(\cdot\right)\Pi_{E'}\right),
\end{equation}
where $s\in [0,1]$. The channel $\mathcal{R}^s_{AE'}$ can be realized by applying the POVM $\left\{\Pi_{E'},|e\rangle\!\langle e|_{E'}\right\}$ on the state of system $E'$. If the outcome corresponding to the POVM element $\Pi_{E'}$ occurs, then the state is replaced by the maximally entangled state $\Phi^d_{AE'}$ with probability $1-s$ and otherwise, with probability $s$, the identity channel is applied. The erased state $\eta^p_{AE'}$ is invariant under the action of the channel $\mathcal{R}^s_{AE'}$ for all $s \in [0,1]$. The channel $\mathcal{R}^s_{AE'}$ acts on $\omega^{a,b}_{AE'}$ as follows:
\begin{align}
    \mathcal{R}^s_{AE'}\!\left(\omega^{a,b}_{AE'}\right) &= ((1-s)b+a)\Phi^d_{AE'} + sb~\pi_{AE'} + (1-a-b)\pi_{A}\otimes|e\rangle\!\langle e|_{E'}\\
    &= \omega^{a+(1-s)b, sb}_{AE'}.
\end{align}
Fix $s = (a+b-p)/b$. The data-processing inequality of the generalized divergence yields the following inequality:
\begin{align}
    \mathbf{D}\!\left(\eta^p_{AB}\middle\Vert\omega^{a,b}_{AE'}\right) &\ge \mathbf{D}\!\left(\mathcal{R}^s_{AB}\!\left(\eta^p_{AB}\right)\middle\Vert\mathcal{R}^s_{AE'}\!\left(\omega^{a,b}_{AE'}\right)\right)\\
    &= \mathbf{D}\!\left(\eta^p_{AB}\middle\Vert\omega^{a+(1-s)b,sb}_{AE'}\right)\\
    &= \mathbf{D}\!\left(\eta^p_{AB}\middle\Vert\omega^{p,a+b-p}_{AE'}\right).
\end{align}
If $\omega^{a,b}_{AE'}\in \mathcal{F}\!\left(\eta^p_{AB}\right)$, then $\omega^{p,a+b-p}_{AE'}$ also lies in set $\mathcal{F}\!\left(\eta^p_{AB}\right)$. Therefore, the bivariate infimum in~\eqref{eq:gen_unext_ent_ab_inf} can be restricted to a single variable infimum as follows:
\begin{equation}
    \mathbf{E}^u\!\left(\eta^p_{AB}\right) = \inf_{b\in [0,1-p]}\frac{1}{2}\mathbf{D}\!\left(\eta^p_{AB}\middle\Vert\omega^{p,b}_{AE'}\right).
\end{equation}
This concludes the proof.
\end{proof}

\medskip

Lemma~\ref{lem:eras_st_opt_marg_gen_div} allows us to obtain an analytical expression for the $\alpha$-geometric unextendible entanglement of the erased state, which we present in Proposition~\ref{prop:geo_unext_eras_st_alpha} stated below. 

\begin{proposition}\label{prop:geo_unext_eras_st_alpha}
    For all $\alpha \in (0,1)\cup(1,2]$, the $\alpha$-geometric unextendible entanglement of a $d$-dimensional erased state $\eta^p_{AB}$, defined in Lemma~\ref{lem:eras_st_opt_marg_gen_div}, evaluates to the following:
    \begin{equation}
        \widehat{E}^u_{\alpha}\!\left(\eta^p_{AB}\right) = \frac{1}{2}\cdot\frac{1}{\alpha - 1}\log_2\!\left(\left(p + \frac{b_{\operatorname{opt}}}{d^2}\right)^{1-\alpha}(1-p)^{\alpha} + (1-p-b_{\operatorname{opt}})^{1-\alpha}p^{\alpha}\right) 
    \end{equation}
    for all $p\in \left(0,\frac{1}{d^{1/\alpha}+1}\right]$, where 
    \begin{equation}
        b_{\operatorname{opt}} \coloneqq \frac{d^2\!\left((1-p)^2 - p^2d^{2/\alpha}\right)}{pd^{2/\alpha} + (1-p)d^2}.
    \end{equation} For all $p\in \left(\frac{1}{d^{1/\alpha}+1},\frac{1}{2}\right]$,
    \begin{equation}
        \widehat{E}^u_{\alpha}\!\left(\eta^p_{AB}\right) = \frac{1}{2}\cdot\frac{1}{\alpha - 1}\log_2\!\left(p^{1-\alpha}(1-p)^{\alpha} + (1-p)^{1-\alpha}p^{\alpha}\right).
    \end{equation}
\end{proposition}
\begin{proof}
    The $\alpha$-geometric unextendible entanglement of the erased state can be computed using Lemma~\ref{lem:eras_st_opt_marg_gen_div} as follows:
    \begin{equation}
        \widehat{E}^u_{\alpha}\!\left(\eta^p_{AB}\right) = \inf_{b\in [0,1-p]}\frac{1}{2}\widehat{D}_{\alpha}\!\left(\eta^p_{AB}\middle\Vert\omega^{p,b}_{AB}\right),
    \end{equation}
    where $\omega^{p,b}_{AB}$ is defined in~\eqref{eq:omega_pb_defn} and $\alpha \in (0,1)\cup(1,2]$. Recall the definition of the $\alpha$-geometric R\'enyi relative entropy given in~\eqref{eq:alpha_geo_rel_ent_defn}. The $\alpha$-geometric R\'enyi relative entropy of $\eta^p_{AB}$ with respect to $\omega^{p,b}_{AB}$ can be computed as follows:
    \begin{equation}\label{eq:geo_rel_ent_eras_vs_omega}
        \widehat{D}_{\alpha}\!\left(\eta^p_{AB}\middle\Vert\omega^{p,b}_{AB}\right) = \frac{1}{\alpha - 1}\log_2\!\left(\left(p + \frac{b}{d^2}\right)^{1-\alpha}(1-p)^{\alpha} + (1-p-b)^{1-\alpha}p^{\alpha}\right).
    \end{equation}
    The above expression is minimized for 
    \begin{equation}\label{eq:b_opt_defn}
        b = b_{\operatorname{opt}} \coloneqq \frac{d^2\!\left((1-p)^2 - p^2d^{2/\alpha}\right)}{pd^{2/\alpha} + (1-p)d^2}.
    \end{equation}
    
    Let us now find the range of $p$ such that $b_{\operatorname{opt}} \in [0,1-p]$. Consider the following expression:
    \begin{align}
        1-p-b_{\operatorname{opt}} = \frac{pd^{2/\alpha}(1-p+pd^2)}{pd^{2/\alpha} + (1-p)d^2}.
    \end{align}
    The above expression is greater than or equal to zero for all $p \in [0,1]$. Therefore, $b_{\operatorname{opt}} \le 1-p$ for all $p \in [0,1]$. Moreover, $b_{\operatorname{opt}} \ge 0$ if and only if $0 \le (1-p)^2-p^2d^{2/\alpha}$, which holds for all $p \in \left[0,\frac{1}{d^{1/\alpha}+1}\right)$. If $b_{\operatorname{opt}} \le 0$, the value of $b \in [0,1-p]$ that minimizes the expression in~\eqref{eq:geo_rel_ent_eras_vs_omega} is $b = 0$. Therefore, for all $\alpha \in (0,1)\cup(1,2]$, the $\alpha$-geometric unextendible entanglement of a $d$-dimensional erased state $\eta^p_{AB}$ evaluates to the following:
    \begin{equation}
        \widehat{E}^u_{\alpha}\!\left(\eta^p_{AB}\right) = \frac{1}{2}\cdot\frac{1}{\alpha - 1}\log_2\!\left(\left(p + \frac{b_{\operatorname{opt}}}{d^2}\right)^{1-\alpha}(1-p)^{\alpha} + (1-p-b_{\operatorname{opt}})^{1-\alpha}p^{\alpha}\right) 
    \end{equation}
    for all $p\in \left(0,\frac{1}{d^{1/\alpha}+1}\right]$, where $b_{\operatorname{opt}}$ is defined in~\eqref{eq:b_opt_defn}. For all $p\in \left(\frac{1}{d^{1/\alpha}+1},\frac{1}{2}\right]$,
    \begin{equation}
        \widehat{E}^u_{\alpha}\!\left(\eta^p_{AB}\right) = \frac{1}{2}\cdot\frac{1}{\alpha - 1}\log_2\!\left(p^{1-\alpha}(1-p)^{\alpha} + (1-p)^{1-\alpha}p^{\alpha}\right).
    \end{equation}
    This concludes the proof.
\end{proof}

\medskip

\begin{proof}[Proof of Proposition~\ref{prop:geo_unext_eras_alpha}]
The $\alpha$-geometric unextendible entanglement of the erased state serves as a lower bound on the $\alpha$-geometric unextendible entanglement of the erasure channel, as is evident from~\eqref{eq:eras_st_gen_div_le_eras_ch}. The expression for the $\alpha$-geometric unextendible entanglement of the erased state $\eta^p_{AB}$ derived in Proposition~\ref{prop:geo_unext_eras_st_alpha} was found to be an upper bound on the $\alpha$-geometric unextendible entanglement of a $d$-dimensional erasure channel with erasure probability $p$. We give a brief outline of the proof here.

Consider a quantum channel $\mathcal{P}^{b^*}_{A\to BE}$ with the following Choi operator:
\begin{equation}
    \Gamma^{\mathcal{P},b^*}_{ABE} = pd~\Phi^d_{AE}\otimes|e\rangle\!\langle e|_{B} + \left(1-p-b^*\right)d~\Phi^d_{AB}\otimes|e\rangle\!\langle e|_{E} + b^*d~\Phi^d_{AB}\otimes \pi_{E}. 
\end{equation}
The two relevant marginals of this Choi operator are as follows:
\begin{equation}
    \Gamma^{\mathcal{E}}_{AB} = d\!\left((1-p)\Phi^d_{AB} + p~\pi_A\otimes |e\rangle\!\langle e|_B\right),
\end{equation}
which is the Choi operator of a $d$-dimensional erasure channel with erasure probability $p$, and
\begin{equation}
    \Gamma^{\mathcal{M},b^*}_{AE} = d\!\left(p~\Phi^d_{AE} + (1-p-b^*)\pi_A\otimes |e\rangle\!\langle e|_E + b^*\pi_{AE}\right).
\end{equation}
The Choi operator $\Gamma^{\mathcal{M},b^*}_{AE}$ corresponds to a channel $\mathcal{M}^{b^*}_{A\to E}$ that lies in the set $\mathcal{F}\!\left(\mathcal{E}^p_{A\to B}\right)$ if $b^* \in [0,1-p]$, where $\mathcal{E}^p_{A\to B}$ is a $d$-dimensional erasure channel with erasure probability $p$. By definition,
\begin{equation}
    \widehat{E}^u_{\alpha}\!\left(\mathcal{E}^p_{A\to B}\right) \le \frac{1}{2}\widehat{D}_{\alpha}\!\left(\mathcal{E}^p_{A\to B}\middle\Vert\mathcal{M}^{b^*}_{A\to E}\right)
\end{equation}
for all $b^* \in [0,1-p]$. Choosing $b^* = \min\{0,b_{\operatorname{opt}}\}$, where $b_{\operatorname{opt}}$ is defined in~\eqref{eq:b_opt_defn}, the $\alpha$-geometric R\'enyi relative entropy of $\mathcal{E}^p_{A\to B}$ with respect to $\mathcal{M}^{b^*}_{A\to E}$ evaluates to the expression which is equal to the $\alpha$-geometric unextendible entanglement of the erased state $\eta^p_{AB}$ derived in Proposition~\ref{prop:geo_unext_eras_st_alpha}. Since the $\alpha$-geometric unextendible entanglement of the erasure channel $\mathcal{E}^p_{A\to B}$ cannot be less than the $\alpha$-geometric unextendible entanglement of the erased state $\eta^p_{AB}$, we conclude that the two quantities are equal.
\end{proof}

\section{Semidefinite programs}\label{app:semidefinite_programs}

In this section we present all the semidefinite programs that were used in this work.

\begin{enumerate}
    \item \textbf{Smooth-min unextendible entanglement of a state:}
    \begin{equation}
        E^{u,\varepsilon}_{\min}(\rho_{AB}) = -\frac{1}{2}\log_2\max\left\{
        \begin{array}{c}
             \mu(1-\varepsilon)-\operatorname{Tr}\!\left[Z_{AB}\right]: \\
             \mu \ge 0, Z_{AB} \ge 0, \omega_{ABE} \ge 0,\\
             \mu\rho_{AB} \le \operatorname{Tr}_B\!\left[\omega_{ABE}\right] + Z_{AB},\\
             \operatorname{Tr}_E\!\left[\omega_{ABE}\right] = \rho_{AB}
        \end{array}
        \right\}.
    \end{equation}
    \item \textbf{Max-unextendible entanglement of a state:} The semidefinite program for the max-unextendible entanglement of a state was given in~\cite{WWW24}. We include it here for completeness.
    \begin{equation}\label{eq:max_unext_ent_st_SDP}
        E^u_{\max}\!\left(\rho_{AB}\right) = -\frac{1}{2}\log_2 \max\left\{
        \begin{array}{c}
             \lambda:  \\
             \lambda \rho_{AB} \le \operatorname{Tr}_{B}\!\left[\omega_{ABE}\right],\\
             \omega_{ABE} \ge 0,\\
             \operatorname{Tr}_E\!\left[\omega_{ABE}\right] = \rho_{AB}
        \end{array}
        \right\}.
    \end{equation}
    \item \textbf{Smooth-min unextendible entanglement of a channel:} The smooth-min relative entropy of a channel $\mathcal{N}$ with respect to a channel $\mathcal{M}$ has a semidefinite program, which was given in~\cite[Proposition~2]{WW19}. We use it to write the semidefinite program for smooth-min unextendible entanglement of a channel as follows:
    \begin{equation}
        E^{u,\varepsilon}_{\min}\!\left(\mathcal{N}_{A\to B}\right) = -\frac{1}{2}\log_2 \max\left\{
        \begin{array}{c}
            \mu(1-\varepsilon) - \lambda: \\
             \lambda \ge 0, \mu \ge 0, Y_{AB}\ge 0,\Gamma^{\mathcal{P}}_{ABE}\ge 0,\\
             \mu\Gamma^{\mathcal{N}}_{AB} \le \operatorname{Tr}_{B}\!\left[\Gamma^{\mathcal{P}}_{AB}\right]+Y_{AB}\\
             \operatorname{Tr}_{B}\!\left[Y_{AB}\right] \le \lambda I_A,\\
             \operatorname{Tr}_E\!\left[\Gamma^{\mathcal{P}}_{ABE}\right] = \Gamma^{\mathcal{N}}_{AB}
        \end{array}
        \right\},
    \end{equation}
    where $\Gamma^{\mathcal{N}}_{AB}$ is the Choi operator of the channel $\mathcal{N}_{A\to B}$ defined in~\eqref{eq:Choi_op_defn}.
        \item \textbf{$\alpha$-geometric unextendible entanglement of a channel:} Fix $\ell \in \mathbb{N}$. The $\alpha$-geometric unextendible entanglement of a channel $\mathcal{N}_{A\to B}$ for $\alpha = 1+2^{-\ell}$ can be computed using the following semidefinite program:
    \begin{equation}
		\widehat{E}^u_{\alpha}\!\left(\mathcal{N}_{A\to B}\right) = 2^\ell \min_{\substack{y\in \mathbb{R}, \Gamma^{\mathcal{P}}_{ABE} \ge 0\\M_{AE}, \left\{N^i_{AE}\right\}_{i=0}^{\ell},\in \operatorname{Herm}}} \log_2 y,
	\end{equation}
	subject to the constraints,
	\begin{align}
		\operatorname{Tr}_{E}\!\left[\Gamma^{\mathcal{P}}_{ABE}\right] &= \Gamma^{\mathcal{N}}_{AB}\label{eq:non_sig_SDP_geo_unext_ent},\\
		\operatorname{Tr}_{E} \left[M_{AE}\right] &\le yI_A\label{eq:geo_unext_ent_cond_1},\\
		\operatorname{Tr}_{B} \left[\Gamma^{\mathcal{P}}_{ABE}\right] &= N^0_{AE}, \label{eq:geo_unext_ent_cond_2}\\
		\begin{bmatrix}
			M_{AE} & \Gamma^{\mathcal{N}}_{AB}\\
			\Gamma^{\mathcal{N}}_{AB} & N^{\ell}_{AE}
		\end{bmatrix}
		&\ge 0,\label{eq:geo_unext_ent_cond_3}\\
		\begin{bmatrix}
			\Gamma^{\mathcal{N}}_{AE} & N^i_{AE}\\
			N^i_{AE} & N^{i-1}_{AE}
		\end{bmatrix}
		&\ge 0 \quad \forall i\in \{1,2,\ldots,\ell\},\label{eq:geo_unext_ent_cond_4}
	\end{align}
 where $\Gamma^{\mathcal{N}}_{AB}$ is the Choi operator of the channel $\mathcal{N}_{A\to B}$ and the system $E$ is isomorphic to the system $B$. To compute the $\alpha$-geometric unextendible entanglement of the channel for other rational values of $\alpha$ see~\cite[Table 4]{FS17}.
\end{enumerate}
\end{document}